\documentclass[3p]{elsarticle}
\usepackage{amssymb,amsmath}

\usepackage{amsthm}
\usepackage[font=small,center]{caption}
\usepackage{subcaption}
\usepackage{xcolor}
\usepackage{graphicx}
\usepackage{wrapfig}
\usepackage{makecell}
\usepackage{bussproofs}
\usepackage{tabularx}
\usepackage{multirow}
\usepackage{makecell}
\usepackage{rotating}

\makeatletter
\RequirePackage[bookmarks,unicode,colorlinks=true]{hyperref}%
   \def\@citecolor{blue}%
   \def\@urlcolor{blue}%
   \def\@linkcolor{blue}%

\def\orcidID#1{\smash{\href{http://orcid.org/#1}{\protect\raisebox{-1.25pt}{\protect\includegraphics{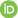}}}}}
\makeatother

\usepackage[ruled,linesnumbered]{algorithm2e}
\SetKwComment{Comment}{/* }{ */}

\definecolor{hibou_col_lf}{RGB}{22, 22, 130}
\newcommand{\hlf}[1]{\textcolor{hibou_col_lf}{#1}}
\definecolor{hibou_col_ms}{RGB}{15, 86, 15}
\newcommand{\hms}[1]{\textcolor{hibou_col_ms}{#1}}
\definecolor{darkspringgreen}{rgb}{0.09, 0.45, 0.27}
\newcommand{\shortColGreen}[1]{\textcolor{darkspringgreen}{#1}}
\newcommand{\shortColRed}[1]{\textcolor{red}{#1}}
\newcommand{\shortColBlue}[1]{\textcolor{blue}{#1}}
\newcommand{\shortColOrange}[1]{\textcolor{orange}{#1}}
\newcommand{\shortColViolet}[1]{\textcolor{violet}{#1}}
\newcommand{\shortColCyan}[1]{\textcolor{cyan}{#1}}

\newcommand{\shortColBrown}[1]{\textcolor{brown}{#1}}

\definecolor{my_grey}{RGB}{160, 160, 160}

\usepackage{scalerel}
\usepackage{centernot}

\DeclareRobustCommand\doubleVerticalTimesDefault{%
  \leavevmode
  {\sbox0{\ddag}%
   \ooalign{\raisebox{\ht0-\height}{$\times$}\cr
            \raisebox{\depth-\dp0}{\scalebox{1}[-1]{$\times$}}\cr}%
  }%
}

\makeatletter
\providecommand{\leftsquigarrow}{%
  \mathrel{\mathpalette\reflect@squig\relax}%
}
\newcommand{\reflect@squig}[2]{%
  \reflectbox{$\m@th#1\rightsquigarrow$}%
}
\makeatother

\newcommand{\doubleVerticalTimes}{\scalerel*{\doubleVerticalTimesDefault}{b}}

\newcommand{\lifelineElim}{\mathsf{rmv}}
\newcommand{\muProjection}{\mathsf{\pi}}

\newcommand{\conflictPredicate}{\,\doubleVerticalTimes\,}

\newcommand{\multiAppend}{\,\symbol{94}\,}

\newcommand{\opStrictSeq}{\,;\,}
\newcommand{\opWeakSeq}{\,;_{\doubleVerticalTimes}\,}
\newcommand{\opInterleaving}{\,||\,}

\newcommand{\macroOKVerdict}{Ok}
\newcommand{\macroKOVerdict}{Nok}

\newcommand{\graphRelationExec}{\,\shortColOrange{\rightsquigarrow_e}\,}
\newcommand{\graphRelationPORExec}{\,\shortColOrange{\rightsquigarrow_e^{\mathbf{POR}}}\,}
\newcommand{\graphRelationHide}{\,\shortColViolet{\rightsquigarrow_r}\,}
\newcommand{\graphRelationPass}{\,\shortColBlue{\rightsquigarrow_o}\,}

\newcommand{\squareop}[1]{%
    \setlength{\fboxsep}{0pt}%
    \setlength{\unitlength}{.7em}%
    \mathrel{%
        \raisebox{-1pt}{\framebox(1.1,1){\(\scriptstyle #1\)}}%
    }%
}

\newcommand{\oneUnambiguous}{\,\squareop{1\!\raisebox{0.085pt}{$\scriptstyle\in$}}\,}

\newtheorem{definition}{Definition}

\newtheorem{theorem}{Theorem}
\newtheorem{property}{Property}
\newtheorem*{property*}{Property}
\newtheorem*{lemma*}{Lemma}
\newtheorem*{theorem*}{Theorem}

\newenvironment{scprooftree}[1]%
  {\gdef\scalefactor{#1}\begin{center}\proofSkipAmount \leavevmode}%
  {\scalebox{\scalefactor}{\DisplayProof}\proofSkipAmount \end{center} }

\newcommand{\collidesLf}{\centernot{\downarrow}^{\doubleVerticalTimes}}

\newcommand{\isPruneBase}{\simeq^{\doubleVerticalTimes}}

\newcommand{\isPruneOf}[1]{\isPruneBase_{#1}}

\newcommand{\anaLocal}{\omega}

\newcommand{\depthAnaLoc}{\delta}

\usepackage{tikz}
\usetikzlibrary{positioning,arrows,shapes,hobby,backgrounds,calc,trees,fit,shapes.multipart,patterns}
\usepackage{tikz-cd}
\pgfdeclarelayer{background}
\pgfsetlayers{background,main}

\journal{Science of Computer Programming}

\begin{document}

\begin{frontmatter}

\title{Efficient Interaction-Based Offline Runtime Verification of Distributed Systems with Lifeline Removal}

\author[1]{
Erwan Mahe
\orcidID{0000-0002-5322-4337}
}
\author[1]{
Boutheina Bannour
\orcidID{0000-0002-4943-7807}
}
\author[1]{
Christophe Gaston
\orcidID{0000-0001-6865-5108}
}
\author[2]{
Pascale Le Gall
\orcidID{0000-0002-8955-6835} 
}

\affiliation[1]{organization={Université Paris-Saclay, CEA, List},
            addressline={F-91120}, 
            city={Palaiseau},
            country={France}}

\affiliation[2]{organization={MICS, CentraleSupélec, Université Paris-Saclay}, 
            addressline={F-91192}, 
            city={Gif-sur-Yvette},
            country={France}}

\begin{abstract}

Runtime Verification (RV) refers to a family of techniques in which system executions are observed and confronted to formal specifications, with the aim of identifying faults.
In Offline RV, observation is done in a first step and verification in a second step, on a static artifact collected during the observation.
In this paper, we define an approach to offline RV of Distributed Systems (DS) against interactions. Interactions are formal models describing communications within a DS. DS are composed of subsystems deployed on different machines and interacting via message passing to achieve common goals. 
Therefore, observing executions of a DS entails logging a collection of local execution traces, one for each subsystem, collected on its host machine.
We call {\em multi-trace} such observational artifacts.
A major challenge in analyzing multi-traces is that there are no practical means to synchronize the ends of observations of all the local traces. 
We address this via an operation, called lifeline removal, which we apply on-the-fly on the specification during the verification of a multi-trace once a local trace has been entirely analyzed.
This operation removes from the interaction the specification of actions occurring on the subsystem that is no-longer observed.
This may allow further execution of the specification via removing deadlocks due to the partial orders of actions.
We prove the correctness of the resulting RV algorithm and introduce two optimization techniques which we also prove correct.
We implement a Partial Order Reduction (POR) technique via the selection of a one-unambiguous action (as a unique first step to a linearization) which existence is determined via another use of the lifeline removal operator.
Additionally, Local Analyses (LOC) i.e., the verification of local traces, can be leveraged during the global multi-trace analysis to prove failure more quickly.
Experiments illustrate the application of our RV approach and the benefits of our optimizations.
\end{abstract}


\end{frontmatter}

\section{Introduction\label{sec:introduction}}

Distributed Systems (DS) are software or cyber-physical systems composed of multiple distinct and potentially distant subsystems that interact with each other through peer-to-peer message exchanges.
The behavioral specification of such DS can be formalized either via a collection of local models combined with a communication policy, or via a global model that includes both.
Interactions are such global models that have the advantage of having a direct and intuitive graphical representation in the form of Sequence Diagrams (SD) akin to UML-SD \cite{UML} or Message Sequence Charts (MSC) \cite{MSC}.

Interactions are natural candidates to be used as reference models in Runtime Verification (RV) for DS. 
RV refers to a family of techniques that consist in observing executions of systems and confronting them to formal specifications or models, with the intent of identifying faults.
In most approaches to RV for DS, the formal references against which system executions are analyzed are specified using formalisms or logics usually equipped with global trace semantics.
A global trace totally orders events that occur in the DS, whichever are the subsystems on which they occur. However, because DS are composed of subsystems deployed on different computers and communicating via message passing, their executions are more naturally represented as collections of local traces observed at the level of the different subsystems' interfaces  \cite{constraint_based_oracles_for_timed_distributed_systems,passive_conformance_testing_of_service_choreographies}. 
Because subsystems may not share a common clock \cite{Lamport19b}, reordering the events that constitute the local traces as a unique global trace is not readily possible. As a result, the entry data to the RV process for DS, in all generality, takes the form of a structured collection of local traces, referred to as a multi-trace 
\cite{constraint_based_oracles_for_timed_distributed_systems,a_small_step_approach_to_multi_trace_checking_against_interactions}.

However, under the hypothesis of the absence of a global clock, as is typical in the context of DS, it is not feasible to synchronize the endings of the different local observations. 
As a result, we have a form of partial observation, having no guarantees that local observation has ceased at the right time (i.e., not too early) on every subsystem.
In a broader context, it is also possible that some subsystems are not observed at all. Such issues may occur due to technical or legal reasons, such as missing monitoring devices, malfunctions, or synchronization issues.

For these reasons, DS have been identified in a recent survey \cite{surveyRV_SanchezSABBCFFK19} as one of the most challenging application domains for RV.
A significant challenge arises from the fact that the reference formalisms or logics used for analyzing system executions adopt global trace semantics.
Various approaches to address DS RV involve identifying global traces resulting from all possible temporal orderings of events in local traces. If none of these global traces conform to the formal reference, an error may be inferred \cite{passive_conformance_testing_of_service_choreographies}. 
In \cite{a_small_step_approach_to_multi_trace_checking_against_interactions}, we have taken advantage of operational semantics of interactions to assess the conformance of a multi-trace w.r.t.~a given interaction directly: events are reorganized into a global trace at the same time as they are compared with the model. 

Yet, these approaches cannot handle partial observation because reconstructing a global trace might be impaired by having some events missing from specific local traces. 
For instance, it is possible that the reception of a message is observed via local observation on a given subsystem, while its corresponding emission was not observed because the local observation on the emitting subsystem ceased too early. 
Hence, although locally, every partial local observation corresponds to a prefix of an ideal full local observation, globally, the partial observation does not correspond to observing a prefix of a globally ordered sequence of events.

In this paper, we introduce a lifeline removal transformation operator on interactions. This operator is utilized to disregard parts of the interaction that are no longer observed, and we elaborate on its application to enhance multi-trace analysis, specifically in confronting multi-traces with interaction specifications.
This work builds upon the framework developed in \cite{denotational_and_operational_semantics_for_interaction_languages_application_to_trace_analysis}, in which we introduced three equivalent global trace semantics of interactions: the first, described as denotational, is based on the composition of algebraic operators on traces; the second, described as operational, allows interactions to be executed step by step, and the third one, called execution, allows for an efficient implementation of trace or multi-trace analysis algorithms as in \cite{a_small_step_approach_to_multi_trace_checking_against_interactions}. 
Expanding upon \cite{interaction_based_offline_runtime_verification_of_distributed_systems}, we use the denotational semantics to prove some properties on the lifeline removal transformation and the operational semantics to define a new offline RV algorithm for DS that leverages lifeline removal to handle {\em partial observation}. 
This allows us to identify prefixes of correct multi-traces, the notion of prefix being here that each local trace is a prefix of the corresponding local trace on an ideal complete multi-trace. 
Because not all prefixes of multi-traces can be obtained by projecting a prefix of a corresponding global trace, a simple adaptation of the algorithm from \cite{a_small_step_approach_to_multi_trace_checking_against_interactions} is not enough. That is why we introduced the lifeline removal operator in \cite{interaction_based_offline_runtime_verification_of_distributed_systems}.

This paper complements \cite{interaction_based_offline_runtime_verification_of_distributed_systems} in several ways. Firstly, we bridge the gap between previous studies \cite{denotational_and_operational_semantics_for_interaction_languages_application_to_trace_analysis} by introducing a projection operator from global traces to multi-traces, which grounds the lifeline removal operator. 
Secondly, we introduce two optimization techniques that both leverage the lifeline removal operator and which are respectively related to partial order reduction and guiding the global analysis via local analyses.

This paper is organized as follows.
After some preliminary notions are introduced in Sec.\ref{sec:prelim}, Sec.\ref{sec:core} introduces multi-traces, interactions and the lifeline removal operator. 
In Sec.\ref{sec:multipref}, we define the baseline offline RV algorithm, already featured in \cite{interaction_based_offline_runtime_verification_of_distributed_systems}, that analyzes multi-traces against interactions under conditions of partial observation.
Sec.\ref{sec:graph_size} introduces partial order reduction and local analyses, two techniques that significantly reduce the size of its search space.
In Sec.\ref{sec:experiments}, we present experimental results conducted on a tool that implements our approach.
Finally, Sec.\ref{sec:related} discusses related works and we conclude in Sec.\ref{sec:conclusion}.

\section{Preliminaries\label{sec:prelim}}

We introduce some preliminary notions that will be used throughout the paper.

\paragraph{Sets and words} 
Given a finite set $A$, $|A|$ designates its cardinality, $\mathcal{P}(A)$ denotes the set of all subsets of $A$ and 
$A^*$ is the set of words on $A$, with $\varepsilon$ the empty word and the ``$.$'' concatenation law. The concatenation law is extended to sets as usual: for $a$ in $A$, for $B$ and $C$ two subsets of $A$, $a.B$ (resp. $B.a$) is the set of words $a.w$ (resp. $w.a$) with $w$ in $B$ and $B.C$ is the set of words $w.w'$ with $w$ in $B$ and $w'$ in $C$. 

For any word $w \in A^*$, a word $w'$ is a prefix of $w$ if there exists a word $w''$, possibly empty, such that $w = w'.w''$. Let $\overline{w}$ denote the set of all prefixes of a word $w \in A^*$ and $\overline{W}$ the set of prefixes of all words of a set $W \subseteq A^*$. 

For any word $w \in A^*$, $|w|$ denotes the length of $w$ and for any $\depthAnaLoc \in \mathbb{N}$, $w\lbrack 0 .. \depthAnaLoc \rbrack$ denotes either $w$ if $\depthAnaLoc \geq |w|$, or the prefix of $w$ that contains its $\depthAnaLoc$ first elements if $\depthAnaLoc < |w|$.

Let $J$ be a finite set.
For a family $(X_j)_{j \in J}$ of sets indexed by $J$, $\prod_{j \in J} X_j$ is the set of tuples $(x_1, \ldots, x_j,\ldots)$ with $\forall j \in J, x_j \in X_j$.
For a tuple  $\mu = (x_1, \ldots, x_j,\ldots)$ indexed by $J$, $\mu_{|j}$ denotes the component $x_j$.

\paragraph{Terms and positions}
Considering a finite set of operation symbols $\mathcal{F} = \bigcup_{\substack{j \geq 0}} \mathcal{F}_j$ with finite arities $j \geq 0$, the set of terms over $\mathcal{F}$ is the smallest set $\mathcal{T}_{\mathcal{F}}$ such that:
\begin{itemize}
    \item $\mathcal{F}_0 \subset \mathcal{T}_{\mathcal{F}}$ (symbols of arity 0 are called constants)
    \item and for any symbol $f$ in $\mathcal{F}_j$ of arity $j>0$ and terms $t_1,\cdots,t_j$ in $\mathcal{T}_{\mathcal{F}}$,\\ $f(t_1,\cdots,t_j) \in \mathcal{T}_{\mathcal{F}}$.
\end{itemize}

For any term $t \in \mathcal{T}_\mathcal{F}$, its set of positions $pos(t) \in \mathcal{P}((\mathbb{N}^+)^*)$ accordingly to the Dewey Decimal Notation \cite{dershowitz_rewrite_systems} is such that: 
\begin{itemize}
    \item for $t$ in $\mathcal{F}^0$, $pos(t) = \{ \varepsilon \}$,
    \item and for $t$ of the form $f(t_1,\cdots,t_j)$ in $\mathcal{T}_\mathcal{F}$, \\
$pos(t) = \{ \varepsilon \} \cup \bigcup_{k \in [1,j]} \{k.p ~|~ p \in pos(t_k)\}$.
\end{itemize}

For a term $t$ in $\mathcal{T}_\mathcal{F}$ and a position $p$ in $pos(t)$, we denote by $t(p)$ the operation symbol at position $p$ within $t$.
Assuming that $t$ may contain several occurrences of a constant $a$, we may use the notation $a@p$ to refer unambiguously to the instance of $a$ at a position $p$ (in a context where $t$ is known).

\paragraph{Binary relations and graphs}

A binary relation $\rightsquigarrow$ on a set $A$ is a subset of $A \times A$, commonly used with an infix notation. For any two relations $\rightsquigarrow$ and $\rightarrow$:
\begin{itemize}
    \item composition is s.t. $\rightsquigarrow \circ \rightarrow = \{(x,z) \in A^2 ~|~ \exists~ y \in A, (x\rightsquigarrow y) \mbox{ and } (y \rightarrow z)\}$,
    \item $\overset{0}{\rightsquigarrow}$ denotes the identity relation $\{(x,x) ~|~ x \in A\}$,
    $\overset{1}{\rightsquigarrow}$ denotes the relation $\rightsquigarrow$ itself,
    and for any $j>1$, $\overset{j}{\rightsquigarrow}$ is the relation $\overset{j-1}{\rightsquigarrow} \circ \rightsquigarrow$,
    \item and $\overset{*}{\rightsquigarrow}$ denotes the reflexive and transitive closure $\bigcup_{j = 0}^{\infty} \overset{j}{\rightsquigarrow}$ of $\rightsquigarrow$.
\end{itemize}

A graph $\mathbb{G}$ is a tuple $(A, \rightsquigarrow)$ where $A$ is its (support) set of vertices and $\rightsquigarrow$ is its transition relation, which is a binary relation on $A$.

\section{Lifeline removal and interaction semantics\label{sec:core}}

\subsection{Interactions\label{sec:interactions}}

Interactions languages \cite{denotational_and_operational_semantics_for_interaction_languages_application_to_trace_analysis} are formal languages that encode sequence diagrams akin to those of UML-SD \cite{UML} or Message Sequence Charts \cite{operational_semantics_for_msc} and their variations \cite{high_level_message_sequence_charts,lscs_breathing_life_into_message_sequence_charts}. Interactions and the associated sequence diagrammatic representations describe the expected behaviors of Distributed Systems (DS).
A DS comprises multiple subsystems potentially distributed across spatial locations.
Such subsystems communicate with each other via asynchronous message passing.
From a black box perspective, the atomic concept to describe the executions of such DS is that of communication {\em actions} occurring on a subsystem's interface.
These actions consist in either the emission or the reception of a message by or from a specific subsystem.

Using the terminology of interaction languages \cite{denotational_and_operational_semantics_for_interaction_languages_application_to_trace_analysis,operational_semantics_for_msc}, a subsystem interface is called a \emph{lifeline} and corresponds to an interaction point on which the subsystem can receive or send some messages. Throughout the paper, we adopt the following notational conventions: 

\begin{definition}
Lifelines belong to the set $\mathcal{L}$ denoting the universe of lifelines, while messages are represented within the universe $\mathcal{M}$.    
\end{definition}

An action occurring on a lifeline is then defined by its kind (emission or reception, identified resp.~by the symbols $!$ and $?$) and the message it conveys. Def.\ref{def:actions_and_traces} formalizes this notion of actions, as well as that of execution {\em traces} which are sequences of observable actions, akin to words which letters are actions.

\begin{definition}\label{def:actions_and_traces}
For a lifeline $\hlf{l} \in \mathcal{L}$, the set $\mathbb{A}_{\hlf{l}}$ of \emph{actions over $\hlf{l}$} is: 
\[\{ \hlf{l} \Delta \hms{m} ~|~ \Delta \in \{!,?\}, ~ \hms{m} \in \mathcal{M} \}\]

\noindent The set of \emph{local traces} over $\hlf{l}$ is $\mathbb{A}_{\hlf{l}}^*$. 

For any set of lifelines $L \subseteq \mathcal{L}$, the set $\mathbb{A}(L)$ of actions over $L$ is the set $\bigcup_{\hlf{l}\in L} \mathbb{A}_{\hlf{l}}$ and the set of \emph{global traces} over $L$ is $\mathbb{A}(L)^*$.

For $a \in \mathbb{A}(L)$ as $\hlf{l}?m$ or $\hlf{l}!m$, $\theta(a)$ refers to the lifeline $\hlf{l}$ on which the action $a$ takes place.
\end{definition}

 An example of interaction is given on Fig.\ref{fig:interaction_example}.
It models the expected behaviors of a DS composed of three remote subsystems, assimilated to their interfaces $\hlf{l_p}$, $\hlf{l_b}$ and $\hlf{l_s}$. This DS implements a simplified publish/subscribe scheme of communications (an alternative to client-server architecture), which is a cornerstone of some protocols used in the IoT such as MQTT \cite{MQTT}. The publisher $\hlf{l_p}$ may publish messages towards the broker $\hlf{l_b}$ which may then forward them to the subscriber $\hlf{l_s}$ if it is already subscribed.

\begin{figure}[ht]
    \centering
    
\begin{subfigure}[t]{.3\linewidth}
    \centering
    \includegraphics[scale=.3]{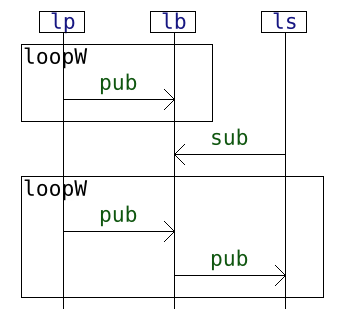}
    \caption{sequence diagram\label{fig:interaction_example_diagram}}
\end{subfigure}
\begin{subfigure}[t]{.42\linewidth}
    \centering
    \scalebox{.75}{
\begin{tikzpicture}[every node/.style = {shape=rectangle, align=center}]
\input{images/int_pbs/term_pbs}
\end{tikzpicture}
}
    \caption{interaction term\label{fig:interaction_example_term}}
\end{subfigure}
\begin{subfigure}[t]{.26\linewidth}
    \centering
    \scalebox{.6}{
\begin{tikzpicture}[every node/.style = {shape=rectangle, align=center}]
\input{images/int_pbs/position_pbs}
\end{tikzpicture}
}
    \caption{positions\label{fig:interaction_example_position}}
\end{subfigure}

    \caption{Example interaction}
    \label{fig:interaction_example}
\end{figure}

Fig.\ref{fig:interaction_example_diagram} corresponds to the sequence 
diagrammatic representation of that interaction.
Each lifeline is depicted by a vertical line labeled by its name at the top. By default, the top-to-bottom direction represents time passing. That is, a communication action depicted above another one on the same lifeline occurs beforehand. Communication actions are represented by horizontal arrows labeled with the action's message.
Whenever an arrow exits (resp.~enters) a lifeline, there is a corresponding emission (resp.~reception) action at that point on the line. 
For example, the horizontal arrow from $\hlf{l_s}$ to the lifeline $\hlf{l_b}$ carrying the message $\hms{\mathtt{sub}}$ indicates that the subscriber emits a subscription message, which the broker then receives.
This behavior corresponds to the strict sequencing of two communication actions: $\hlf{l_s}!\hms{\mathtt{sub}}$ followed by $\hlf{l_b}?\hms{\mathtt{sub}}$.

The diagram's top-to-bottom reading only concerns actions that occur on the same lifeline. For instance, while the $\hlf{l_b}!\hms{\mathtt{pub}}$ at the bottom of the diagram can only occur after the $\hlf{l_b}?\hms{\mathtt{pub}}$ above it, the $\hlf{l_s}!\hms{\mathtt{sub}}$ on the right can occur either before or after any of the $\hlf{l_p}!\hms{\mathtt{pub}}$ on the left.
Hence, we associate this top-to-bottom direction with an operator called weak sequencing (weak in contrast to the strict sequencing associated with the horizontal arrows).

More complex behaviors can be introduced through operators (see combined fragments in UML-SD) drawn in the shape of boxes that frame sub-behaviors of interest. For instance, in Fig.\ref{fig:interaction_example_diagram}, $loop_W$ corresponds to a weakly sequential loop, a repetition of the box content using weak sequencing. From the perspective of the $\hlf{l_b}$ lifeline, this implies that it can observe words of the form $(\hlf{l_b}?\hms{\mathtt{pub}})^*.\hlf{l_b}?\hms{\mathtt{sub}}.(\hlf{l_b}?\hms{\mathtt{pub}}.\hlf{l_b}!\hms{\mathtt{pub}})^*$ i.e. it can receive an arbitrary number of instances of the $\hms{\mathtt{pub}}$ message then one instance of $\hms{\mathtt{sub}}$ and then it can receive and transmit an arbitrary number of $\hms{\mathtt{pub}}$.

An example of a global trace, that is representative of the behaviors specified by the interaction from Fig.\ref{fig:interaction_example} is given below: \\
\centerline{$\hlf{l_s}!\hms{\mathtt{sub}}.\hlf{l_p}!\hms{\mathtt{pub}}.\hlf{l_b}?\hms{\mathtt{sub}}.\hlf{l_b}?\hms{\mathtt{pub}}.\hlf{l_b}!\hms{\mathtt{pub}}.\hlf{l_s}?\hms{\mathtt{pub}}$}

This trace illustrates that the $\hlf{l_p}$ and $\hlf{l_s}$ lifelines can send their respective messages $\hms{\mathtt{pub}}$ and $\hms{\mathtt{sub}}$ in any order since there are no constraints on their ordering. In contrast, the reception of a message necessarily takes place after its emission.
Since the reception of the message $\hms{\mathtt{sub}}$ takes place before that of the $\hms{\mathtt{pub}}$ message, this last message necessarily corresponds to the one occurring in the bottom loop.
This trace is a typical example of a trace {\em accepted} by the interaction in Fig.\ref{fig:interaction_example}, as this trace completely realizes the specified behavior by: 
\begin{itemize}
    \item unfolding zero times the first loop; \item realizing the passing of the message $\hms{\mathtt{sub}}$ between lifelines $\hlf{l_s}$ and $\hlf{l_b}$; 
    \item unfolding one time the second loop. 
\end{itemize}
Let us point out that none of the prefixes of this accepted trace is an accepted trace.

Sequence diagrams, such as the one in Fig.\ref{fig:interaction_example_diagram} can be formalized as terms of an inductive language. 
In this paper, we consider the language from \cite{a_small_step_approach_to_multi_trace_checking_against_interactions,denotational_and_operational_semantics_for_interaction_languages_application_to_trace_analysis}, which we recall in Def.\ref{def:interaction_language}.

\begin{definition}\label{def:interaction_language}
Given $L \subseteq \mathcal{L}$, the set $\mathbb{I}(L)$ of interactions over $L$ is the set of terms built over the following symbols provided with arities in $\mathbb{N}$:
\begin{itemize}
    \item the empty interaction $\varnothing $ and any action $a$ in $\mathbb{A}(L)$  of arity 0;
    \item the three loop operators $loop_S$, $loop_W$ and $loop_P$ of arity 1; 
    \item and the four operators $strict$, $seq$, $par$ and $alt$ of arity 2.
\end{itemize}
\end{definition}

\begin{figure}[ht]
    \centering
\begin{subfigure}[t]{.4\linewidth}
    \centering
\begin{tikzpicture}
\node[draw] (i0) at (0,0) {\includegraphics[scale=.225]{images/int_pbs/fsen_scp_ex_pbs.png}};
\node[draw,below right=-.75cm and .25cm of i0] (i1) {\includegraphics[scale=.225]{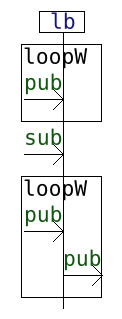}};
\draw (i0.230) edge[bend right=25,->] node[midway,draw,fill=white] {\scriptsize$\lifelineElim_{\{\hlf{l_p},\hlf{l_s}\}}$} (i1.200);
\begin{pgfonlayer}{main}
\fill[red!100,opacity=.4] (.7,1.3) -- (1.1,1.3) -- (1.1,-1.2) -- (.7,-1.2) -- cycle;
\fill[red!100,opacity=.4] (-1.06,1.3) -- (-.66,1.3) -- (-.66,-1.2) -- (-1.1,-1.2) -- cycle;
\end{pgfonlayer};
\end{tikzpicture}
    \caption{...on the diagram\label{fig:lf_removal_diagram}}
\end{subfigure}
\begin{subfigure}[t]{.575\linewidth}
    \centering
    \scalebox{.7}{
\begin{tikzpicture}
\node[draw] (transfo) at (0,0) {\begin{tikzpicture}[every node/.style = {shape=rectangle, align=center}]
\input{images/int_pbs/term_pbs}
\draw[red,thick] (o1111.south east) -- (o1111.north west);
\draw[red,thick] (o1111.south west) -- (o1111.north east);
\draw[red,thick] (o121.south east) -- (o121.north west);
\draw[red,thick] (o121.south west) -- (o121.north east);
\draw[red,thick] (o2111.south east) -- (o2111.north west);
\draw[red,thick] (o2111.south west) -- (o2111.north east);
\draw[red,thick] (o2122.south east) -- (o2122.north west);
\draw[red,thick] (o2122.south west) -- (o2122.north east);
\draw[green,thick] (o111.east) -- (o111.west);
\draw[->,green,thick] ([xshift=5pt,yshift=-2pt] o1112.north) to [bend right=45] ([xshift=-2.5pt] o111.east);
\draw[green,thick] (o12.east) -- (o12.west);
\draw[->,green,thick] ([xshift=5pt,yshift=-2pt] o122.north) to [bend right=45] ([xshift=-2.5pt] o12.east);
\draw[green,thick] (o211.east) -- (o211.west);
\draw[->,green,thick] ([xshift=5pt,yshift=-2pt] o2112.north) to [bend right=45] ([xshift=-2.5pt] o211.east);
\draw[green,thick] (o212.east) -- (o212.west);
\draw[->,green,thick] ([xshift=6pt,yshift=-2pt] o2121.north) to [bend right=25] ([yshift=1.5pt] o212.south);
\node[draw,align=left] (legP) at (-2,-3.1) {
\textcolor{red}{$\blacksquare$} $\lifelineElim_{\{\hlf{l_p},\hlf{l_s}\}}$\\
\textcolor{green}{$\blacksquare$} simplification
};
\end{tikzpicture}};
\node[draw, below=.3cm of transfo] (result) {\begin{tikzpicture}[every node/.style = {shape=rectangle, align=center}]
\node (o) { $seq$ } [sibling distance=2cm,level distance=0.3cm]
  child { node (o1) { $seq$ } [sibling distance=1cm,level distance=0.4cm]
  child { node (o11) {$loop_W$} [sibling distance=.75cm]
    child { node (o111) {$\hlf{l_b}?\hms{\mathtt{pub}}$} }
  }
  child { node (o12) {$\hlf{l_b}?\hms{\mathtt{sub}}$} }
  }
  child { node (o2) {$loop_W$} [level distance=0.4cm]
    child { node (o21) {$seq$} [sibling distance=1.1cm]
      child { node (o211) {$\hlf{l_b}?\hms{\mathtt{pub}}$} }
      child { node (o212) {$\hlf{l_b}!\hms{\mathtt{pub}}$} }
    }
  };
\end{tikzpicture}};
\draw[->] (transfo) -- (result);
\end{tikzpicture}
}
    \caption{...on the term\label{fig:lf_removal_term}}
\end{subfigure}
    \caption{Application of lifeline removal on the example from Fig.\ref{fig:interaction_example}...}
    \label{fig:lf_removal}
\end{figure}

Because the language from Def.\ref{def:interaction_language} contains symbols of arities up to 2, we can represent its terms as binary trees.
Fig.\ref{fig:interaction_example_term} takes advantage of this to represent an encoding of the sequence diagram from Fig.\ref{fig:interaction_example_diagram} as an interaction term.
We use the strict sequencing operator ``$strict$'' to encode message passing as it enforces strict happen-before relations between the actions within its left sub-term of those within its right sub-term.
By contrast, the more generic top-to-bottom order of the diagram is encoded using weak sequencing ``$seq$'' because this operator only enforces a strict order between actions occurring on the same lifeline.
The ``$par$'' and ``$alt$'' operators resp.~correspond to interleaving and non-deterministic choice.
``$loop_S$'', ``$loop_W$'' and ``$loop_P$'' correspond to repetitions (Kleene closures, see \cite{denotational_and_operational_semantics_for_interaction_languages_application_to_trace_analysis}) using resp. $strict$, $seq$ and $par$. Thus, for an interaction $i$, traces accepted by $loop_S(i)$ (resp. $loop_W(i)$, $loop_P(i)$) are either the empty trace or a trace accepted by  $strict(i,loop_S(i))$ (resp. $seq(i,loop_W(i))$, $par(i,loop_P(i))$).

Using terms to denote interactions allows the use of positions to designate occurrences of actions unambiguously.
Fig.\ref{fig:interaction_example_position} represents the positions of all the symbols of the term from Fig.\ref{fig:interaction_example_term}.
For example, the two distinct instances of $\hlf{l_p}!\hms{\mathtt{pub}}$ can be referred to as $\hlf{l_p}!\hms{\mathtt{pub}}@1111$ or $\hlf{l_p}!\hms{\mathtt{pub}}@2111$.

\subsection{Lifeline removal \label{sec:lf-removal}}

Interactions describe coordinated behaviors that take place on different lifelines. If we are only interested in a subset of them,  it is possible to restrict the scope of interactions to lifelines of interest. 
Fig.\ref{fig:lf_removal_diagram} illustrates this using sequence diagrammatic representations on a specific example.
Starting from the diagram of Fig.\ref{fig:interaction_example_diagram}, to project it onto $\{\hlf{l_b}\}$, it suffices to remove everything that concerns the other lifelines (i.e., the complementary $L \setminus \{\hlf{l_b}\}$), as highlighted in red. Removal simply replaces each action occurring on $L \setminus \{\hlf{l_b}\}$ with the empty interaction $\varnothing$.
Such transformations can also be described as a form of horizontal decomposition as opposed to the vertical decomposition induced by the weak sequential operator  $seq$, which schedules actions on a lifeline according to the passing of time.
Such transformations can be implemented via a syntactic operator $\lifelineElim$ (and in particular, $\lifelineElim_{\{\hlf{l_p},\hlf{l_s}\}}$ for the example from Fig.\ref{fig:lf_removal}) on our interaction language, which we define in Def.\ref{def:lifeline_removal_operator_on_interactions}.

\begin{definition}\label{def:lifeline_removal_operator_on_interactions}
Let $L$ and $H$ be two sets of lifelines verifying $H \subseteq L$. We define the removal operation 
 $\lifelineElim_H: \mathbb{I}(L) \rightarrow \mathbb{I}(L \setminus H)$ by: \\
$\lifelineElim_H(i) = \textbf{ match } i \textbf{ with}$
\begin{center}
$\begin{array}{lll}
|~\varnothing & \rightarrow 
& \varnothing
\\
|~a \in \mathbb{A}(L) & 
\rightarrow 
& \textbf{if } \theta(a) \in H \textbf{ then } \varnothing \textbf{ else } a
\\
|~f(i_1,i_2) & \rightarrow 
& f(\lifelineElim_H(i_1),\lifelineElim_H(i_2)) \textbf{ for } f \in \{ strict,seq,alt,par \}
\\
|~loop_k(i_1) & \rightarrow 
& loop_k(\lifelineElim_H(i_1)) \textbf{ for } k \in \{S,W,P\}
\end{array}$ 
\end{center} 
with $i$ an interaction in $ \mathbb{I}(L)$.
\end{definition}

Lifeline removal, as defined in Def.\ref{def:lifeline_removal_operator_on_interactions} in functional style, preserves the term structure of interactions, replacing actions on the removed lifelines with the empty interaction $\varnothing$. 
Fig.\ref{fig:lf_removal_term}, which is the interaction language counterpart to the example of Fig.\ref{fig:lf_removal_diagram}, illustrates the effect of lifeline removal on the term encoding from Fig.\ref{fig:interaction_example_term}.
We may use additional simplification steps involving the empty interaction $\varnothing$ to keep the resulting term simple. 
Indeed, by the very meaning of $\varnothing$ as the empty interaction, $\varnothing$ is a neutral element for scheduling operations. This means that, for $f$ in $\{strict,par,seq\}$, $f(i,\varnothing)$, $f(\varnothing,i)$ and $i$ will share the same set of accepted traces. Likewise, $alt(\varnothing,\varnothing)$, $loop_f(\varnothing)$ for $f$ in $\{S,W,P\}$ will share with $\varnothing$ the same set of accepted traces, i.e. $\{\varepsilon\}$. 
In \cite{denotational_and_operational_semantics_for_interaction_languages_application_to_trace_analysis,finite_automata_synthesis_from_interactions}, we have expressed these properties as rewriting rules\footnote{in this paper, we only consider rules that remove redundant $\varnothing$, as in \cite{finite_automata_synthesis_from_interactions}, while in \cite{denotational_and_operational_semantics_for_interaction_languages_application_to_trace_analysis} additional rules are also considered.} on the set of interaction terms preserving the set of accepted traces.
In the remainder of the paper, we will use such simplification rules primarily for the sake of simplicity.

\begin{figure}[ht]
    \centering

\begin{tikzpicture}
\node[draw] (i0) at (0,0) {\includegraphics[scale=.2]{images/int_pbs/fsen_scp_ex_pbs.png}};
\node[draw,below left=-.75cm and .5cm of i0] (i1) {\includegraphics[scale=.2]{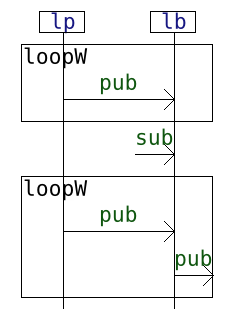}};
\node[draw,below right=-.75cm and .5cm of i0] (i2) {\includegraphics[scale=.2]{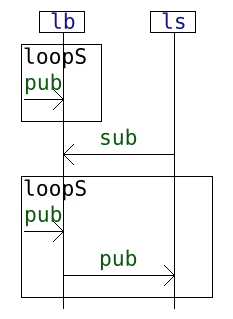}};
\draw (i0.150) edge[->,bend right=25] node[midway,draw,fill=white] {\scriptsize$\lifelineElim_{\{\hlf{l_s}\}}$} (i1.north);
\draw (i0.30) edge[->,bend left=25] node[midway,draw,fill=white] {\scriptsize$\lifelineElim_{\{\hlf{l_p}\}}$} (i2.north);
\node[draw,below right=.25cm and -1cm of i0] (i3) {\includegraphics[scale=.2]{images/int_pbs/fsen_scp_ex_b.png}};
\draw (i0.230) edge[->, bend right=20] node[midway,draw,fill=white] {\scriptsize$\lifelineElim_{\{\hlf{l_p},\hlf{l_s}\}}$} (i3.150);
\draw (i1.south) edge[->,bend right=25] node[midway,draw,fill=white] {\scriptsize$\lifelineElim_{\{\hlf{l_p}\}}$} (i3.210);
\draw (i2.290) edge[->,bend left=25] node[midway,draw,fill=white] {\scriptsize$\lifelineElim_{\{\hlf{l_s}\}}$} (i3.300);
\end{tikzpicture}
    
    \caption{Illustrating algebraic properties of $\lifelineElim$ on the example from Fig.\ref{fig:interaction_example}}
    \label{fig:lifeline_removal_interaction_commut}
\end{figure}

By construction, the syntactic lifeline removal operator on interaction has interesting algebraic properties related to composition: 

\begin{property}
\label{prop:commut_hide}
For any $L \subseteq \mathcal{L}$, $H \subseteq L$, $H' \subseteq L \setminus H$ and $i \in \mathbb{I}(L)$ we have:\\
$\lifelineElim_H \circ \lifelineElim_{H'} (i) = \lifelineElim_{H'} \circ \lifelineElim_H (i) = \lifelineElim_{H \cup H'} (i)$
\end{property}

\begin{proof}
Trivial.
\end{proof}

Prop.\ref{prop:commut_hide} is illustrated on Fig.\ref{fig:lifeline_removal_interaction_commut} with example on Fig.\ref{fig:interaction_example}.
Lastly, lifeline removal preserves the term structure and the existence and position of all actions that do not occur on removed lifelines.

\begin{property}
\label{prop:positions_and_lifeline_removal}
For any $L \subset \mathcal{L}$, any interaction $i \in \mathbb{I}(L)$ and any set of lifelines $\emptyset \subsetneq H \subsetneq L$, we have:
\begin{itemize}
    \item $pos(i) = pos(\lifelineElim_H(i))$
    \item for all $a$ in $\mathbb{A}(L)$ with $\theta(a) \not\in H$, if there exists $p$ in $pos(i)$ with $i(p) = a$ then $\lifelineElim_H(i)(p) = a$
    \item reciprocally, for all $a$ in $\mathbb{A}(L)$, if there exists $p$ in $pos(\lifelineElim_H(i))$ such that $\lifelineElim_H(i)(p) = a$ then $i(p) = a$.
\end{itemize}
\end{property}

\begin{proof}
We reason without any interaction term simplification. As per Def.\ref{def:lifeline_removal_operator_on_interactions}, for any interaction $i$ and set $H \subset L$, $i$ and $\lifelineElim_H(i)$ have the same set of positions and for any action $a \in \mathbb{A}(L)$ s.t.~$\theta(a) \not\in H$, there is a $a@p$ in $i$ iff there is also the same $a@p$ in $\lifelineElim_H(i)$.
\end{proof}

\subsection{Traces and multi-traces\label{sec:multitraces}}

From the observation of an execution of a DS, the behavior of that DS can be characterized by the actions that occurred during the span of that execution and the order of occurrences of these actions.
Still, our ability to order these actions is limited by the existence of clocks \cite{Lamport19b} that the different subsystems of the DS may or may not share.

In this paper, we only consider two simple cases, as illustrated in Fig.\ref{fig:colocalization_schema}, which are based on the example DS described by the interaction from Fig.\ref{fig:interaction_example}.
The publisher $\hlf{l_p}$ is represented as a sensor device, the broker $\hlf{l_b}$ via a cloud, and the subscriber $\hlf{l_s}$ via a smartphone.

If all subsystems share the same clock (Fig.\ref{fig:colocalization_schema_global}), it is possible to reorder all actions to obtain a single global trace. 
By contrast, in the case where every subsystem has its own local clock (Fig.\ref{fig:colocalization_schema_discrete}), this is generally not possible, as we cannot reliably correlate the timestamps of distant clocks \cite{Lamport19b}. Therefore, the behavior of the DS can only be characterized by a structured set of local traces (one per lifeline), which we call a multi-trace \cite{constraint_based_oracles_for_timed_distributed_systems,a_small_step_approach_to_multi_trace_checking_against_interactions}.

\begin{figure}[h]
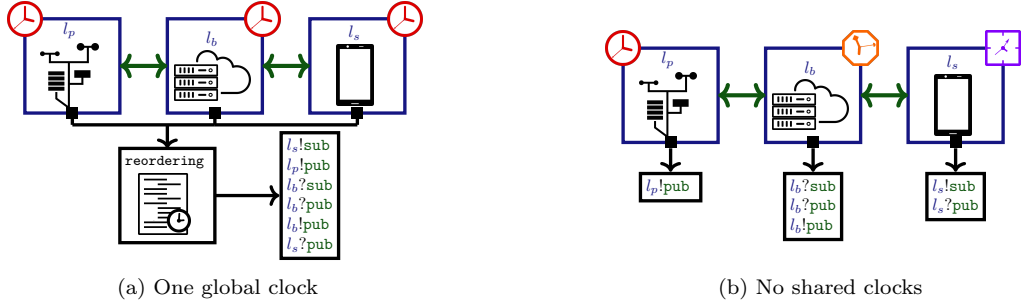

    \centering

\begin{subfigure}{.475\linewidth}
    \centering
    \scalebox{.625}{\input{figures/subfig/colocal/archi_glob}}
    \caption{One global clock\label{fig:colocalization_schema_global}}
\end{subfigure}
\begin{subfigure}{.475\linewidth}
    \centering
    \scalebox{.625}{\input{figures/subfig/colocal/archi_loc}}
    \vspace*{.25cm}
    \caption{No shared clocks\label{fig:colocalization_schema_discrete}}
\end{subfigure}

    \caption{Clocks and characterization of executed behaviors\label{fig:colocalization_schema}}
\end{figure}

We formalize this notion of {\em multi-trace} as a structured set of local traces.
\begin{definition}\label{def:multitrace}
Given $L \subseteq \mathcal{L}$, the set $\mathbb{M}(L)$ of {\em multi-traces} over $L$ is $\prod_{\hlf{l} \in L} \mathbb{A}_{\hlf{l}}^*$.
\end{definition}

For $\mu = (t_{\hlf{l}})_{\hlf{l} \in L}$ in $\mathbb{M}(L)$, $\mu_{|\hlf{l}}$ is the trace component $t_{\hlf{l}} \in \mathbb{A}_{\hlf{l}}^*$.  
We denote by $\varepsilon_L$ the empty multi-trace in $\mathbb{M}(L)$ defined by: $\forall~\hlf{l} \in L, {\varepsilon_L}_{|\hlf{l}} = \varepsilon$.
Additionally, for any $\mu \in \mathbb{M}(L)$, we use the notations $\mu[t]_{\hlf{l}}$ to designate the multi-trace $\mu$ in which the component on $\hlf{l} \in L$ has been replaced by $t \in \mathbb{A}_{\hlf{l}}^*$ and $|\mu|$ to designate the cumulative length $|\mu| = \sum_{\hlf{l} \in L} |\mu_{|\hlf{l}}|$ of $\mu$.

Let us introduce operations to add an action 
to the left (resp.~right) of a multi-trace. For the sake of simplicity, we use the same symbol $\multiAppend$ for these left- and right-concatenation operations: 
\[
\forall a \in \mathbb{A}(L), \forall \mu \in \mathbb{M}(L), \;  a\multiAppend\mu = \mu[a.\mu_{|\theta(a)}]_{\theta(a)} \mbox{~~~and~~~} \mu \multiAppend a= \mu[\mu_{|\theta(a)}.a]_{\theta(a)}\] 
Note that for any $\mu$ and $a$, we have $|\mu \multiAppend a| = |a \multiAppend \mu| = |\mu|+1$.
We extend the operation ``$\multiAppend$'' to sets of multi-traces $M \subset \mathbb{M}(L)$ as follows:
\begin{center}
$a \multiAppend M = \{ a \multiAppend \mu ~|~ \mu \in M \}$ and 
$M \multiAppend a = \{ \mu \multiAppend a ~|~ \mu \in M \}$.
\end{center}

Global traces over $L$ can be rewritten into a multi-trace over $L$  by separating their actions according to their lifelines. For that, we can use the  $\multiAppend$ operator that places the actions in the appropriate local trace:

\begin{definition}
\label{def:multi_trace_projection}
For any $L \subseteq \mathcal{L}$, we define $\muProjection_L : \mathbb{A}(L)^* \rightarrow \mathbb{M}(L)$ as follows:
\begin{itemize}
\item $\muProjection_L(\varepsilon) = \varepsilon_L$
\item $\forall~a \in \mathbb{A}(L),~\forall~t \in \mathbb{A}(L)^*,
\muProjection_L(a.t) = a \multiAppend \muProjection_L(t)$
\end{itemize}
$\muProjection_L$ is canonically extended to sets of traces. 
\end{definition}

Fig.\ref{fig:multitrace_projections} illustrates the use of this operator with our example, which amounts to transforming the global trace of Fig.\ref{fig:colocalization_schema_global} into the multi-trace of Fig.\ref{fig:colocalization_schema_discrete}.
Given a global trace and a multi-trace obtained from observing the same execution of a DS, the multi-trace likely contains less information because of the absence of strict ordering between actions occurring on distinct lifelines.
Consequently, decomposing a global trace into a multi-trace leads to a loss of information: for a given $\mu \in \mathbb{M}(L)$, there can be several $t \in \mathbb{A}(L)^*$ s.t.~$\muProjection_L(t) = \mu$.
The multi-trace of Fig.\ref{fig:colocalization_schema_discrete} involves 3 local traces: 
$\mu_{|\{\hlf{l_p}\}\}} = \hlf{l_p}!\hms{\mathtt{pub}}$ for subsystem $\hlf{l_p}$, 
$\mu_{|\{\hlf{l_b}\}\}} = \hlf{l_b}?\hms{\mathtt{sub}}.\hlf{l_b}?\hms{\mathtt{pub}}.\hlf{l_b}!\hms{\mathtt{pub}}$ 
for subsystem $\hlf{l_b}$,
and 
$\mu_{|\{\hlf{l_s}\}\}} = \hlf{l_s}!\hms{\mathtt{sub}}.\hlf{l_s}?\hms{\mathtt{pub}}$ for $\hlf{l_s}$. 
This multi-trace can be obtained via projecting the global trace from Fig.\ref{fig:colocalization_schema_global} (as illustrated on Fig.\ref{fig:multitrace_projections}), but also via projecting:\\
\noindent 
\centerline{$\hlf{l_p}!\hms{\mathtt{pub}}.\hlf{l_s}!\hms{\mathtt{sub}}.\hlf{l_b}?\hms{\mathtt{sub}}.\hlf{l_b}?\hms{\mathtt{pub}}.\hlf{l_b}!\hms{\mathtt{pub}}.\hlf{l_s}?\hms{\mathtt{pub}}$}

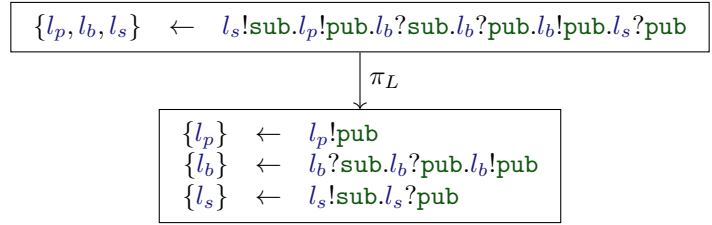
\begin{figure}[h]
    \centering

\begin{tikzpicture}
\node[draw] (ng) {
$
\begin{array}{rcl}
\{\hlf{l_p},\hlf{l_b},\hlf{l_s}\} 
&
\leftarrow
&
\hlf{l_s}!\hms{\mathtt{sub}}.\hlf{l_p}!\hms{\mathtt{pub}}.\hlf{l_b}?\hms{\mathtt{sub}}.\hlf{l_b}?\hms{\mathtt{pub}}.\hlf{l_b}!\hms{\mathtt{pub}}.\hlf{l_s}?\hms{\mathtt{pub}}
\end{array}
$
};
\node[draw, below=.75cm  of ng] (nd) {
$
\begin{array}{rcl}
\{\hlf{l_p}\}
&
\leftarrow 
&
\hlf{l_p}!\hms{\mathtt{pub}}
\\
\{\hlf{l_b}\}
&
\leftarrow 
&
\hlf{l_b}?\hms{\mathtt{sub}}.\hlf{l_b}?\hms{\mathtt{pub}}.\hlf{l_b}!\hms{\mathtt{pub}}
\\
\{\hlf{l_s}\}
&
\leftarrow 
&
\hlf{l_s}!\hms{\mathtt{sub}}.\hlf{l_s}?\hms{\mathtt{pub}}
\end{array}
$
};
\draw (ng) edge[->] node[midway,right] {$\muProjection_L$} (nd);
\end{tikzpicture}
    
    \caption{Multi-trace projection on the example from Fig.\ref{fig:colocalization_schema}}
    \label{fig:multitrace_projections}
\end{figure}


Global traces and multi-traces characterize distributed behaviors. 
Via the definition of algebraic operators, we can thus formalize means to compose and structure distributed behaviors.
In the following, we recall and adapt the operators on global traces from \cite{denotational_and_operational_semantics_for_interaction_languages_application_to_trace_analysis}.

Weak sequencing requires the introduction of a ``conflict'' predicate $\conflictPredicate$ such that for any trace $t$ in $\mathbb{A}(L)^*$ and any lifeline $l$ in $\mathcal{L}$, $t \conflictPredicate l$ signifies that $t$ contains an action occurring on lifeline $l$. It is defined as follows:
\[
\varepsilon \conflictPredicate l = \bot 
~~~~~~~\text{ and }~~~~~~~
(a.t) \conflictPredicate l = (\theta(a)=l) \vee (t \conflictPredicate l)
\]

Below, we define several operators, with two versions, one for global traces and the other for multi-traces. We use the same notation for both cases because, in practice, the context makes it possible to distinguish between them. Let us consider two global traces $t_1$ and $t_2$ in $\mathbb{A}(L)^*$ and two multi-traces $\mu_1$ and $\mu_2$ in $\mathbb{M}(L)$ with $L \subseteq \mathcal{L}$. \\

\noindent$\bullet$ $t_1 \cup t_2$ and $\mu_1 \cup \mu_2$ denote alternative defined as follows: 
{
\[
t_1 \cup t_2 = \{ t_1,~t_2 \} 
~~~~~~~~~~~~~~~~~~~~~
\mu_1 \cup \mu_2 = \{ \mu_1,~\mu_2 \}
\]}
\noindent$\bullet$ $t_1 \opStrictSeq t_2$ and $\mu_1 \opStrictSeq \mu_2$ denote strict sequencing defined as follows:
{
\[
\begin{array}{ccc}
t _1 \opStrictSeq \varepsilon = t_1 
&
~~~~~~~\text{ and }~~~~~~~
&
t_1 \opStrictSeq (a.t_2) = (t_1.a) \opStrictSeq t_2
\\
\mu _1\opStrictSeq \varepsilon_L = \mu _1 
&
~~~~~~~\text{ and }~~~~~~~
&
\mu_1 \opStrictSeq (a \multiAppend \mu_2) = (\mu_1 \multiAppend a) \opStrictSeq \mu_2
\end{array}
\]}
\noindent$\bullet$ $t_1 \opInterleaving t_2$ and $\mu_1 \opInterleaving \mu_2$ denote interleaving defined as follows:
{
\[
\begin{array}{c}
\varepsilon \opInterleaving t_2  = \{ t_2 \}  ~~~~~~~~~~~~~~~~~
t_1 \opInterleaving \varepsilon  = \{ t_1 \} \\
(a_1.t_1)  ~\opInterleaving~  (a_2.t_2) ~ = ~ ( a_1.(t_1 \opInterleaving ( a_2.t_2 ))) ~\cup~ (a_2.((a_1.t_1) \opInterleaving t_2 ))\\
\varepsilon_L \opInterleaving \mu_2  = \{ \mu_2 \}  ~~~~~~~~~~~~~~~~~
\mu_1 \opInterleaving \varepsilon_L  = \{ \mu_1 \} \\
(a_1 \multiAppend \mu_1)  ~\opInterleaving~  (a_2 \multiAppend \mu_2) ~ = ~ ( a_1 \multiAppend (\mu_1 \opInterleaving ( a_2 \multiAppend \mu_2 ))) ~\cup~ (a_2 \multiAppend ((a_1 \multiAppend \mu_1) \opInterleaving \mu_2 ))\\
\end{array}
\]}
\noindent$\bullet$ $t_1 \opWeakSeq t_2$ denotes weak sequencing on global traces defined as follows:
{
\[
\begin{array}{c}
\varepsilon \opWeakSeq t_2  = \{ t_2 \}  ~~~~~~~~~~~~~~~~~
t_1 \opWeakSeq \varepsilon  = \{ t_1 \} \\
(a_1.t_1) \opWeakSeq (a_2.t_2) =
\left\{
\begin{array}{l}
a_1.(t_1 \opWeakSeq ( a_2.t_2 )) 
~~~~~~~~~~~~~~~~~~~~~~~~~~~~~
{
\text{ if } (a_1.t_1) \conflictPredicate \theta(a_2)}
\\
( a_1.(t_1 \opWeakSeq ( a_2.t_2 ))) \cup (a_2.((a_1.t_1) \opWeakSeq \mu_2 ))
~~~~~~~
{
\text{ else}}
\end{array}
\right.
\end{array}
\]}

These operators are classically extended to sets of traces and multi-traces using the same notations (e.g., $M_1 \opInterleaving M_2 = \bigcup_{\mu_1 \in M_1} \bigcup_{\mu_2 \in M_2} \mu_1 \opInterleaving \mu_2$).
The projection operator $\muProjection_L$ preserves $\cup$, $\opStrictSeq$ and $\opInterleaving$.

Alternative $\cup$ corresponds to a non-deterministic choice between behaviors.
Interleaving $\opInterleaving$ allows any interleavings of the actions of its left and right operands.
Strict sequencing $\opStrictSeq$ imposes strict precedence (w.r.t.~``$.$'' and ``$\multiAppend$'' for resp.~traces and multi-traces) of all the actions from its left operand over all the actions from its right operand. 
Weak sequencing $\opWeakSeq$ imposes strict precedence only between actions on the same lifeline. Let us remark that there is no need for a weak sequencing operator for multi-traces, as it would then be identical to the strict sequencing operator. When considering multi-traces, weak and strict sequencing can no longer be distinguished.
Indeed, given $L \in \mathcal{L}$, let us consider three actions $a_1$, $a_2$ and $a_3$ in $\mathbb{A}(L)$ such that $\theta(a_1) \neq \theta(a_2)$ and $\theta(a_2) = \theta(a_3)$ 
Then, we have:
\[
\begin{array}{rcl}
\muProjection_L(a_2 \opWeakSeq a_1.a_3)
&
=
&
\muProjection_L(\{ a_2.a_1.a_3,~a_1.a_2.a_3 \})
\\
&
=
&
\{ (a_1,~a_2.a_3) \}
\\
\muProjection_L(a_2 \opStrictSeq a_1.a_3)
&
=
&
\muProjection_L(\{ a_2.a_1.a_3 \})
\\
&
=
&
\{ (a_1,~a_2.a_3) \}
\\
\end{array}
\]
On this example, after projection, we can no longer distinguish $a_2$ occurring before or after $a_1$.
Formally, this observation is a consequence of Prop.\ref{prop:projection_on_algebraic_operators}.

\begin{property}
\label{prop:projection_on_algebraic_operators}
For any $L \subseteq \mathcal{L}$ and any $t_1$ and $t_2$ in $\mathbb{A}(L)^*$:
\[
\begin{array}{rclcrcl}
\muProjection_L(t_1 \cup t_2) & = & \muProjection_L(t_1) \cup \muProjection_L(t_2)
&~~~~&
\muProjection_L(t_1 \opStrictSeq t_2) & = & \muProjection_L(t_1) \opStrictSeq \muProjection_L(t_2)\\
\muProjection_L(t_1 \opInterleaving t_2) & = & \muProjection_L(t_1) \opInterleaving \muProjection_L(t_2)
&~~~~&
\muProjection_L(t_1 \opWeakSeq t_2) & = & \muProjection_L(t_1) \opStrictSeq \muProjection_L(t_2)
\end{array}
\]
\end{property}

\begin{proof} See \ref{app:section2}.
\end{proof}

Operators $\opStrictSeq$, $\opInterleaving$ and $\opWeakSeq$ being associative, this allows for the definition of repetition operators in the same manner as the Kleene star is defined over the classical concatenation. Given $\diamond \in \{ \opStrictSeq,~\opWeakSeq,~\opInterleaving \}$, the Kleene closure $^{\diamond *}$ is such that for any set of traces or multi-traces $M$ we have:
\[ M^{\diamond *}
= \bigcup_{\substack{j \in \mathbb{N}}} M^{\diamond j} \text{ with } 
\left\{
\begin{array}{l}
M^{\diamond 0}
= \{ \varepsilon \} \text{ or } \{ \varepsilon_L \} 
\\
\text{and }
M^{\diamond j}
=
M \diamond M^{\diamond (j-1)}
\text{ for } j > 0
\end{array}
\right.
\]

\subsection{Lifeline removal on multi-traces \label{sec:lf_tr_removal}}

Just as we defined a lifeline removal operator $\lifelineElim$ on interactions in Def.\ref{def:lifeline_removal_operator_on_interactions}, we now define a similar operator on multi-traces using the same notation in Def.\ref{def:lifeline_removal_operator_on_multitraces}.

\begin{definition}\label{def:lifeline_removal_operator_on_multitraces}
The function $\lifelineElim_H: \mathbb{M}(L) \rightarrow \mathbb{M}(L\setminus H)$, defined for any $L \in \mathcal{L}$ and $H \subset L$, is such that:
\[
\forall \mu \in \mathbb{M}(L),~
\lifelineElim_H(\mu) = (\mu_{|l})_{l \in L \setminus H}
\]
\end{definition}

It simply consists in removing the local traces of all the lifelines in the set $H \subseteq L$ of lifelines we want to remove.
The function $\lifelineElim_H$ is canonically extended to sets of multi-traces. 
Prop.\ref{lem:elimination_append} relates the $\multiAppend$ concatenation operation with the removal operation $\lifelineElim$.

\begin{property}
\label{lem:elimination_append}
For any $L \in \mathcal{L}$, $H \subseteq L$, $\mu \in \mathbb{M}(L)$ and $a \in \mathbb{A}(L)$:
\begin{itemize}
    \item if $\theta(a) \in H$ then $\lifelineElim_H(a \multiAppend \mu) = \lifelineElim_H(\mu)$ and $\lifelineElim_H(\mu \multiAppend a) = \lifelineElim_H(\mu)$
    \item else $\lifelineElim_H(a \multiAppend \mu) = a \multiAppend \lifelineElim_H(\mu)$ and $\lifelineElim_H(\mu \multiAppend a) = \lifelineElim_H(\mu) \multiAppend a$
\end{itemize}
\end{property}

\begin{proof}
Trivial.
\end{proof}

$\mathbb{M}(L)$ fitted with the set of algebraic operators $\mathcal{F} = \{\cup, \opStrictSeq, \opInterleaving, ~^{\opStrictSeq *},~^{\opInterleaving *} \}$ is an $\mathcal{F}$-algebra of signature $L \subseteq \mathcal{L}$.
For any $H \subseteq L$, the lifeline removal operator $\lifelineElim_H$ preserves the algebraic structures between the $\mathcal{F}$-algebras of signatures $L$ and $L \setminus H$.
In particular, this implies Prop.\ref{prop:multi_trace_elimination_preserves_sched}.

\begin{property}
\label{prop:multi_trace_elimination_preserves_sched}
For any $L \in \mathcal{L}$, any $H \subseteq L$, any $\mu_1$ and $\mu_2$ in $\mathbb{M}(L)$, for any $\diamond \in \{\cup,\opStrictSeq,\opInterleaving\}$, we have:
\[
\lifelineElim_H(\mu_1 \diamond \mu_2) = \lifelineElim_H(\mu_1) \diamond \lifelineElim_H(\mu_2)
\]
\end{property}

\begin{proof} See \ref{app:section2}.
\end{proof}

The result from Prop.\ref{prop:multi_trace_elimination_preserves_sched} can be extended to sets of multi-traces. Repetitions of those scheduling algebraic operators with their Kleene closures are also preserved by $\lifelineElim$.

\subsection{Multi-trace semantics of interactions \label{sec:semantics}}

 In Def.\ref{def:interaction_semantics}, we provide interactions with denotational semantics in terms of multi-traces. Generally speaking, a denotational semantics consists in associating each syntactic operator of the language with an algebraic counterpart in the carrier sets, in our case, in the set of multi-traces.

\begin{definition}
\label{def:interaction_semantics}
For any $L \subseteq \mathcal{L}$, we define $\sigma_L : \mathbb{I}(L) \rightarrow {\mathcal P}(\mathbb{M}(L))$ as follows:
{\small
\[
\begin{array}{lcr}
\sigma_L(\varnothing) = \{\varepsilon_L\}
&
~~~~~~~~~~
&
\forall~a \in \mathbb{A}(L),~\sigma_L(a) = \{ a \multiAppend \varepsilon_L \}
\end{array}
\]
}
For any $i_1$ and $i_2$ in $\mathbb{I}(L)$, any $(k,\diamond) \in \{(S,\opStrictSeq),~(W,\opStrictSeq),~(P,\opInterleaving)\}$:
{\small
\[
\begin{array}{c}
\begin{array}{lllclll}
\sigma_L(alt(i_1,i_2)) & = & \sigma_L(i_1) \cup \sigma_L(i_2)
&
~~~~~~~
&
\sigma_L(strict(i_1,i_2)) & = & \sigma_L(i_1) \opStrictSeq \sigma_L(i_2)
\\
\sigma_L(seq(i_1,i_2)) & = & \sigma_L(i_1) \opStrictSeq \sigma_L(i_2)
&
~~~~~~~
&
\sigma_L(par(i_1,i_2)) & = & \sigma_L(i_1) \opInterleaving \sigma_L(i_2)
\end{array}
\\
\sigma_L(loop_k(i)) = \sigma_L(i)^{\diamond *}
\end{array}
\]
}
\end{definition}

Def.\ref{def:interaction_semantics}
associates an operation carrying on multi-traces to each language operator ($alt$, $strict$, $seq$, $par$, $loop_k$) of $\mathbb{I}(L)$. For example, the operator $alt$ is interpreted as the union operator $\cup$ in sets, indicating that a multi-trace of an interaction of the form $alt(i_1,i_2)$ is either a multi-trace of $i_1$ or a multi-trace or $i_2$. So, Def.\ref{def:interaction_semantics} defines the set of multi-traces accepted by an interaction $i$ inductively on the term structure of $i$ in terms of operations on multi-traces. A curiosity of Def.\ref{def:interaction_semantics} is that the two operators $strict$ and $seq$ have identical semantic interpretations, that is the strict sequencing operation ``$\opStrictSeq$''. This follows directly from the observation made in Sec.\ref{sec:multitraces} that the two operators, strict and weak sequencing, are distinct for global traces but indistinguishable for multi-traces.

In this paper, we are interested in interaction-based runtime verification, that is, by the analysis of multi-traces, built by grouping all remote local traces of a DS. In this context, multi-trace semantics like that of Def.\ref{def:interaction_semantics} are clearly necessary. Nevertheless, in the literature, e.g. in \cite{UMLInteractionsMeetSM17,denotational_and_operational_semantics_for_interaction_languages_application_to_trace_analysis}, denotational semantics of interactions are classically given in terms of global traces. Syntactic operators ($\varnothing$, $a$, $alt$, $strict$, $seq$, $par$, $loop_k$) of the interaction language are them interpreted with operators on global traces rather than multi-traces (resp. $\{\varepsilon\}$, $\{a\}$, $\cup$, $\opStrictSeq$, $\opWeakSeq$, $\opInterleaving$, appropriate Kleene operators). Let us denote by $\sigma : \mathbb{I}(L) \rightarrow {\mathcal P}(\mathbb{A}(L)^*)$ such denotational semantics defined for global traces (and comprehensively defined in~\cite{denotational_and_operational_semantics_for_interaction_languages_application_to_trace_analysis}).
Thanks to Prop.\ref{prop:projection_on_algebraic_operators}, $\sigma_L$ and $\sigma$ are directly related via the projection operator $\muProjection_L$ as follows:
\[
\forall~i \in \mathbb{I}(L),~ \sigma_L(i) = \muProjection_L(\sigma(i))
\]
The fact that we map the weak sequencing $seq$ symbol of the interaction language to the strict sequencing operator $\opStrictSeq$ on multi-traces (and, by extension $loop_W$ to $~^{\opStrictSeq *}$) enables us to leverage Prop.\ref{prop:multi_trace_elimination_preserves_sched} while being proven sound thanks to Prop.\ref{prop:projection_on_algebraic_operators}.

Th.\ref{th:semantics_of_lifeline_removal_in_interactions} relates the semantics of an interaction $i \in \mathbb{I}(L)$ w.r.t.~that of $\lifelineElim_H(i)$ for any $L \subseteq \mathcal{L}$ and $H \subseteq L$.

\begin{theorem}
\label{th:semantics_of_lifeline_removal_in_interactions}
For any $L \subseteq \mathcal{L}$, any $H \subseteq L$ and any $i \in \mathbb{I}(L)$, we have:
\[
\sigma_{L\setminus H}(\lifelineElim_H(i)) = \lifelineElim_H(\sigma_L(i))
\]
\end{theorem}

\begin{proof}
See \ref{app:section2}.
\end{proof}

Thanks to $\sigma_L$ being related to the denotational style $\sigma$ of \cite{denotational_and_operational_semantics_for_interaction_languages_application_to_trace_analysis} via $\sigma_L = \muProjection_L \circ \sigma$, we can also leverage the operational style formulation of the semantics from \cite{denotational_and_operational_semantics_for_interaction_languages_application_to_trace_analysis}.
The definition of such semantics, in the style of Plotkin \cite{a_structural_approach_to_operational_semantics}, relies on two predicates denoted by $\downarrow$ and $\rightarrow$ and respectively called termination and execution relation.
For any interaction $i$, we say that $i$ terminates if $i \downarrow$, which implies that $\varepsilon_L \in \sigma_L(i)$. For any action $a$, any position $p$ with $i(p) = a$ and interaction $i'$, we say that $i'$ is a derivative \cite{derivative_Brzozowski64} of $i$ after the execution of $a@p$ iff $i \xrightarrow{a@p} i'$. In particular, this means that all the multi-traces of the form $a \multiAppend \mu'$ with $\mu' \in \sigma_L(i')$ are in $\sigma_L(i)$.

\begin{property}
\label{prop:operational_formulation_multitrace}
There exist a predicate $\downarrow \subseteq \mathbb{I}(L)$ and a relation $\rightarrow \subseteq \mathbb{I}(L) \times (\mathbb{A}(L)\times\{1,2\}^*) \times \mathbb{I}(L)$ such that:
\begin{itemize}
    \item for any interactions $i$ and $i'$, any action $a$ and position $p$, if $i \xrightarrow{a@p} i'$, then $a = i(p)$ and $i'$ is unique
    \item for any $i \in \mathbb{I}(L)$, $\sigma(i) \subset \mathbb{A}(L)^*$ is a set of traces such that $\sigma_L(i) = \muProjection_L(\sigma(i))$ and such that, for any $t \in \mathbb{A}(L)^*$, the statement $t \in \sigma(i)$ holds iff it can be proven using the following two rules:
\end{itemize}

\begin{center}
\begin{minipage}{3cm}
\begin{prooftree}
\AxiomC{$i \downarrow$\vphantom{$\xrightarrow{a@p}$}}
\UnaryInfC{$\varepsilon \in \sigma(i)$}
\end{prooftree}
\end{minipage}
\begin{minipage}{4cm}
\begin{prooftree}
\AxiomC{$\mu \in \sigma(i')$}
\AxiomC{$i \xrightarrow{a@p} i'$}
\BinaryInfC{$a.t \in \sigma(i)$}
\end{prooftree}
\end{minipage}
\end{center}

\end{property}

\begin{proof}
Definitions of $\downarrow$ and $\rightarrow$ as well as proof for $\sigma$ on global trace available in \cite{denotational_and_operational_semantics_for_interaction_languages_application_to_trace_analysis}. 
In \ref{app:opsem_lfrem}, we give the set of predicates ($\downarrow$, $\collidesLf$, $\isPruneOf{l}$, $\xrightarrow{a@p}$) defining the operational semantics, with the notations from~\cite{denotational_and_operational_semantics_for_interaction_languages_application_to_trace_analysis} and with additional annotations for position traceability. The annotation of the positions $p$ comes at no cost in the inductive definition of $\rightarrow$ from \cite{denotational_and_operational_semantics_for_interaction_languages_application_to_trace_analysis}.  The reader is referred to \cite{denotational_and_operational_semantics_for_interaction_languages_application_to_trace_analysis} for comments and justifications of the various rules.
\end{proof}

The algebraic characterization of Def.\ref{def:interaction_semantics} underpins results involving the use of the $\lifelineElim$ function (namely Th.\ref{th:semantics_of_lifeline_removal_in_interactions}) while the operational characterization of Prop.\ref{prop:operational_formulation_multitrace} is required in the definition and proof of algorithms that involves the execution of interactions.
In this paper, the definition of operational semantics is not essential for the paper and is therefore only given in \ref{app:opsem_lfrem}. It suffices to consider their existence (Prop.\ref{prop:operational_formulation_multitrace}). 
In addition, we will use the notation $i \xrightarrow{a} i'$ (resp. $i\not\xrightarrow{a}$) when there exists (resp. does not exist) a position $p \in pos(i)$ and an interaction $i'$ s.t. $i \xrightarrow{a@p} i'$. 
In the following, for the sake of concision, we may interchangeably use the notations $i \xrightarrow{a} i'$ and $i \xrightarrow{a@p} i'$ depending on whether or not the knowledge of the specific position $p$ is useful.

The $\rightarrow$ execution relation of the operational formulation also manifests interesting properties w.r.t.~the lifeline removal operator $\lifelineElim$.
Prop.\ref{prop:execution_and_lifeline_removal} then relates follow-up interactions from interactions $i$ and $\lifelineElim_H(i)$ that result from the execution of the same action $a$ at the same position $p$.
Roughly speaking, this property refers to the preservation of ``executions'' by the lifeline removal operator, provided the lifeline of the executed action is not removed.

\begin{property}
\label{prop:execution_and_lifeline_removal}
For any $L \subset \mathcal{L}$, any interaction $i \in \mathbb{I}(L)$, any set of lifelines $ H \subset L$, any action $a \in \mathbb{A}(L)$ s.t.~ $\theta(a) \not\in H$ and any position $p \in pos(i)$:
\[
(\exists~i' \in \mathbb{I}(L),~i \xrightarrow{a@p} i') ~\Rightarrow~ (\lifelineElim_H(i) \xrightarrow{a@p} \lifelineElim_H(i'))
\]
\end{property}

\begin{proof}
This is implied by Prop.\ref{prop:positions_and_lifeline_removal} and via reasoning by structural induction on the rules of the structural operational semantics from \cite{denotational_and_operational_semantics_for_interaction_languages_application_to_trace_analysis}. See \ref{app:opsem_lfrem} for details.
\end{proof}

\section{Partial observation: solving membership for multi-trace prefixes\label{sec:multipref}}

\subsection{Partial observation and multi-prefixes\label{ssec:partial_obs_multipref}}

Offline Runtime Verification (ORV) \cite{runtime_verification_for_decentralised_and_distributed_systems_FrancalanzaPS2018} designates a two steps process which consists 
\begin{enumerate}
    \item in collecting (multi-)traces that are observations of some executions of a system and 
    \item in analyzing these (multi-)traces against a formal specification of the system to detect potential deviations from its expected behaviors.
\end{enumerate}
In this aspect, ORV is related to passive testing \cite{passive_conformance_testing_of_service_choreographies} (in so far as passive tests amounts to a-posteriori trace analysis) and the membership problem \cite{pattern_matching_and_membership_for_hierarchical_message_sequence_charts} (i.e.~solving whether or not a trace belongs to a set of traces that corresponds to the semantics of a specification).

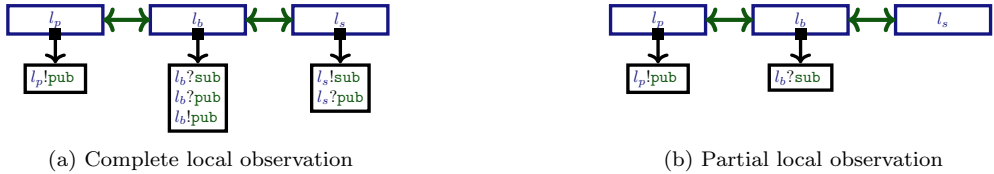
\begin{figure}[h]
    \centering
    
\begin{subfigure}[b]{.475\textwidth}
    \centering
    \scalebox{.625}{\begin{tikzpicture}
\tikzstyle{subsys}=[draw,hibou_col_lf,rectangle,line width=2pt,minimum height=.6cm,minimum width=2cm,inner sep=0,outer sep=0]
\tikzstyle{element_blk}=[draw,black,rectangle,line width=2pt,minimum height=2cm,minimum width=2cm,inner sep=0,outer sep=0]
\node[subsys] (A) at (-3,0) {
    $\hlf{l_p}$
};
\node[subsys] (B) at (0,0) {
    $\hlf{l_b}$
};
\node[subsys] (C) at (3,0) {
    $\hlf{l_s}$
};
\draw[<->,line width=2.5pt,hibou_col_ms] (A) -- (B);
\draw[<->,line width=2.5pt,hibou_col_ms] (C) -- (B);
\node[rectangle, fill=black] (Alog) at (A.south) {};
\node[rectangle, fill=black] (Blog) at (B.south) {};
\node[rectangle, fill=black] (Clog) at (C.south) {};
\node[draw,black,rectangle,line width=2pt,align=left,below=.5cm of Alog] (Atra) {
$\hlf{l_p}!\hms{\mathtt{pub}}$
};
\node[draw,black,rectangle,line width=2pt,align=left,below=.5cm of Blog] (Btra) {
$\hlf{l_b}?\hms{\mathtt{sub}}$\\
$\hlf{l_b}?\hms{\mathtt{pub}}$\\
$\hlf{l_b}!\hms{\mathtt{pub}}$
};
\node[draw,black,rectangle,line width=2pt,align=left,below=.5cm of Clog] (Ctra) {
$\hlf{l_s}!\hms{\mathtt{sub}}$\\
$\hlf{l_s}?\hms{\mathtt{pub}}$
};
\draw[->,line width=2pt,black] (Alog) -- (Atra);
\draw[->,line width=2pt,black] (Blog) -- (Btra);
\draw[->,line width=2pt,black] (Clog) -- (Ctra);
\end{tikzpicture}}
    \caption{Complete local observation\label{fig:observation_complete}}
\end{subfigure} 
\begin{subfigure}[b]{.475\textwidth}
    \centering
    \scalebox{.625}{\begin{tikzpicture}
\tikzstyle{subsys}=[draw,hibou_col_lf,rectangle,line width=2pt,minimum height=.6cm,minimum width=2cm,inner sep=0,outer sep=0]
\tikzstyle{element_blk}=[draw,black,rectangle,line width=2pt,minimum height=2cm,minimum width=2cm,inner sep=0,outer sep=0]
\node[subsys] (A) at (-3,0) {
    $\hlf{l_p}$
};
\node[subsys] (B) at (0,0) {
    $\hlf{l_b}$
};
\node[subsys] (C) at (3,0) {
    $\hlf{l_s}$
};
\draw[<->,line width=2.5pt,hibou_col_ms] (A) -- (B);
\draw[<->,line width=2.5pt,hibou_col_ms] (C) -- (B);
\node[rectangle, fill=black] (Alog) at (A.south) {};
\node[rectangle, fill=black] (Blog) at (B.south) {};
\node[draw,black,rectangle,line width=2pt,align=left,below=.5cm of Alog] (Atra) {
$\hlf{l_p}!\hms{\mathtt{pub}}$
};
\node[draw,black,rectangle,line width=2pt,align=left,below=.5cm of Blog] (Btra) {
$\hlf{l_b}?\hms{\mathtt{sub}}$
};
\draw[->,line width=2pt,black] (Alog) -- (Atra);
\draw[->,line width=2pt,black] (Blog) -- (Btra);
\node[opacity=0,draw,black,rectangle,line width=2pt,align=left,below=.5cm of Blog] (BtraPhantom) {
$\hlf{l_b}?\hms{\mathtt{sub}}$\\
$\hlf{l_b}?\hms{\mathtt{pub}}$\\
$\hlf{l_b}!\hms{\mathtt{pub}}$
};
\end{tikzpicture}}
    \caption{Partial local observation\label{fig:observation_partial}}
\end{subfigure} 

    \caption{Complete and partial observation for the example from Fig.\ref{fig:colocalization_schema_discrete}\label{fig:observation}}
\end{figure}

Yet, in the first step of collecting the input traces that are to be fed to the analysis algorithm, we are not likely guaranteed to observe the executed behavior in its entirety.
Indeed, it might be so that some subsystems cannot be equipped with observation devices. Moreover, due to the absence of synchronization between the local observations, the different logging processes might cease at uncorrelated moments.

For example, if we consider our example behavior from Fig.\ref{fig:colocalization_schema}, depending on the (monitoring) architecture which allows the observation of the corresponding DS, the observation could yield, among others: the multi-trace from Fig.\ref{fig:observation_complete} or the one from Fig.\ref{fig:observation_partial}.
While the former corresponds to a complete observation of our behavior, for the latter, subsystem $\hlf{l_s}$ was not observed and observation ceased too early on subsystem $\hlf{l_b}$.
Indeed, if the observer on $\hlf{l_b}$ kept observing the execution after $\hlf{l_b}?\hms{\mathtt{sub}}$, it would, after an unspecified amount of time, have observed $\hlf{l_b}?\hms{\mathtt{pub}}$ and $\hlf{l_b}!\hms{\mathtt{pub}}$.

In this paper, the notion of partial observation of a behavior is only limited to the two aforementioned cases i.e., \textbf{(1)} some subsystems not being observed at all or \textbf{(2)} some subsystems not being observed long enough.
This notion of partial observation coincides with the observation of prefixes (in the sense of multi-traces) of the multi-traces that characterize the corresponding behaviors completely. We call such prefixes {\em multi-prefixes} and we define them in Def.\ref{def:multitrace_prefix}.

\begin{definition}
\label{def:multitrace_prefix}
For any $L \subseteq \mathcal{L}$ and any $\mu \in \mathbb{M}(L)$, the set of all {\em multi-prefixes} of $\mu$ is denoted by $\overline{\mu} = \{ \mu' ~|~\mu' \in \mathbb{M}(L), \forall~l \in L, \mu'_{|l} \in \overline{\mu_{|l}}\}$.
\end{definition}

If we denote by $\mu$ and $\mu'$ the multi-traces resp.~represented on Fig.\ref{fig:observation_complete} and Fig.\ref{fig:observation_partial}, we have that $\mu' \in \overline{\mu}$ because for each of the three lifelines $\hlf{l} \in \{\hlf{l_p},\hlf{l_b},\hlf{l_s}\}$ we have $\mu'_{|\hlf{l}} \in \overline{\mu_{|\hlf{l}}}$.

Let us remark that $\mu'$ is a prefix of a multi-trace $\mu$ (i.e.~$\mu' \in \overline{\mu}$) iff there exists $\mu''$ verifying $\mu' \opStrictSeq \mu'' = \mu$.
We extend the $\overline{\phantom{\mu}}$ notation to sets of multi-traces as follows: $\overline{M} = \bigcup_{\mu \in M} \overline{\mu}$ for any $M \subseteq \mathbb{M}(L)$.

\begin{property}
\label{prop:multiprefix_and_projection}
We have:
\[
\begin{array}{lrcl}
\forall~L \subseteq \mathcal{L},~\forall~t \in \mathbb{A}(L)^*,
&
\muProjection_L(\overline{t})
&
\subseteq
&
\overline{\muProjection_L(t)}
\\
\exists~L \subseteq \mathcal{L},~\exists~t \in \mathbb{A}(L)^*,
&
\overline{\muProjection_L(t)}
&
\not\subset 
&
\muProjection_L(\overline{t})
\end{array}
\]
\end{property}

\begin{proof}
For the first item, let us reason by induction:
\[
\begin{array}{rcl}
\muProjection_L(\overline{\varepsilon})
&
=
&
\muProjection_L(\{\varepsilon\}) = \{\varepsilon_L\} = \overline{\{\varepsilon_L\}} = \overline{\muProjection_L(\varepsilon)}
\\
\muProjection_L(\overline{a.t})
&
=
&
\muProjection_L(\{\varepsilon\} \cup a.\overline{t}) = \{\varepsilon_L \} \cup  a \multiAppend \muProjection_L(\overline{t}) = \{\varepsilon_L \} \cup  a \multiAppend \overline{\muProjection_L(t)} \\
& = & \overline{a \multiAppend \muProjection_L(t)} = \overline{\muProjection_L(a.t)} 
\end{array}
\]
For the second item, let us consider $L = \{\hlf{l_1},\hlf{l_2}\}$ and $t = \hlf{l_1}!\hms{m}.\hlf{l_2}?\hms{m}$. Then:
\[
\begin{array}{rcl}
\muProjection_L(\overline{t})
&
=
&
\muProjection_L(\{\varepsilon,~~\hlf{l_1}!\hms{m},~~\hlf{l_1}!\hms{m}.\hlf{l_2}?\hms{m}\})
=
\{\varepsilon_L,~~(\hlf{l_1}!\hms{m},~\varepsilon),~~(\hlf{l_1}!\hms{m},~\hlf{l_2}?\hms{m})\}
\\
\overline{\muProjection_L(t)}
&
=
&
\overline{(\hlf{l_1}!\hms{m},~\hlf{l_2}?\hms{m})}
=
\{\varepsilon_L,~~(\hlf{l_1}!\hms{m},~\varepsilon),~~(\varepsilon,~\hlf{l_2}?\hms{m}),~~(\hlf{l_1}!\hms{m},~\hlf{l_2}?\hms{m})\}
\end{array}
\]
\end{proof}

Prop.\ref{prop:multiprefix_and_projection}, 
points out that not all multi-prefixes of a multi-trace correspond to the projection of a prefix of a corresponding global trace. With Prop.\ref{prop:lifeline_elimination_and_multi_prefixes}, we relate $\lifelineElim$ to multi-prefixes.

\begin{property}
\label{prop:lifeline_elimination_and_multi_prefixes}
For any multi-trace $\mu \in \mathbb{M}(L)$, any set of multi-traces $M \subseteq \mathbb{M}(L)$ and any set of lifelines $H \subset L$:
\[
\begin{array}{c}
\lifelineElim_H(\overline{M}) = \overline{\lifelineElim_H(M)}\\
\text{and}\\
\left(
\begin{array}{c}
( \lifelineElim_H(\mu) \in \lifelineElim_H(M) )~
\wedge~( \forall~l \in H,~ \mu_{|l} = \varepsilon )
\end{array}
\right)
\Rightarrow 
( \mu \in \overline{M} )
\end{array}
\]
\end{property}

\begin{proof}
The first point in trivial. 
For the second, if $\lifelineElim_H(\mu) \in \lifelineElim_H(M)$ this means that $(\mu_{|l})_{l \in L \setminus H} \in \lifelineElim_H(M)$. 
Then, there exists a multi-trace $\mu^0$ in $M$ and $|H|$ trace components $(t_l^0 \in \mathbb{A}_l^*)_{l \in H}$
such that $\forall~l \in L \setminus H$, $\mu_{|l}^0 = \mu_{|l}$ and $\forall~l \in H$, $\mu_{|l}^0 = t_l^0$. 
Let us then consider the multi-trace $\mu^1$ such that $\forall~l \in L \setminus H$, $\mu_{|l}^1 = \varepsilon$ and $\forall~l \in H$, $\mu_{|l}^1 = t_l^0$. Then, by construction, of $\mu^0$ and $\mu^1$, we get $\mu^0 = \mu \opStrictSeq \mu^1$ and hence $\mu$ is a prefix (in the sense of multi-traces) of $\mu^0 \in M$. Therefore $\mu \in \overline{M}$.
\end{proof}

\subsection{Offline Runtime Verification from interactions}

In \cite{a_small_step_approach_to_multi_trace_checking_against_interactions} we proposed an algorithm to check whether or not a multi-trace $\mu$ belongs to the semantics of an interaction $i$.
Its key principle is to find a globally ordered behavior specified by $i$ (via the $\rightarrow$ execution relation) that matches $\mu$ i.e., an accepted global trace $t$ that can be projected into $\mu$.

To do so, it relies on a rule $(i,a \multiAppend \mu') \leadsto (i',\mu')$ with $i \xrightarrow{a} i'$ in which an action $a$ is simultaneously consuming from the multi-trace and executed in the interaction.
Because these actions always match the head of a local trace component of the multi-trace, if successive applications of this rule result in emptying the multi-trace (i.e., reaching the empty multi-trace), then the concatenation of these actions yields a global trace that projects into the initial multi-trace.
Moreover, because these actions correspond to actions that are successively executed on the initial interaction and its successive derivatives, their concatenation corresponds to a prefix (in the sense of global traces) of a trace in the semantics of the initial interaction. In addition, following Prop.\ref{prop:operational_formulation_multitrace}, if the last derivative terminates, this means that their concatenation belongs to the semantics of the initial interaction.

\begin{figure}[h]
    \centering

\begin{subfigure}{.975\textwidth}
\centering
\scalebox{.725}{\input{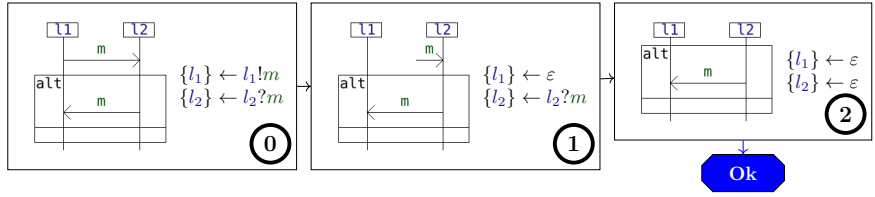}}
\caption{Completely observed correct behavior}
\label{fig:expla_ana_complete}
\end{subfigure}

\vspace*{.25cm}

\begin{subfigure}{.975\textwidth}
\centering
\scalebox{.725}{\input{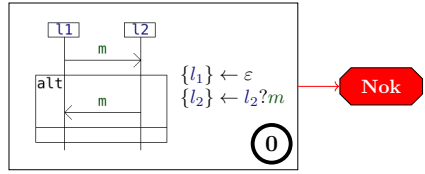}}
\caption{Partially observed correct behavior}
\label{fig:expla_ana_partial}
\end{subfigure}
    
    \caption{Principle and limitations of the algorithm from \cite{a_small_step_approach_to_multi_trace_checking_against_interactions}}
    \label{fig:expla_ana}
\end{figure}

Fig.\ref{fig:expla_ana_complete} illustrates the principle of that algorithm on a simple example. It consists of a proof which shows that the multi-trace $\mu_0 = (\hlf{l_1}!\hms{m},~\hlf{l_2}?\hms{m})$ over $L = \{\hlf{l_1},\hlf{l_2}\}$ belongs to the semantics of:\\
\noindent\centerline{
$
i_0 = seq(
          strict(\hlf{l_1}!\hms{m},\hlf{l_2}?\hms{m}),
          alt(
              strict(\hlf{l_2}!\hms{m},\hlf{l_1}?\hms{m}),
              \varnothing
          )
      )
$
}
\\
\noindent On Fig.\ref{fig:expla_ana_complete}, each square annotated with a circled number (e.g. \textcircled{0}) contains an interaction drawn on the left and a multi-trace on the right.
Fig.\ref{fig:expla_ana_complete} represents a graph in which vertices correspond to tuples $(i,\mu)$ (here, the three annotated squares corresponding to vertices $(i_0,\mu_0)$, $(i_1,\mu_1)$ and $(i_2,\mu_2)$) and transitions correspond to elementary steps of the algorithm from \cite{a_small_step_approach_to_multi_trace_checking_against_interactions}. The exploration of this graph, from the starting vertex $(i_0,\mu_0)$, allows determining whether or not $\mu_0 \in \sigma_L(i_0)$. In this example, the fact that vertex $(i_2,\mu_2)$ is reachable from $(i_0,\mu_0)$ and that $i_2$ accepts the empty multi-trace proves that $\mu_0 \in \sigma_L(i_0)$. This directly follows from Prop.\ref{prop:operational_formulation_multitrace}.

This algorithm solves the membership problem for $\sigma_L(i)$, for any $i \in \mathbb{I}(L)$. However, its practical goal is to analyze multi-traces that correspond to executions of a DS under test. But, as discussed in Sec.\ref{ssec:partial_obs_multipref}, in all generality, we have no guarantee of having observed an execution in its entirety. For that reason, following the definition of partial observation in Sec.\ref{ssec:partial_obs_multipref}, a sound Offline Runtime Verification algorithm from interaction should rather solve the membership problem for $\overline{\sigma_L(i)}$.
This would avoid false negatives arising from interpreting a consequence of partial observation as the system under test deviating from its specification.

However, adapting the algorithm from \cite{a_small_step_approach_to_multi_trace_checking_against_interactions} to that new problem is not trivial. 
Intuitively, one could do away with the verification of termination so that we may be able to identify prefixes of accepted behaviors.
Yet, those prefixes would then only correspond to projections of prefixes of a global trace in $\sigma(i)$. However, Prop.\ref{prop:multiprefix_and_projection} states not all multi-prefixes of a multi-trace correspond to the projection of a prefix of a corresponding global trace.
This is illustrated on Fig.\ref{fig:expla_ana}, which corresponds to the counter example from the proof of Prop.\ref{prop:multiprefix_and_projection}.
On Fig.\ref{fig:expla_ana_partial}, we can see that although $(\varepsilon,~\hlf{l_2}?\hms{m})$ is a multi-prefix of $(\hlf{l_1}!\hms{m},~\hlf{l_2}?\hms{m})$, the algorithm cannot progress 
to state that the multi-trace is an accepted multi-prefix.

\begin{figure}[h]
    \centering
    \scalebox{.725}{\input{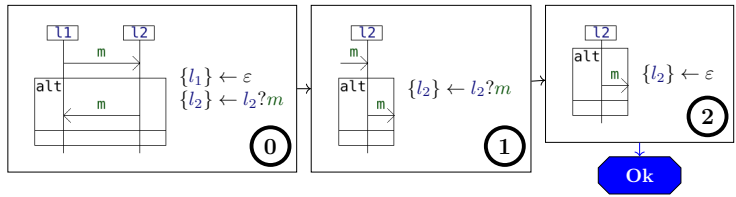}}
    \caption{Application of an algorithm leveraging $\lifelineElim$}
    \label{fig:expla_ana_hiding}
\end{figure}

In the following, we leverage the use of $\lifelineElim$, and in particular Th.\ref{th:semantics_of_lifeline_removal_in_interactions} to define an algorithm solving membership for any $\overline{\sigma_L(i)}$.
To be able to pursue the analysis once a situation such that the one from Fig.\ref{fig:expla_ana_partial} is reached, we use $\lifelineElim$ to remove lifelines on which no action is observed from both the interaction $i$ and the multi-trace $\mu$ that constitute the vertex $(i,\mu)$ of the analysis graph.
The two types of partial observation (unobserved subsystems and early interruption of observation) are approached in the same manner, noting in particular that an empty local trace can be seen both as missing and incomplete.
Fig.\ref{fig:expla_ana_hiding} illustrates the use of this new algorithm on the example from Fig.\ref{fig:expla_ana_partial}.
Once $\lifelineElim_{\{\hlf{l_1}\}}$ is applied (from \textcircled{0} to \textcircled{1}), the execution of $\hlf{l_2}?\hms{m}$ (from \textcircled{1} to \textcircled{2}) becomes possible, and the multi-trace $\mu_0$ is identified as belonging to $\overline{\sigma_L(i_0)}$.

\subsection{Analysis graph\label{ssec:graph}}

Our algorithm relies on exploring an analysis graph, which we define in Def.\ref{def:search_graph}.
Let us consider a set of vertices of the form $(i,\mu) \in \mathbb{I}(L) \times \mathbb{M}(L)$, for any $L \subseteq \mathcal{L}$, as well as a specific verdict vertex $\macroOKVerdict$.
Let us denote:
\[\mathbb{V} = \{\macroOKVerdict\} \cup ( \; \bigcup_{ L \subseteq {\mathcal{L}}}  \mathbb{I}(L) \times \mathbb{M}(L) \; )\]
We consider a graph $\mathbb{G}$ whose vertices are in $\mathbb{V}$ and whose set of arcs is defined by three binary relations:
\begin{itemize}
    \item $\graphRelationExec \subseteq \mathbb{V} \times \mathbb{V}$, denoting an ``execution'' rule that transforms a vertex $(i,a \multiAppend \mu)$ into a vertex $(i',\mu)$ by consuming action $a$ at the head of the $(a \multiAppend \mu)_{|\theta(a)}$ local component of the multi-trace and executing it on the interaction $i$ via the choice of a position $p$ such that $i \xrightarrow{a@p} i'$
    \item $\graphRelationHide \subseteq \mathbb{V} \times \mathbb{V}$, denoting a ``removal'' rule that transforms a vertex $(i,\mu)$, given a non empty set $H \subseteq  \{l \in L~|~ \mu_{|l} = \varepsilon\}$  into $(\lifelineElim_H(i),\lifelineElim_H(\mu))$
    \item $\graphRelationPass \subseteq \mathbb{V} \times \mathbb{V}$, denoting a ``ok'' rule that links all vertices of the form $(i,\varepsilon_L)$ to $\macroOKVerdict$
\end{itemize}

These three relations are defined as follows.

\begin{definition}
\label{def:search_graph}
We define $\graphRelationExec$, $\graphRelationHide$ and $\graphRelationPass$, three binary relations over $\mathbb{V}$ as follows: for any $L \subset \mathcal{L}$, any vertices $(i,\mu) \in \mathbb{I}(L) \times \mathbb{M}(L)$ and $(i',\mu') \in \mathbb{V}$, we have:
\[
\begin{array}{lcl}
(i,\mu)
\graphRelationPass
Ok
&
\Leftrightarrow 
&
\mu = \varepsilon_L
\\ 
(i,\mu)
\graphRelationHide
(i',\mu')
&
\Leftrightarrow 
&
\exists~H \in \mathcal{P}(L) \setminus \{\emptyset,L\},
\left\{
\begin{array}{ll}
&\forall~l \in H,~\mu_{|l} = \varepsilon\\
\wedge&
\mu' = \lifelineElim_{H}(\mu)\\
\wedge&
i' = \lifelineElim_{H}(i)
\end{array}
\right.
\\ 
(i,\mu)
\graphRelationExec
(i',\mu')
&
\Leftrightarrow 
&
\exists~a \in \mathbb{A}(L),~
(\mu = a \multiAppend \mu')\wedge (i \xrightarrow{a} i')
\\ 
\end{array}
\]

\noindent
$\mathbb{G} = (\mathbb{V},\leadsto)$ is the graph where $ \leadsto = (\graphRelationPass \cup \graphRelationHide \cup \graphRelationExec)$ defines the set of 
 its arcs.
\end{definition}

Given an origin vertex $(i,\mu)$ defined up to $L$, $\graphRelationExec$ and $\graphRelationHide$ specify arcs of the form $(i,\mu) \leadsto (i',\mu')$ with $i'$ and $\mu'$ defined on the same set of lifelines, which is either $L$ for $\graphRelationExec$ or $L \setminus H$ for a certain $H$ for $\graphRelationHide$.
The application of $\graphRelationExec$ corresponds to the simultaneous consumption of an action at the head of a component of $\mu$ and the execution of a matching action in $i$ while the application of $\graphRelationHide$ corresponds to the removal (from both $i$ and $\mu$) of lifelines $l$ for which the corresponding component $\mu_{|l}$ is empty.


Prop.\ref{prop:confluence_lifeline_removal} states a confluence property of graph $\mathbb{G}$ that is related to the commutative applications of $\graphRelationHide$ and $\graphRelationExec$.
It states that if, from a given vertex $(i,\mu)$, we can reach $\macroOKVerdict$ by any given means, then, if we can also apply $\graphRelationHide$ w.r.t.~a lifeline $h$ s.t., $\mu_{|h} = \varepsilon$ so that $(i, \mu) \graphRelationHide (\lifelineElim_{\{h\}}(i),\lifelineElim_{\{h\}}(\mu))$, then we can also reach $\macroOKVerdict$ from $(\lifelineElim_{\{h\}}(i),\lifelineElim_{\{h\}}(\mu))$.

\begin{property}
\label{prop:confluence_lifeline_removal}
For any $L \subset \mathcal{L}$, any $i \in \mathbb{I}(L)$, any $\mu \in \mathbb{M}(L)$ and any $h \in L$ s.t.~$\mu_{|h} = \varepsilon$, we have\footnote{$\overset{*}{\leadsto}$ is the reflexive and transitive closure of $\leadsto$ as defined in Sec.\ref{sec:prelim}.}:

\begin{prooftree}
\AxiomC{$(i, \mu) \overset{*}{\leadsto} \macroOKVerdict$}
\AxiomC{$(i, \mu) \graphRelationHide (\lifelineElim_{\{h\}}(i),\lifelineElim_{\{h\}}(\mu))$}
\LeftLabel{}
\RightLabel{}
\BinaryInfC{$(\lifelineElim_{\{h\}}(i),\lifelineElim_{\{h\}}(\mu)) \overset{*}{\leadsto} \macroOKVerdict$}
\end{prooftree}
\end{property}

\begin{proof}
Let us suppose $(i, \mu) \overset{*}{\leadsto} \macroOKVerdict$ with $\mu_{|h} = \varepsilon$ and  $(i, \mu) \graphRelationHide (\lifelineElim_{\{h\}}(i),\lifelineElim_{\{h\}}(\mu))$.

We will reason by induction on the derivation length of $(i, \mu) \overset{*}{\leadsto} \macroOKVerdict$.

Let us analyze according to the nature of the first derivation step:

\begin{itemize}
\item $(i, \mu) \overset{*}{\leadsto} \macroOKVerdict$, with $(i, \mu) \graphRelationPass \macroOKVerdict$ as first step. This means that $\mu = \varepsilon_L$. Therefore $\lifelineElim_{\{h\}}(\mu) = \varepsilon_{L \setminus \{h\}}$ and  
$(\lifelineElim_{\{h\}}(i),\lifelineElim_{\{h\}}(\mu))  \graphRelationPass \macroOKVerdict$. As a result, $(\lifelineElim_{\{h\}}(i),\lifelineElim_{\{h\}}(\mu)) \overset{*}{\leadsto} \macroOKVerdict$.
\item $(i, \mu) \overset{*}{\leadsto} \macroOKVerdict$, with $(i, \mu) \graphRelationHide (\lifelineElim_H(i),\lifelineElim_H(\mu))$ as first step, then:
\begin{itemize}
    \item The conclusion is straightforward if $H = \{h\}$.
    \item If $h \in H$ (with $H$ not reduced to the singleton $\{h\}$), then:\\
$(\lifelineElim_{\{h\}}(i),
\lifelineElim_{\{h\}}(\mu)) \graphRelationHide 
(\lifelineElim_{H \setminus \{h\}}(\lifelineElim_{\{h\}}(i)),
\lifelineElim_{H \setminus \{h\}}(\lifelineElim_{\{h\}}(\mu)))$\\
and because $\lifelineElim_{H \setminus \{h\}} \circ \lifelineElim_{\{h\}} = \lifelineElim_{H}$ (Prop.\ref{prop:commut_hide}), we have by transitivity:\\ $(\lifelineElim_{\{h\}}(i),
\lifelineElim_{\{h\}}(\mu)) \overset{*}{\leadsto} \macroOKVerdict$.
    \item If $h \not\in H$, we remark that:\\
        $(\lifelineElim_{\{h\}}(i),\lifelineElim_{\{h\}}(\mu)) \graphRelationHide (\lifelineElim_{H}(\lifelineElim_{\{h\}}(i)),\lifelineElim_{H}(\lifelineElim_{\{h\}}(\mu)))$\\
        and because we have decremented the derivation length by applying a first time $\graphRelationHide$, we can apply the induction hypothesis so that\\
        $(\lifelineElim_{H \cup \{h\}}(i),\lifelineElim_{H \cup \{h\}}(\mu)) \overset{*}{\leadsto} \macroOKVerdict$\\
        we then conclude by transitivity
\end{itemize}
\item $(i, \mu) \overset{*}{\leadsto} \macroOKVerdict$, with $(i, \mu) \graphRelationExec
(i',\mu')$ as first step. This means that there exists $a$ in $\mathbb{A}(L)$ such that $\mu =  a \multiAppend \mu'$ and $i \xrightarrow{a} i'$.  We have $(i', \mu') \overset{*}{\leadsto} \macroOKVerdict$ with a derivation length to $\macroOKVerdict$ shorter than that of $(i,\mu)$. As $\mu_{|h} = \varepsilon$,  $\mu'_{|h} = \varepsilon$ and we have $(i',\mu') \graphRelationHide (\lifelineElim_{\{h\}}(i'),\lifelineElim_{\{h\}}(\mu'))$. By induction hypothesis, we obtain 
$(\lifelineElim_{\{h\}}(i'),\lifelineElim_{\{h\}}(\mu')) \overset{*}{\leadsto} \macroOKVerdict 
 (*)$. From $i \xrightarrow{a} i'$, we can deduce by Prop.\ref{prop:execution_and_lifeline_removal} that $ \lifelineElim_{\{h\}}(i) \xrightarrow{a} 
\lifelineElim_{\{h\}}(i')$, the action execution taking place at the same position in $i$ and $\lifelineElim_{\{h\}}(i)$. Therefore 
$(\lifelineElim_{\{h\}}(i),\lifelineElim_{\{h\}}(\mu)) \graphRelationExec (\lifelineElim_{\{h\}}(i'),\lifelineElim_{\{h\}}(\mu')) (**)$. From $(*)$ and $(**)$, we have $(\lifelineElim_{\{h\}}(i),\lifelineElim_{\{h\}}(\mu)) \overset{*}{\leadsto} \macroOKVerdict $.
\end{itemize}
\end{proof}

\subsection{Conformity and verdict\label{ssec:algo}}

We show in Th.\ref{th:multipref_equates_pass} that proving $\mu \in \overline{\sigma_L(i)}$, amounts to exhibiting a path in $\mathbb{G}$ starting from $(i,\mu)$ and leading to the verdict $\macroOKVerdict$. 

\begin{theorem}
\label{th:multipref_equates_pass}
For any $i \in \mathbb{I}(L)$ and any $\mu \in \mathbb{M}(L)$:
\[
\left(
\begin{array}{c}
\mu \in \overline{\sigma_L(i)}\vphantom{\overset{*}{\leadsto}}
\end{array}
\right)
\Leftrightarrow 
\left(
\begin{array}{c}
(i,\mu) \overset{*}{\leadsto} \macroOKVerdict
\end{array}
\right)
\]
\end{theorem}

\begin{proof}
To make expressions easier to read, we use the notation $\overline{\sigma_L}(i)$ to denote the set $\overline{\sigma_L(i)}$.
Let us reason by induction on the measure $|(i,\mu)|$ where the measure on vertices $v\in\mathbb{V}$ of $\mathbb{G}$ is defined as follows: 
\[
|v|=
\left\{
\begin{array}{ll}
     0 & \text{if }v = \macroOKVerdict ~\text{   (by convention)} \\
     |\mu|+|L| + 1 & \text{if }v = (i,\mu)\text{ with } L\subseteq \mathcal{L},\ i \in \mathbb{I}(L),\ \mu \in \mathbb{M}({L})
\end{array}
\right.
\]

\noindent $\bullet$ Case $\mu = \varepsilon_L$, the result follows directly from the facts that for any interaction $i$, we have $\mu \in  \overline{\sigma_L}(i)$ and $(i,\varepsilon_L) \graphRelationPass \macroOKVerdict$.

\noindent $\bullet$ Case $\mu \neq \varepsilon_L$ such that there exists $h$ in $L$ verifying $\mu_{|h} = \varepsilon$

Then we have $(i,\mu) \graphRelationHide (\lifelineElim_{\{h\}}(i),\lifelineElim_{\{h\}}(\mu))$.
\begin{itemize}
    \item[$\Rightarrow$] Let us suppose $\mu \in \overline{\sigma_L}(i)$.
\[
\begin{array}{rclr}
\mu \in \overline{\sigma_L}(i)
&
\Rightarrow
&
\lifelineElim_{\{h\}}(\mu) \in \lifelineElim_{\{h\}}(\overline{\sigma_L}(i))
&
\text{\scriptsize application of $\lifelineElim_{\{h\}}$}
\\
&
\Rightarrow
&
\lifelineElim_{\{h\}}(\mu) \in \overline{\lifelineElim_{\{h\}}(\sigma_L(i))}
&
~~\text{\scriptsize 1st point of Prop.\ref{prop:lifeline_elimination_and_multi_prefixes}}
\\
&
\Rightarrow
&
\lifelineElim_{\{h\}}(\mu) \in \overline{\sigma_{L \setminus \{h\}}(\lifelineElim_{\{h\}}(i))}
&
\text{\scriptsize Th.\ref{th:semantics_of_lifeline_removal_in_interactions}}
\\
&
\Rightarrow
&
(\lifelineElim_{\{h\}}(i),\lifelineElim_{\{h\}}(\mu) \overset{*}{\leadsto} \macroOKVerdict
&
\text{\scriptsize induction}
\end{array}
\]
    Hence, by $(i,\mu) \graphRelationHide (\lifelineElim_{\{h\}}(i),\lifelineElim_{\{h\}}(\mu))$ and  $(\lifelineElim_{\{h\}}(i),\lifelineElim_{\{h\}}(\mu) \overset{*}{\leadsto} \macroOKVerdict$, we obtain $(i,\mu) \overset{*}{\leadsto} \macroOKVerdict$.
    \item[$\Leftarrow$] Case $(i,\mu) \overset{*}{\leadsto} \macroOKVerdict$. As per  Prop.\ref{prop:confluence_lifeline_removal}, we have $(\lifelineElim_{\{h\}}(i),\lifelineElim_{\{h\}}(\mu)) \overset{*}{\leadsto} \macroOKVerdict$. By the induction hypothesis (the measure decremented by one),
    $\lifelineElim_{\{h\}}(\mu) \in \overline{\sigma_{L \setminus \{h\}}}(\lifelineElim_{\{h\}}(i))$. By Th.\ref{th:semantics_of_lifeline_removal_in_interactions}, we have  $\overline{\sigma_{L \setminus \{h\}}}(\lifelineElim_{\{h\}}(i))   = \overline{\lifelineElim_{\{h\}}(\sigma_L(i))}$. By the first point of Prop.\ref{prop:lifeline_elimination_and_multi_prefixes}, $\overline{\lifelineElim_{\{h\}}(\sigma_L(i))} = \lifelineElim_{\{h\}}(\overline{\sigma_L(i)})$. Via the second point of Prop.\ref{prop:lifeline_elimination_and_multi_prefixes}, as $\mu_{|h} = \varepsilon$ and $\lifelineElim_{\{h\}}(\mu) \in \lifelineElim_{\{h\}}(\overline{\sigma_L(i)})$,
  $\mu \in \overline{\sigma_L}(i)$.
\end{itemize}

\noindent $\bullet$ Case $\mu \neq \varepsilon_L$ such that  there are no lifeline $h$ with $\mu_{|h} = \varepsilon$
\begin{itemize}
    \item[$\Rightarrow$] if $\mu \in \overline{\sigma_L}(i)$, then there exists $\mu_+$ s.t. $\mu \opStrictSeq \mu_+ \in \sigma_L(i)$. Then, because $\mu \opStrictSeq \mu_+ \neq \varepsilon_L$, as per Prop.\ref{prop:operational_formulation_multitrace} there exist $a$, $i'$ and $\mu_*$ s.t. $\mu \opStrictSeq \mu_+ = a \multiAppend \mu_*$ and $i \xrightarrow{a} i'$ and $\mu_* \in \sigma_L(i')$.
    Then, because there is no empty trace component on $\mu$, action $a$ belongs to $\mu$. 
    Therefore there exists $\mu'$ such that $\mu = a \multiAppend \mu'$ and hence $(a \multiAppend \mu') \opStrictSeq \mu_+ = a \multiAppend \mu_*$ and thus $\mu' \opStrictSeq \mu_+ = \mu_* \in \sigma_L(i')$. Hence $\mu' \in \overline{\sigma_L}(i')$. We can therefore apply the induction hypothesis because $|(i',\mu')| = |(i,\mu)| - 1$ so that $(i',\mu') \overset{*}{\leadsto} \macroOKVerdict$. Moreover, we have $i \xrightarrow{a} i'$ and $\mu = a \multiAppend \mu'$, which allows us to apply rule $\graphRelationExec$ so that we have $(i,\mu) \graphRelationExec (i',\mu')$ and therefore, by transitivity, $(i,\mu) \overset{*}{\leadsto} \macroOKVerdict$.
    \item[$\Leftarrow$] if  $(i,\mu) \overset{*}{\leadsto} \macroOKVerdict$, then, given that we cannot apply $\graphRelationHide$, the  first transition is an application of $\graphRelationExec$ i.e., there exist $a$, $i'$ and $\mu'$ s.t. $i \xrightarrow{a} i'$ and $\mu = a \multiAppend \mu'$ and $(i,\mu) \graphRelationExec (i',\mu') \overset{*}{\leadsto} \macroOKVerdict$. Then, because $|(i',\mu')| = |(i,\mu)| - 1$, by the induction hypothesis, we have $\mu' \in \overline{\sigma_L}(i')$ which implies the existence of $\mu_+'$ such that $\mu' \opStrictSeq \mu_+' \in \sigma_L(i')$. 
    Additionally, the fact that $i \xrightarrow{a} i'$ and $\mu' \opStrictSeq \mu_+' \in \sigma_L(i')$ implies, as per Prop.\ref{prop:operational_formulation_multitrace}, that $a \multiAppend (\mu' \opStrictSeq \mu_+') \in \sigma_L(i)$. In particular, this implies that $\mu = a \multiAppend \mu' \in \overline{\sigma_L}(i)$.
\end{itemize}

\end{proof}

\subsection{Complexity\label{ssec:complexity}}

In \cite{a_small_step_approach_to_multi_trace_checking_against_interactions}, the problem of determining whether or not $\mu \in \sigma_{L}(i)$ has been proven to be NP-Hard via a reduction of the 1-in-3\,SAT problem (inspired by \cite{realizability_and_verification_of_msc_graphs}). 
In this paper, we reduce the more general 3SAT satisfiability problem to prove (see Prop.\ref{prop:multipref_nphard_complexity}) the NP-hardness of determining whether or not $\mu \in \overline{\sigma_{|L}(i)}$.

Let $X=\{x_1,\ldots,x_n\}$ be a finite set of Boolean variables. A literal $\ell$ is a Boolean variable $x\in X$ or its negation $\neg{x}$. A 3\,CNF (Conjunctive Normal Form) formula is an expression of the form $\phi = C_1 \wedge \ldots \wedge C_j \wedge \ldots \wedge  C_k $ with every clause $C_j$ being a disjunction of three distinct literals. On the left of Fig.\ref{fig:3sat_red} is given, as an example, such a Boolean expression $\phi$. The 3\,SAT problem is then to determine whether or not $\phi$ is satisfiable (whether or not there exists a variable assignment which sets all clauses in $\phi$ to $true$).

\begin{property}
\label{prop:multipref_nphard_complexity}
The problem of determining whether or not $\mu \in \overline{\sigma_{|L}(i)}$ is NP-hard.
\end{property}

\begin{proof}
Given a 3\,CNF formula $\phi$, with $|X| = n$ variables and $k$ clauses, we consider a set of lifeline $L=\{l_1,\ldots,l_k\}$ (a lifeline per clause), a unique message $m$, and the multi-trace $\mu_{\scriptscriptstyle \text{3\,SAT}}=(l_1?m,\ldots,l_k?m)$.

For any literal $\ell$, we build a multi-trace $\mu_{\ell}$ such that for any $l_j \in L$, if $\ell$ occurs in clause $C_j$ then $\mu_{\ell|l_j} = l_j?m$ and otherwise $\mu_{\ell|l_j} = \varepsilon$.
That is, every positive (resp. negative) occurrence of a variable $x$ in a clause $C_j$ is represented by an action $l_j?m$ in $\mu_x$ (resp. in $\mu_{\neg{x}}$). 
Let us then consider the set of multi-traces $T= (\{\mu_{x_1}\} \cup \{\mu_{ \neg{x_1}}\} );(\{\mu_{x_2}\} \cup \{\mu_{ \neg{x_2}}\} );\ldots;(\{\mu_{x_n}\} \cup \{\mu_{ \neg{x_n}}\} )$.
Every $\mu \in T$ corresponds to a variable assignment of the 3\,SAT problem. Indeed, to build $\mu$ either $\mu_{x}$ or $\mu_{\neg{x}}$ is selected (via $\cup$) in the definition of $T$, and not both. As $T$ is built using the sequencing (via $;$) of such alternatives for all variables, multi-traces in $T$ simulate all possible variable assignments (the search space for satisfying $\phi$).  
Because every clause contains three literals, one of which must be set to $true$, there is at least one literal $\ell \in \{x, \neg{x}\}$ in $C_j$ set to $true$. Hence $l_j?m \in \overline{\mu_{|l_j}}$. We remark that  $\mu_{|l_j}$ can be a sequence of such emissions $l_j?m$ if more than one literal is set to true in $C_j$. This reasoning can be applied to all the clauses i.e. $\forall j \in [1,k]$, $l_j?m \in \overline{\mu_{|l_j}}$ which implies that $\mu_{\scriptscriptstyle \text{3\,SAT}} \in \overline{\mu}$ and hence $\mu_{\scriptscriptstyle \text{3\,SAT}} \in \overline{T}$.
Given that $T$ is equivalent to the semantics of an interaction $i$ of the form $seq(alt(i_{x_1},i_{\neg{x_1}}),seq(alt(i_{x_2},i_{\neg{x_2}}),\cdots,alt(i_{x_n},i_{\neg{x_n}})\cdots))$, with, for any literal $\ell$, $i_\ell$ being the sequencing of all $l_j?m$ such that $\ell$ appears in $C_j$, solving the 3\,SAT problem equates to solving $\mu_{\scriptscriptstyle \text{3\,SAT}} \in \overline{\sigma_{|L}(i)}$.
\end{proof}

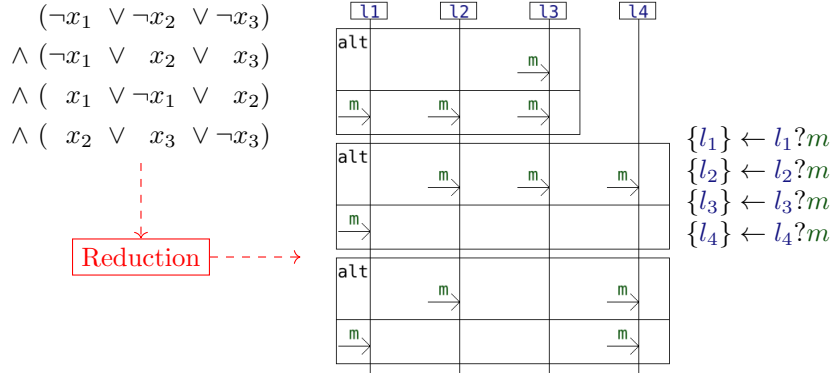
\begin{figure}[ht]
    \centering
\begin{tikzpicture}
\node[ 
line width=1.25pt] (problem_3sat) at (0,0) {
    $\begin{aligned}
    ~~~& ( \neg x_1 ~\vee \neg x_2 ~\vee \neg x_3 )\\
    \wedge~ & ( \neg x_1 ~\vee \phantom{\neg}x_2 ~\vee \phantom{\neg}x_3 )\\
    \wedge~ & ( \phantom{\neg}x_1 ~\vee \neg x_1 ~\vee \phantom{\neg}x_2 )\\
    \wedge~ & ( \phantom{\neg}x_2 ~\vee \phantom{\neg}x_3 ~\vee \neg x_3 )
    \end{aligned}$
};
\node[
line width=1.25pt,below right=-2.5cm and .25cm of problem_3sat] (problem_multipref) {
    \begin{tikzpicture}
    \node (int) {\includegraphics[scale=.3]{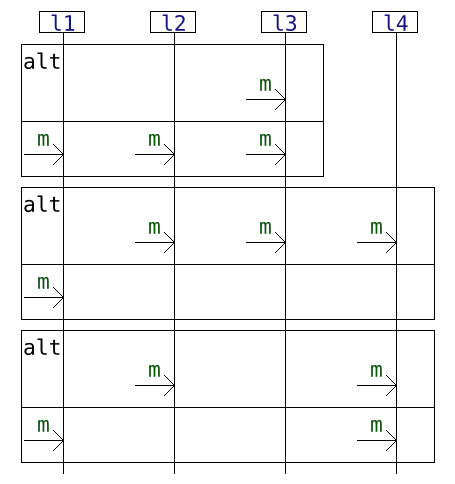}};
    \node[right=-.25cm of int,align=left] (mu) {
    $\{\hlf{l_1}\} \leftarrow \hlf{l_1}?\hms{m}$\\
    $\{\hlf{l_2}\} \leftarrow \hlf{l_2}?\hms{m}$\\
    $\{\hlf{l_3}\} \leftarrow \hlf{l_3}?\hms{m}$\\
    $\{\hlf{l_4}\} \leftarrow \hlf{l_4}?\hms{m}$
    };
    \end{tikzpicture}
};
\node[draw,rectangle,red,below=1cm of problem_3sat] (red) {
    Reduction
};
\draw[dashed,red,->] (problem_3sat) -- (red);
\draw[dashed,red,->] (red) -- (problem_multipref.west |- red);
\end{tikzpicture}
    \caption{Principle of 3\;SAT reduction}
    \label{fig:3sat_red}
\end{figure}

Hence, we have provided a polynomial reduction of 3\,SAT to the problem of recognizing multi-prefixes of accepted multi-traces. The reduction of the problem on the left of Fig.\ref{fig:3sat_red} is represented on the right of Fig.\ref{fig:3sat_red}, via drawing the resulting interaction and multi-trace.
The problem has 3 variables and 4 clauses. In the corresponding interaction, lifeline $l_1$ corresponds to the first clause and we see that it has a $l_1?m$ in the right branch of the first alternative, corresponding to $\neg x_1$, the right branch of the second for $\neg x_2$ and the right branch of the third for $\neg x_3$. The same applies to $l_2$, $l_3$ and $l_4$.

The underlying problem being NP-Hard, our approach should ideally use additional techniques to help reduce the average complexity.
The algorithm which we propose relies on an exploration of $\mathbb{G}$ from a certain vertex $(i,\mu)$.
The part of $\mathbb{G}$ that is reachable from $(i,\mu)$ constitutes its search space. 
In Sec.\ref{sec:graph_size}, we propose techniques that reduce the size of this search space.

\section{Techniques to reduce the search space\label{sec:graph_size}}

Given a vertex $(i,\mu)$, the set of vertices reachable from $(i,\mu)$ is finite:

\begin{property}
\label{prop:finitesearchspace}
For any $L \subseteq \mathcal{L}$, $\mu \in \mathbb{M}(L)$ and $i \in \mathbb{I}(L)$, the set\\$\{ v \in \mathbb{V} ~|~ (i,\mu) \overset{*}{\leadsto} v \}$ is finite. 
\end{property}

\begin{proof}
It follows from the following two observations
\textbf{(1)} any path in that sub-graph is finite
and \textbf{(2)} there is a finite number of paths.

The first point \textbf{(1)} can be proven by reasoning on the measure $|v|$ of vertices $ v \in \mathbb{V}$.
For any transition $(i,\mu) \leadsto (i',\mu')$ in $\mathbb{G}$ other than $\graphRelationPass$, we have $|(i',\mu')| \leq |(i,\mu)| - 1$ whether the rule that is applied is $\graphRelationExec$ or $\graphRelationHide$. More precisely, the number of actions decreases by one for $\graphRelationExec$  and the number of lifelines strictly decreases for $\graphRelationHide$.

The second point \textbf{(2)} comes from the fact that for any vertex $(i,\mu)$, there exists a finite number of outgoing transitions.
There are at most $|i|$ possible applications of $\graphRelationExec$, $|i|$ being the number of actions in $i$.
There is at most $2^{|L|}  - 2$ possible applications of $\graphRelationHide$, which correspond to all possible subsets of $L$ different of $\emptyset$ and $L$.
There is at most $1$ possible application of $\graphRelationPass$.
\end{proof}

Given Prop.\ref{prop:finitesearchspace} and Th.\ref{th:multipref_equates_pass}, exploring graph $\mathbb{G}$ to find a path $(i,\mu) \overset{*}{\leadsto} \macroOKVerdict$ is a sound decision procedure to determine whether or not $\mu \in \overline{\sigma_L(i)}$. Indeed, while the former ensures it returns in finite time, the latter ensures its correctness.

The definition of graph $\mathbb{G}$ in Def.\ref{def:search_graph} allows for many possible interleavings of the applications of rules $\graphRelationExec$ and $\graphRelationHide$. Let us define two variants $\shortColViolet{\rightsquigarrow'_r}\, \subset \graphRelationHide$ and $\shortColOrange{\rightsquigarrow'_e}\, \subset \graphRelationExec$ of the removal and execution rules as follows:
\[
\begin{array}{lcl}
(i,\mu)
\,\shortColViolet{\rightsquigarrow'_r}\,
(\lifelineElim_{H}(i),\lifelineElim_{H}(\mu))
&
\text{iff} 
&
H = \{ l \in L ~|~ \mu_{|l} = \varepsilon \}
\\ 
(i,a \multiAppend \mu)
\,\shortColOrange{\rightsquigarrow'_e}\,
(i',\mu)
&
\text{iff} 
&
\forall~l \in L,~\mu_{|l} \neq \varepsilon
\end{array}
\]

This corresponds to:
\begin{itemize}
    \item rationalizing the use of $\graphRelationHide$ by enforcing the removal, in a single step, of all the lifelines $l \in L$ such that $\mu_{|l} = \varepsilon$
    \item prioritizing the evaluation of $\graphRelationHide$ over $\graphRelationExec$ by preventing the application of $\graphRelationExec$ as long as it is possible to apply $\graphRelationHide$
\end{itemize}

Prop.\ref{prop:commut_hide} (commutativity of lifeline removal) and Prop.\ref{prop:confluence_lifeline_removal} (confluence w.r.t.~the use of $\graphRelationHide$) ensure that this simplification preserves the existence of paths $(i,\mu) \overset{*}{\leadsto} \macroOKVerdict$.
Because this replacing of $\graphRelationHide$ by $\shortColViolet{\rightsquigarrow'_r}$ and of $\graphRelationExec$ by $\shortColOrange{\rightsquigarrow'_e}$ comes at no cost and improves graph size, we take it for granted in the remainder of the paper and override the $\graphRelationHide$ and $\graphRelationExec$ notations.

\begin{figure}[ht]
    \centering
    \scalebox{.65}{\input{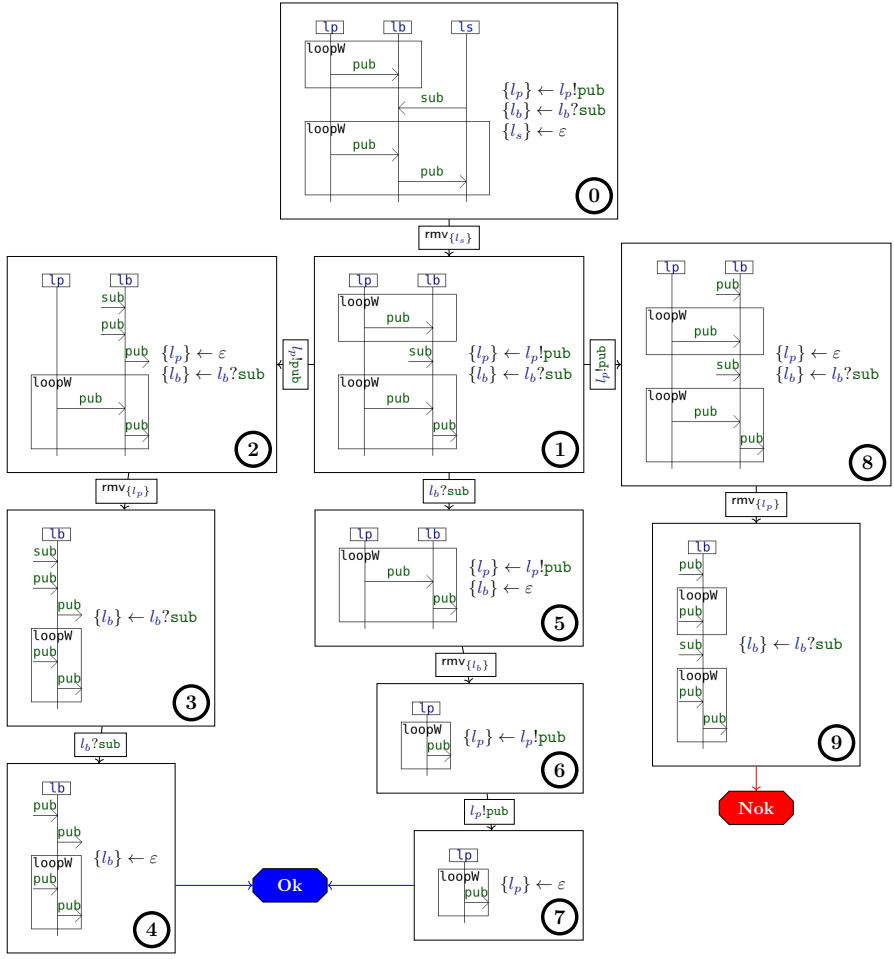}}
    \caption{Part of $\mathbb{G}$ reachable from $(i_0,\mu_0)$ with $i_0$ from Fig.\ref{fig:interaction_example} and $\mu_0$ from Fig.\ref{fig:observation_partial}}
    \label{fig:fsen_ex_lfrem_ana}
\end{figure}

Fig.\ref{fig:fsen_ex_lfrem_ana} represents the part of $\mathbb{G}$ that is reachable via this reduced version of $\leadsto$ from the vertex $(i_0,\mu_0)$ defined over $L=\{\hlf{l_p},\hlf{l_b},\hlf{l_s}\}$ with $i_0$ our example interaction from Fig.\ref{fig:interaction_example} and $\mu_0$ the partially observed behavior from Fig.\ref{fig:observation_partial}.
As in Fig.\ref{fig:expla_ana}, each square annotated with a circled number contains an interaction drawn on the left and a multi-trace on the right.
Starting from the $(i_0,\mu_0)$ vertex (denoted by \textcircled{0}), only $\graphRelationHide$ is applicable.
It consists of removing lifeline $\hlf{l_s}$ which is not/no longer observed.
The application of $\graphRelationHide$ yields the vertex $(i_1,\mu_1)$ (denoted by \textcircled{1}).
From there, only $\graphRelationExec$ is applicable, but there are 3 possible applications, yielding three distinct vertices \textcircled{2}, \textcircled{5} and \textcircled{8}. Indeed, we can match the $\hlf{l_b}?\hms{sub}$ at the beginning of $\mu_{1|\hlf{l_b}}$ with the only $\hlf{l_b}?\hms{sub}$ of $i_1$ and we can also match the $\hlf{l_p}!\hms{pub}$ at the beginning of $\mu_{1|\hlf{l_p}}$ with either of the two $\hlf{l_p}!\hms{pub}$ in $i_1$.
Verdict $\macroOKVerdict$ is represented by the blue vertex at the bottom left of Fig.\ref{fig:fsen_ex_lfrem_ana}. It is reachable from \textcircled{4} and \textcircled{7} via $\graphRelationPass$ and thus we can conclude that $\mu_0 \in \overline{\sigma_L(i_0)}$ via Th.\ref{th:multipref_equates_pass}.
For the sake of pedagogy, we represent a second abstract verdict $\macroKOVerdict$ in red (on the bottom right of Fig.\ref{fig:fsen_ex_lfrem_ana}) to highlight $(i,\mu)$ vertices that have no successors and therefore cannot lead to $\macroOKVerdict$.

In the following, we present two additional techniques, which both involve the lifeline removal operator $\lifelineElim$ (in different manners and distinctly from its use in $\graphRelationHide$) and which aim is to reduce the size of the search space further:
\begin{itemize}
    \item a Partial Order Reduction technique (POR), which mitigates concurrency between the different local traces that arise when using $\graphRelationExec$
    \item a technique of ``LOCal analyses'' (LOC), which cuts parts of the graph reachable from vertices $(i,\mu)$ that cannot be in $\overline{\sigma_L(i)}$. 
\end{itemize}

\subsection{Partial Order Reduction via one-unambiguous actions\label{ssec:partial_order_reduction}}

Partial Order Reduction (POR) \cite{combining_partial_order_reduction_with_on_the_fly_model_checking,partial_order_methods_for_model_checking_from_linear_time_to_branching_time} refers to a family of techniques that are widely used in model checking to tackle state explosion.
In a few words, POR leverages the absence of effect that various interleavings of atomic operations may have on the properties of a system to restrict the exploration to a single linearization of these operations.
In our context, POR aims to restrict the possible applications of $\graphRelationExec$.

Various POR approaches are based on the comparison of pairs of atomic transitions (see e.g., ``accordance'' relations between transitions from \cite{eliminating_redundant_interleavings_during_concurrent_program_verification,partial_order_reduction_for_parity_games_and_parameterized_boolean_equation_systems}).
However, this does not work well on our problem of multi-prefix membership from interactions.
Intuitively (and erroneously), given a set of lifelines $L$, for any two actions $\shortColCyan{a_1}$ and $\shortColBrown{a_2}$ of $\mathbb{A}(L)$ such that $\theta(\shortColCyan{a_1}) \neq \theta(\shortColBrown{a_2})$, when analyzing from $(i, \shortColCyan{a_1} \multiAppend \shortColBrown{a_2} \multiAppend \mu)$, given $i \in \mathbb{I}(L)$ and $\mu \in \mathbb{M}(L)$, if, every vertex $(i',\mu')$ that can be reached by executing at first $\shortColBrown{a_2}$ and then $\shortColCyan{a_1}$ can also be reached by executing at first $\shortColCyan{a_1}$ and then $\shortColBrown{a_2}$, then we might as well only keep a single linearization by only considering executing $\shortColCyan{a_1}$ first.

\begin{figure}[ht]
\centering
\scalebox{.725}{\input{images/por3/por3}}
\captionsetup{singlelinecheck=off}
\caption{
Here both $\shortColBrown{l_1?m_2}$ and $\shortColCyan{l_2?m_1}$ are \shortColRed{one-ambiguous}
\\
It is recommended that : $\bullet$ both $\shortColBrown{l_1?m_2}$ and $\shortColCyan{l_2?m_1}$ should be explored\\
\phantom{It is recommended that : }$\bullet$ both \textcircled{1} and \textcircled{2} should be kept~~~~~~~~~~~~~~
}
\label{fig:por3}
\end{figure}

The example from Fig.\ref{fig:por3} demonstrates that this intuition is wrong.
The vertices that can be reached by executing $\shortColCyan{l_2?m_1}$ at first and then $\shortColBrown{l_1?m_2}$ are \textcircled{3} and \textcircled{4} hatched in cyan. On the other hand, only \textcircled{4}, which is also hatched in brown, is reachable by executing at first $\shortColBrown{l_1?m_2}$ and then $\shortColCyan{l_2?m_1}$.
Here, all vertices that are reachable by executing at first $\shortColBrown{l_1?m_2}$ and then $\shortColCyan{l_2?m_1}$ are also reachable by executing at first $\shortColCyan{l_2?m_1}$ and then $\shortColBrown{l_1?m_2}$.
Yet, only by executing $\shortColBrown{l_1?m_2}$ at first, one can reach the $\macroOKVerdict$.

We propose another solution for POR  based on a notion of ``one-ambiguity'' of actions that we adapt from \cite{one_unambiguous_regular_languages,deterministic_regular_expressions_in_linear_time,learning_deterministic_regular_expressions_for_the_inference_of_schemas_from_xml_data} to our context in Def.\ref{def:one_unambiguous_action}.

\begin{definition}
\label{def:one_unambiguous_action}    
For any $L \subset \mathcal{L}$, any interaction $i \in \mathbb{I}(L)$ and any action $a \in \mathbb{A}(L)$, we say that $a$ is one-unambiguous in $i$, which we denote as ``$a \oneUnambiguous i$'', iff:
\[
\exists!~p \in \{1,2\}^* ~s.t.~ \exists~i' \in \mathbb{I}(\{\theta(a)\}) ~s.t. ~\lifelineElim_{L \setminus \{\theta(a)\}}(i) \xrightarrow{a@p} i'
\]
\end{definition}
If $a$ is not one-unambiguous in $i$, we say it is {\em one-ambiguous} in $i$. 

$a \oneUnambiguous i$ 
signifies that when restricted to the lifeline $\theta(a)$ alone, a unique instance of action $a$ is immediately executable. 
To be able to know this, we use the $\lifelineElim$ operator to remove the constraints imposed by the behavior of the other lifelines (i.e., we do not have to wait for the execution of third party actions on other lifelines to be able to execute instances of $a$).
In other words, an action $a$ is one-unambiguous in $i$ if, whatever occurs on the other lifelines, there is always in $i$ a single unique manner to interpret it as a local behavior on $\theta(a)$.

One-unambiguity of $a$ in $i$ implies that, if $a$ is executable in $i$, then there is only one unique manner to interpret it, i.e., one unique follow-up $i'$ such that $i \xrightarrow{a} i'$ (see Prop.\ref{prop:one_ambiguous_unicity_before_projection}).

\begin{property}
\label{prop:one_ambiguous_unicity_before_projection}
For any $i \in \mathbb{I}(L)$ and any $a \in \mathbb{A}(L)$ we have:
\[
\left( 
(a \oneUnambiguous i) \wedge (\exists~i' \in \mathbb{I}(L),~i \xrightarrow{a} i')
\right) 
\Rightarrow
(\exists!~i' \in \mathbb{I}(L),~i \xrightarrow{a} i')
\]
\end{property}

\begin{proof}
As,by hypothesis, there exists $i'$ in $\mathbb{I}(L)$ with $i \xrightarrow{a} i'$, 
let us denote $p$ the position ensuring $i \xrightarrow{a@p} i'$.
Then per Prop.\ref{prop:execution_and_lifeline_removal},
$
\lifelineElim_{L \setminus \{\theta(a)\}}(i) \xrightarrow{a@p} \lifelineElim_{L \setminus \{\theta(a)\}}(i')
$.

On the other hand, by definition of $a \oneUnambiguous i$ (Def.\ref{def:one_unambiguous_action}), we have
\[
\exists!~p \in \{1,2\}^* ~s.t.~ \exists~i'' \in \mathbb{I}(\{\theta(a)\}) ~s.t. ~\lifelineElim_{L \setminus \{\theta(a)\}}(i) \xrightarrow{a@p} i''
\]

The unicity of $p$ implies
the unicity of $i'$.
\end{proof}

If from a vertex $(i,\mu)$, we have an action $a$, an interaction $i'$ and a multi-trace $\mu'$ such that $\mu = a \multiAppend \mu'$, $i \xrightarrow{a} i'$ and $a 
 \oneUnambiguous i$, then it is safe to select $(i',\mu')$ as a unique successor to $(i,\mu)$, thus ignoring all the other $(i'',\mu'')$ s.t., $(i,\mu) \graphRelationExec (i'',\mu'')$.
We prove this with Prop.\ref{prop:correctness_partial_order_reduction}, which states a confluence property of graph $\mathbb{G}$.

\begin{property}
\label{prop:correctness_partial_order_reduction}
For any $L \subset \mathcal{L}$, any interactions $i$ and $i'$ in $\mathbb{I}(L)$, any multi-trace $\mu \in \mathbb{M}(L)$ and action $a \in \mathbb{A}(L)$, we have the following property on graph $\mathbb{G}$:

\begin{scprooftree}{.9}
\AxiomC{$(i, a \multiAppend \mu) \overset{*}{\leadsto} \macroOKVerdict$}
\AxiomC{$(i, a \multiAppend \mu) \graphRelationExec (i',\mu)$}
\LeftLabel{}
\RightLabel{$a \oneUnambiguous i$}
\BinaryInfC{$(i',\mu) \overset{*}{\leadsto} \macroOKVerdict$}
\end{scprooftree}

\end{property}

\begin{proof}

As per Prop.\ref{prop:one_ambiguous_unicity_before_projection}, the fact that $(i, a \multiAppend \mu) \graphRelationExec (i',\mu)$ and that $a \oneUnambiguous i$ implies that $i'$ is the unique interaction such that $i \xrightarrow{a} i'$.

As per Prop.\ref{prop:operational_formulation_multitrace} and because there can be no other source of $a$ actions, we get:
$\forall~t \in \mathbb{A}(L)^*$,  $(a.t \in \sigma(i))$ iff $(t \in \sigma(i'))$. By using $\muProjection_L$, we obtain:
$\forall~\mu' \in \mathbb{M}(L)$,  $(a \multiAppend \mu' \in \sigma_L(i))$ iff $(\mu' \in \sigma_L(i'))$.

When considering prefixes, we get:
$\forall~\mu' \in \mathbb{M}(L)$, $(a \multiAppend \mu' \in \overline{\sigma_L(i)})$ iff $(\mu' \in \overline{\sigma_L(i')})$. On the other hand, $(i,a \multiAppend \mu) \overset{*}{\leadsto} \macroOKVerdict$ implies, as per Th.\ref{th:multipref_equates_pass}, that $a \multiAppend \mu \in \overline{\sigma_L(i)}$.
Combining the two previous points, we have: $\mu \in \overline{\sigma_L(i')}$. As per Th.\ref{th:multipref_equates_pass}, we get:
$(i',\mu) \overset{*}{\leadsto} \macroOKVerdict$.

\end{proof}

Reasoning on one-unambiguous actions
on the example from Fig.\ref{fig:por3}, we conclude that it is not safe to ignore either \textcircled{1} or \textcircled{2}.
We represent on the top left corner of Fig.\ref{fig:por3}, the determination of which consumable (via $\graphRelationExec$) actions are one-unambiguous in $i_0$.
By projecting $i_0$ on $\hlf{l_1}$ (via $\lifelineElim_{\hlf{l_2}}$), we remark that there are two immediately executable instances of $\shortColBrown{l_1?m_2}$ and, as a result, $\shortColBrown{l_1?m_2}$ is not one-unambiguous. 
Similarly, when projecting $i_0$ on $\hlf{l_2}$ (via $\lifelineElim_{\hlf{l_1}}$), we observe two immediately executable instances of $\shortColCyan{l_2?m_1}$ and, as a result, $\shortColCyan{l_2?m_1}$ is not one-unambiguous.

\begin{figure}[ht]
\centering
\scalebox{.725}{\input{images/por1/por1}}
\caption{
Here $\shortColCyan{l_2!m_1}$ is \shortColGreen{one-unambiguous} while $\shortColBrown{l_1?m_1}$ is \shortColRed{one-ambiguous}
\\
It is recommended that : $\bullet$ $\shortColCyan{l_2!m_1}$ should be explored over $\shortColBrown{l_1?m_1}$\\
\phantom{It is recommended that : }$\bullet$ \textcircled{2} should be ignored and \textcircled{1} kept~~~~
}
\label{fig:por1}
\end{figure}

On Fig.\ref{fig:por1} and Fig.\ref{fig:por2}, we provide two other examples in which the one-ambiguity criterion allows the selection of a unique linearization.
On Fig.\ref{fig:por1}, $\shortColBrown{l_1?m_1}$ is ambiguous while we have $\shortColCyan{l_2!m_1} \oneUnambiguous i_0$. Hence, it is safe to execute $\shortColCyan{l_2!m_1}$ before any other action and retain only this single linearization (thus ignoring vertex \textcircled{2} which is greyed-out).
On Fig.\ref{fig:por2}, $\shortColCyan{l_1?m_2}$ is one-unambiguous and $\shortColBrown{l_2?m_1}$ is not. This then justifies the choice of $\shortColCyan{l_1?m_2}$ over $\shortColBrown{l_2?m_1}$.

\begin{figure}[ht]
\centering
\scalebox{.725}{\input{images/por2/por2}}
\caption{
Here $\shortColCyan{l_1?m_2}$ is \shortColGreen{one-unambiguous} while $\shortColBrown{l_2?m_1}$ is \shortColRed{one-ambiguous}
\\
It is recommended that : $\bullet$ $\shortColCyan{l_1?m_2}$ should be explored over $\shortColBrown{l_2?m_1}$\\
\phantom{It is recommended that : }$\bullet$ \textcircled{2} should be ignored and \textcircled{1} kept~~~~
}
\label{fig:por2}
\end{figure}

The notion of one-unambiguous actions is reminiscent of works on the use of ``one-unambiguous''\footnote{one-unambiguous regular expressions are also called ``deterministic'' (in this precise context) because their associated Glushkov automaton is deterministic} regular expressions for SGML\footnote{Standard Generalized Markup Language} and later DTD\footnote{Document Type Definition} and XML\footnote{eXtensible Markup Language} schemas \cite{one_unambiguous_regular_languages,deterministic_regular_expressions_in_linear_time,learning_deterministic_regular_expressions_for_the_inference_of_schemas_from_xml_data}. 
Indeed, these schemas (which are used to format standardized text documents) are described by regular expressions that are required to be one-unambiguous.
In this context, a regular expression is one-unambiguous \cite{one_unambiguous_regular_languages} if, without peaking ahead at the input word, one can match each letter of that word unambiguously to a specific instance of that letter in the regular expression (uniquely defined via its position in the regular expression). For instance (example from \cite{learning_deterministic_regular_expressions_for_the_inference_of_schemas_from_xml_data}), $e=(a+b)^*a$ is not one-unambiguous because the first symbol in $aaa$ could be matched to either the first or the second $a$ in $e$ (and it is impossible to know which one is correct without peaking ahead). By contrast, $b^*a(v^*a)^*$ is one-unambiguous because only the first $a$ is a possible match.

In order to include POR to the definition of the graph, we define a variant $\graphRelationPORExec \subset \graphRelationExec$ such that for any interactions $i \in \mathbb{I}(L)$ and multi-trace $\mu \in \mathbb{M}(L)$:
\begin{itemize}
    \item if there exist $a \in \mathbb{A}(L)$ and $i' \in \mathbb{I}(L)$ s.t.~$\mu = a \multiAppend \mu'$, $i \xrightarrow{a} i'$ and $a \oneUnambiguous i$ then there is a unique such $a$, $i'$ and $\mu'$ such that $(i,\mu) \graphRelationPORExec (i',\mu')$. In particular, this implies there are no other $(i'',\mu'')$ such that $(i,\mu) \graphRelationPORExec (i'',\mu'')$, even if there are one-unambiguous actions other than $a$ in $i$;

    \item if not, then $\forall~(i',\mu') \in \mathbb{V}$ s.t.~$(i,\mu) \graphRelationExec (i',\mu')$, we also have $(i,\mu) \graphRelationPORExec (i',\mu')$.
\end{itemize}

We then define $\leadsto^{\mathbf{POR}} = (\graphRelationPass \cup \graphRelationHide \cup \graphRelationPORExec)$.

\subsection{Local analyses\label{ssec:local_analyses}}

In the following, we refer to the exploration of $\mathbb{G}$ to find a path $(i,\mu) \overset{*}{\leadsto} \macroOKVerdict$ as a ``global analysis''. 
We define its success as the discovery of such a path. If $\mathbb{G}$ is explored entirely and no such path is found, then the analysis fails because it implies that no such path exists.
As per Th.\ref{th:multipref_equates_pass} this implies that a global analysis algorithm solves the following membership problem: whether or not $\mu \in \overline{\sigma_L(i)}$.

``{\em Local analyses}'' involve examining the execution of individual lifelines $l\in L$ to verify the validity of components $\mu_{|l}$ within a multi-trace $\mu$, with regards to corresponding restricted local views of the interaction $i$. 
Concretely, these local views are obtained via removing from the interaction $i$ all the actions that do not occur on $l$ i.e., they correspond to $\lifelineElim_{L \setminus \{l\}}(i)$.

Local analyses can be leveraged during a global analysis to reduce the size of its search space (i.e., the parts of $\mathbb{G}$ that need to be explored).
Let us indeed illustrate this via the example on Fig.\ref{fig:locana_principle_glocal}.

\begin{figure}[h]
    \centering

\scalebox{.725}{\input{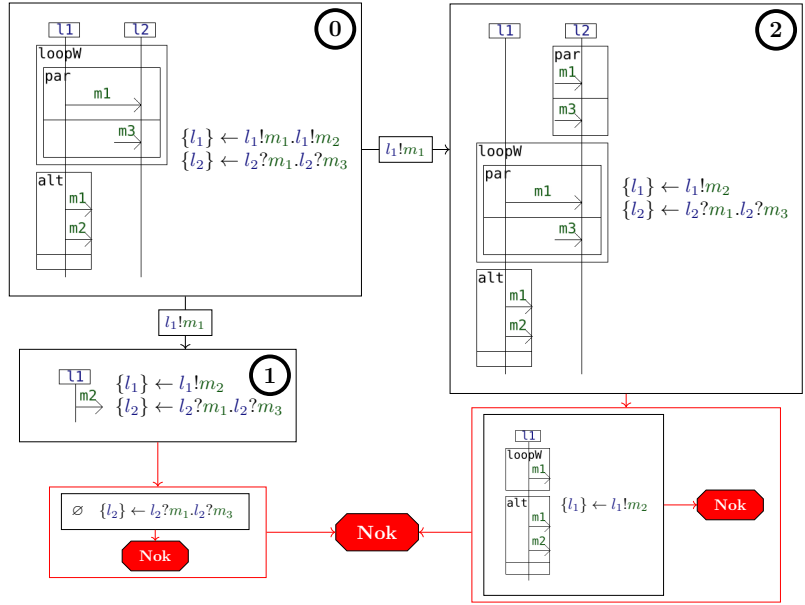}}
\caption{Local analyses to predict non-membership}
\label{fig:locana_principle_glocal}
\end{figure}

Given $L_0 = \{\hlf{l_1},\hlf{l_2}\}$, let us consider the interaction $i_0$ and multi-trace $\mu_0$ drawn on the top left of Fig.\ref{fig:locana_principle_glocal}.
Fig.\ref{fig:locana_principle_glocal} describes a global analysis (solving whether or not $\mu_0 \in \overline{\sigma_{L_0}(i_0)}$) that is complemented with local analyses.

Starting from the initial vertex $(i_0,\mu_0)$ of $\mathbb{G}$, denoted by \textcircled{0}, only $\graphRelationExec$ can be applied and the only possible matching action in $\mu_0$ is the $\hlf{l_1}!\hms{m_1}$. 
There are then two possible applications of $\graphRelationExec$, matching $\hlf{l_1}!\hms{m_1}$ from $\mu_0$ with either occurrence of $\hlf{l_1}!\hms{m_1}$ in $i_0$.

In the case where we match the action underneath the $alt$, via $(i_0,\mu_0) \graphRelationExec (i_1,\mu_1)$ 
(denoted by 
\textcircled{0} and 
\textcircled{1} on Fig.\ref{fig:locana_principle_glocal}), although we could keep exploring the graph by applying $\graphRelationExec$ from 
\textcircled{1} (via executing $\hlf{l_1}!\hms{m_2}$), a local 
analysis on \textcircled{1} (framed in red on the bottom left of Fig.\ref{fig:locana_principle_glocal}) allows us to stop the exploration beyond \textcircled{1}.
Indeed, this local analysis reveals an issue with the component $\mu_{1|\hlf{l_2}}$ of $\mu_1$ on $\hlf{l_2}$, which is not a prefix of a behavior that is accepted by the corresponding local view of $i_1$ i.e., $\lifelineElim_{\{\hlf{l_1}\}}(i_1)$.

Similarly, in the case where the $\hlf{l_1}!\hms{m_1}$ underneath the $loop_W$ is matched, yielding $(i_0,\mu_0) \graphRelationExec (i_2,\mu_2)$, a local analysis on \textcircled{2} reveals a similar problem with $\mu_{2|\hlf{l_1}}$.
This avoids further exploration of the graph (i.e., executing $\hlf{l_2}?\hms{m_1}$ and then $\hlf{l_2}?\hms{m_3}$).

In the following, we will reason on multi-traces that are reduced to a single lifeline.
These multi-traces can be equated to traces (the $\multiAppend$ operator then equating classical concatenation).
For the sake of simplicity, for any multi-trace $\mu \in \mathbb{M}(L)$ and any lifeline $l \in L$, we use the notation $\mu_{|l}$ of the trace component of $\mu$ on $l$ to denote the single lifeline multi-trace $\lifelineElim_{L \setminus \{l\}}(\mu)$.
For any $\depthAnaLoc \in \mathbb{N}$, we denote by $\mu \lbrack 0..\depthAnaLoc \rbrack$ the multi-trace such that $\forall~l \in L$, $(\mu \lbrack 0..\depthAnaLoc \rbrack)_{|l} = \mu_{|l}\lbrack 0..\depthAnaLoc \rbrack$.
This implies that $\mu \lbrack 0..\depthAnaLoc \rbrack \in \overline{\mu}$ and $\lifelineElim_{L \setminus \{l\}}(\mu \lbrack 0..\depthAnaLoc \rbrack) = \mu_{|l}\lbrack 0..\depthAnaLoc \rbrack$

Th.\ref{th:semantics_of_lifeline_removal_in_interactions} states that, for any subset $L \subseteq \mathcal{L}$, any interaction $i \in \mathbb{I}(L)$, and any proper subset $H \subsetneq L$, we have $\sigma_{L \setminus H}(\lifelineElim_H(i)) = \lifelineElim_H(\sigma_L(i))$. 
In particular, for any multi-trace $\mu \in \overline{\sigma_L(i)}$ and any lifeline $l \in L$, it implies that:

\[
\begin{array}{ll}
&  
\lifelineElim_{L \setminus \{l\}}(\mu) \in \overline{\sigma_{\{l\}}(\lifelineElim_{L \setminus \{l\}}(i))}
\\[.3cm]
\text{and } \forall~\depthAnaLoc \in \mathbb{N},
&
\mu_{|l}\lbrack 0..\depthAnaLoc \rbrack \in \overline{\sigma_{\{l\}}(\lifelineElim_{L \setminus \{l\}}(i))}
\end{array}
\]

Conversely, given a $\depthAnaLoc \in \mathbb{N}$, if there exists a lifeline $l \in L$ such that $\mu_{|l}\lbrack 0..\depthAnaLoc \rbrack \not\in \overline{\sigma_{\{l\}}(\lifelineElim_{L \setminus \{l\}}(i))}$, then it implies that $\mu \not\in \overline{\sigma_L(i)}$.

Def.\ref{def:ana_local} defines local analyses as conducting these verifications on each lifeline $l \in L$ individually up to a specified depth $\depthAnaLoc \in \mathbb{N}$. 
This depth corresponds to a look-ahead length, where non-membership verification is performed only for a prefix $\mu \lbrack 0..\depthAnaLoc \rbrack$, the first $\depthAnaLoc$ actions of each trace component $\mu_{|l}$ being analyzed.
Parameterizing $\depthAnaLoc$ can help mitigate the cost of performing the local analyses.

\begin{definition}
\label{def:ana_local}
For any $L \subseteq \mathcal{L}$ and $p\in \mathbb{N}$, the predicate $\anaLocal_{L,\depthAnaLoc} \subseteq \mathbb{I}(L) \times \mathbb{M}({L})$ is such that for any $i \in \mathbb{I}(L)$ and $\mu \in \mathbb{M}(L)$:
\[
\anaLocal_{L,\depthAnaLoc}(i,\mu)
\Leftrightarrow 
\left(
\forall~l \in L,~
(\lifelineElim_{L \setminus \{l\}}(i),\mu_{|l}[0..\depthAnaLoc]) \overset{*}{\leadsto} \macroOKVerdict
\right)
\]
\[
\anaLocal_{L,\infty}(i,\mu)
\Leftrightarrow 
\left(
\forall~l \in L,~
(\lifelineElim_{L \setminus \{l\}}(i),\mu_{|l}) \overset{*}{\leadsto} \macroOKVerdict
\right)
\]
\end{definition}

Given a vertex $(i, \mu)$, the failure of any local analysis (i.e., having $\neg \anaLocal_{L,\depthAnaLoc}(i, \mu)$ or $\neg \anaLocal_{L,\infty}(i, \mu)$) implies, by the contrapositive of Prop.\ref{prop:local_analyses_deviation}, the failure of the global analysis. Thus, exploring the sub-graph reachable from $(i,\mu)$ would serve no purpose.

\begin{property}
\label{prop:local_analyses_deviation}
For $L \subseteq \mathcal{L}$, $\depthAnaLoc \in \mathbb{N}$, $i \in \mathbb{I}(L)$ and $\mu \in \mathbb{M}(L)$, we have:
\[ 
\left(
(i,\mu) \overset{*}{\leadsto} \macroOKVerdict
\right)
\Rightarrow
\anaLocal_{L,\depthAnaLoc}(i,\mu) 
\]
\end{property}

\begin{proof}
Implied by Th.\ref{th:semantics_of_lifeline_removal_in_interactions} and the fact that $\mu_{|l}[0..\depthAnaLoc] \in \overline{\lifelineElim_{L \setminus \{l\}}(\mu)}$ for any $l \in L$, $\mu \in \mathbb{M}(L)$ and $\depthAnaLoc \in \mathbb{N}$.
\end{proof}

By the way, Prop.\ref{prop:local_analyses_deviation} ensures $
\left(
(i,\mu) \overset{*}{\leadsto} \macroOKVerdict
\right)
\Rightarrow
\anaLocal_{L,\infty}(i,\mu) 
$.

In the example from Fig.\ref{fig:locana_principle_glocal}, via the use of local analyses, we reduced the size of the search space from $7$ vertices to $3$. 
However, this gain can be much more consequent, as can be shown by generalizing the example.
Indeed, let us consider $i^* = seq(loop_W(strict(\hlf{l_1}!\hms{m_1},\hlf{l_2}?\hms{m_1})),alt(seq(\hlf{l_1}!\hms{m_1},\hlf{l_1}!\hms{m_2}),\varnothing))$ and, for any $n\geq 2$, $i^\dag_n=seq(\hlf{l_2}!\hms{m_2},seq(\cdots,seq(\hlf{l_2}!\hms{m_{n-1}},\hlf{l_2}!\hms{m_n})\cdots))$, let us then consider $i_n = seq(i^*,i^\dag_n)$ and $\mu_n = (\hlf{l_1}!\hms{m_1}.\hlf{l_1}!\hms{m_2},~\hlf{l_2}?\hms{m_1}.\hlf{l_2}!\hms{m_2}.\cdots.\hlf{l_2}!\hms{m_n})$. 
Here, we have $\mu_n \not\in \overline{\sigma_L(i_n)}$ but to prove it without using local analyses, we need to explore $n + 4$ vertices.
With local analyses (applied to the equivalent of \textcircled{1} and \textcircled{2} of Fig.\ref{fig:locana_principle_glocal}), whichever is the value of $n$, we only need to explore $3$ vertices.

However, the success of all local analyses on a vertex $(i, \mu)$ doesn't guarantee global analysis success. For instance let us consider:

\centerline{$\mu=(\hlf{l_1}!\hms{m},~\hlf{l_2}?\hms{m},~\hlf{l_3}?\hms{m})$ and $i=alt(strict(\hlf{l_1}!\hms{m},\hlf{l_2}?\hms{m}),strict(\hlf{l_1}!\hms{m},\hlf{l_3}?\hms{m}))$ }
\noindent All three components $\mu_{|\hlf{l}}$ of the multi-trace precisely belong to the semantics $\sigma(\lifelineElim_{L \setminus {\hlf{l}}}(i))$ of the interaction reduced to that same lifeline. Yet, the multi-trace itself is not accepted because only one of the two receptions can occur. 

As expected, local analyses alone cannot substitute a global analysis.
However, in specific cases, leveraging local analyses can enhance the performance of global analyses.
This depends on the interaction $i$, the analyzed multi-trace $\mu$, and the exploration method of the graph $\mathbb{G}$ reachable from $(i,\mu)$. This is particularly true in cases where $\mu \not\in \overline{\sigma_L(i)}$, as we have illustrated in Fig.\ref{fig:locana_principle_glocal}.

\section{Experimental evaluation\label{sec:experiments}}

We have implemented our approach in the HIBOU tool \cite{hibou_label} (version 0.8.7). It allows the parameterized exploration of graph $\mathbb{G}$ from any vertex $(i,\mu)$.

\subsection{Experimental evaluation of graph sizes \label{sec:size-graph}}

In the following, we aim to estimate the size of the graph $\mathbb{G}$. As input data, we build a benchmark as follows:
\begin{itemize}
    \item We randomly generate $100$ distinct interaction terms with $5$ lifelines and $6$ messages. This random generation proceeds inductively, drawing a random symbol. The process ends if that symbol is a constant ($\varnothing$ or an action). Otherwise, it is called recursively to generate its subterms. We require a minimum term depth of $6$ and a minimum total number of symbols (after term simplification) of $20$. 
    \item For each of these $100$ interactions, we generate:
    \begin{itemize}
        \item $240$ random accepted multi-traces of sizes between $1$ and $30$
        \item $240$ random multi-prefixes, one for each of the $240$ multi-traces
        \item $240$ ``noise'' mutants, obtained via inserting a random action in one of the $240$ multi-prefixes
        \item $240$ ``swap action'' mutants, obtained via interverting the order of two actions on a local trace component of one of the $240$ multi-prefixes
        \item $240$ ``swap component'' mutants, obtained via interverting two components on the same lifeline from two distinct multi-prefixes
    \end{itemize}
\end{itemize}

In practice, we have generated $114~794$ unique $(i,\mu)$ initial vertices.
The metric which we evaluate is the size of the following four sets of vertices, for any of the $114~794$ initial vertices $(i,\mu)$:
\[ \mathtt{graph}(i,\mu) =
\{ (i',\mu') \in \mathbb{V} ~|~ (i,\mu) \overset{*}{\leadsto} (i',\mu') \} \]
\[ \mathtt{graph}^{\textbf{P}}(i,\mu) 
=
\{ (i',\mu') \in \mathbb{V} ~|~ (i,\mu) \overset{*}{\leadsto}\vphantom{\leadsto}\hspace*{-.1cm}^{\textbf{POR}} (i',\mu') \}
\]

\[
\mathtt{graph}_{\textbf{L}}(i,\mu) = \{(i,\mu)\} \cup 
\bigcup_{(i',\mu') \in I} \mathtt{graph}_{\textbf{L}}(i',\mu') \]

\noindent
with $I = \{ (i',\mu') ~|~(i,\mu) \leadsto (i',\mu') \wedge
\anaLocal_{L,\infty}(i,\mu) \}$.

\[
\mathtt{graph}_{\textbf{L}}^{\textbf{P}}(i,\mu) 
= \{(i,\mu)\} \cup 
\bigcup_{(i',\mu') \in J} \mathtt{graph}_{\textbf{L}}^{\textbf{P}}(i',\mu') 
\]

\noindent with $J = \{ (i',\mu') ~|~(i,\mu)  \leadsto^{\textbf{POR}} (i',\mu')
 \wedge
 \anaLocal_{L,\infty}(i,\mu) \}$.

These metrics correspond to measuring the size of the sub-graph of $\mathbb{G}$ that is reachable from $(i,\mu)$ using 4 distinct restrictions of the traversal rules.
For $\mathtt{graph}(i,\mu)$, the baseline $\leadsto$ rule-set is used.
For $\mathtt{graph}^{\textbf{P}}(i,\mu)$, we use partial order reduction via $\leadsto^{\textbf{POR}}$ from Sec.\ref{ssec:partial_order_reduction}. The two other methods are variants of these two first that use local analyses as presented in Sec.\ref{ssec:local_analyses}. 
We use $\infty$ as a look-ahead depth to fully analyse all local traces when performing a local analysis.

All the details on these experiments and the means to reproduce them are given in \cite{hibou_lfrem_exp_graph_size}.
We performed the experiments on an Intel(R) Core(TM) i5-6360U CPU (2.00GHz) with 8GB RAM with HIBOU version 0.8.7.
We set a $3$ seconds timeout for the traversal of the graph.

\begin{figure}[ht]
    \centering
\scalebox{.875}{
\begin{tabular}{|l|r|r|r|r|r|r|r|r|}
\hline 
\multirow{2}{*}{}
&
\multirow{2}{*}{\textbf{ACPT}}
&
\multirow{2}{*}{\textbf{PREF}}
&
\multicolumn{2}{c|}{\textbf{NOIS}}
&
\multicolumn{2}{c|}{\textbf{SACT}}
&
\multicolumn{2}{c|}{\textbf{SCMP}}
\\
\cline{4-9} 
&
&
&
$\macroOKVerdict$
&
$\macroKOVerdict$
&
$\macroOKVerdict$
&
$\macroKOVerdict$
&
$\macroOKVerdict$
&
$\macroKOVerdict$
\\
\hline 
{\footnotesize \textbf{TOTAL}}
&
22498
&
23059
&
2533
&
20571
&
16690
&
6376
&
22451
&
616
\\
\hline 
{\footnotesize $\mathtt{timeout}$}
&
1369
&
148
&
26
&
196
&
129
&
15
&
118
&
11
\\
\hline 
{\footnotesize $\mathtt{timeout}_{\mathbf{LOC}}$}
&
1308
&
143
&
23
&
1
&
126
&
0
&
112
&
3
\\
\hline 
{\footnotesize $\mathtt{timeout}^{\mathbf{POR}}$}
&
30
&
1
&
1
&
60
&
0
&
6
&
0
&
1
\\
\hline 
{\footnotesize $\mathtt{timeout}_{\mathbf{LOC}}^{\mathbf{POR}}$}
&
9
&
1
&
0
&
0
&
1
&
0
&
1
&
0
\\
\hline 
\end{tabular}
  }  
    \caption{Occurrences of timeouts when computing graph size}
    \label{fig:graph_size_timeouts}
\end{figure}

Fig.\ref{fig:graph_size_timeouts} presents results related to the timeouts in a table.
Its columns correspond to the type of multi-trace that is involved, full multi-traces, multi-prefixes, the three kinds of mutants (further divided into those that are still within the multi-prefix semantics of the corresponding interaction and those that are not).
The first row counts the total number of data points (whether or not there has been a timeout using any method).
The second row counts the number of timeouts when exploring $\mathtt{graph}(i,\mu)$ while for the third, fourth and fifth it is when exploring respectively $\mathtt{graph}_{\mathbf{L}}(i,\mu)$, $\mathtt{graph}^{\mathbf{P}}(i,\mu)$ and $\mathtt{graph}_{\mathbf{L}}^{\mathbf{P}}(i,\mu)$.

We observe that the overall number of timeouts decreases whenever we use the partial order reduction technique (two last rows), whether or not we use local analyses.
The number of timeouts decreases when using local analyses only when considering multi-traces that are $\macroKOVerdict$. This is particularly the case when considering \textbf{NOIS} $\macroKOVerdict$ multi-traces, whether or not we also use partial order reduction.

\begin{figure}[ht]

\resizebox{\textwidth}{!}{
\begin{tabular}{|c|c|c|}
\cline{2-3}
\multicolumn{1}{c|}{}
&
{\footnotesize$\neg$\textbf{POR}}
&
{\footnotesize\textbf{POR}}
\\
\cline{1-1}
\multicolumn{1}{|c}{
\raisebox{9\normalbaselineskip}[0pt][0pt]{\rotatebox[origin=c]{90}{{\footnotesize$\neg$\textbf{LOC}}}}
}
&
\includegraphics[scale=1]{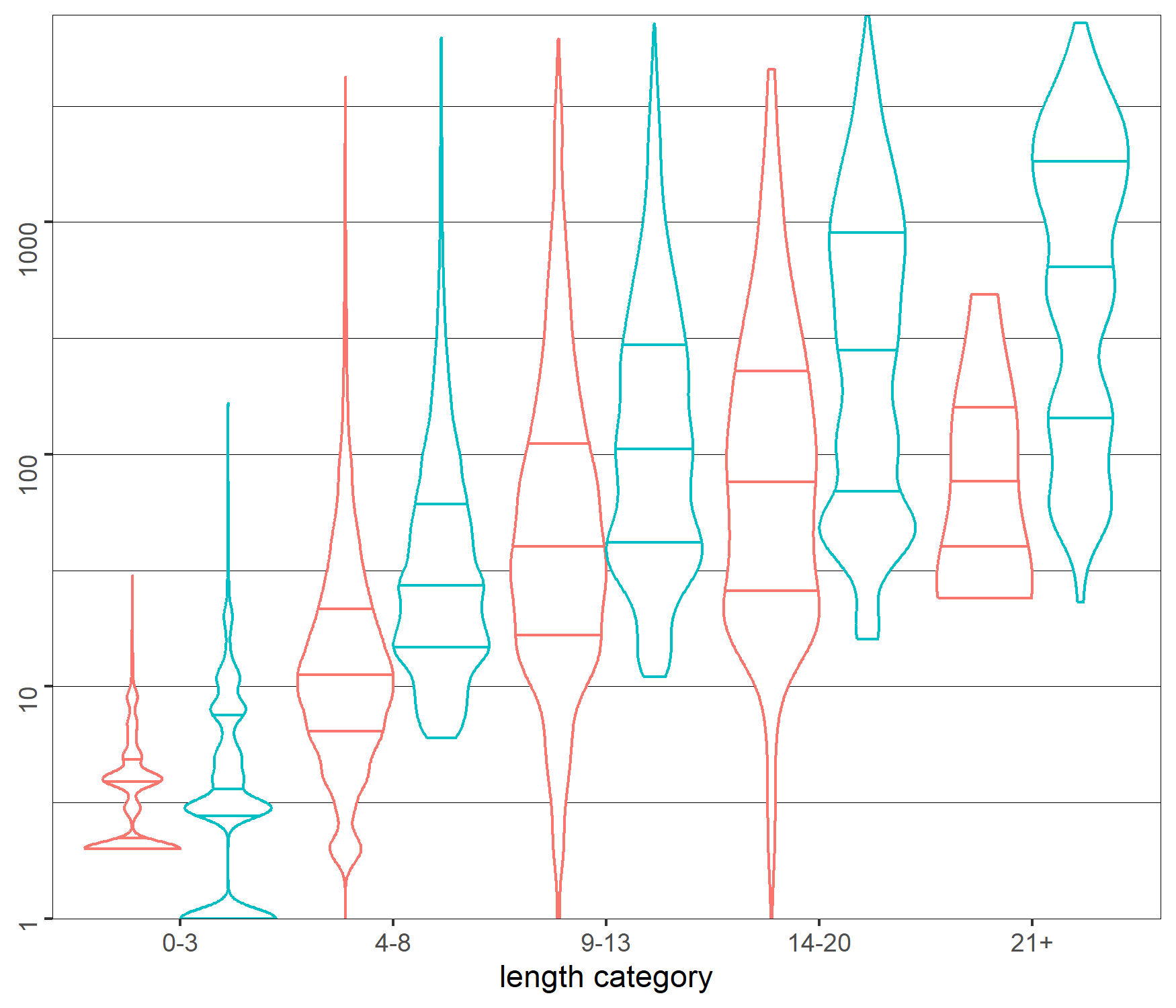}
&
\includegraphics[scale=1]{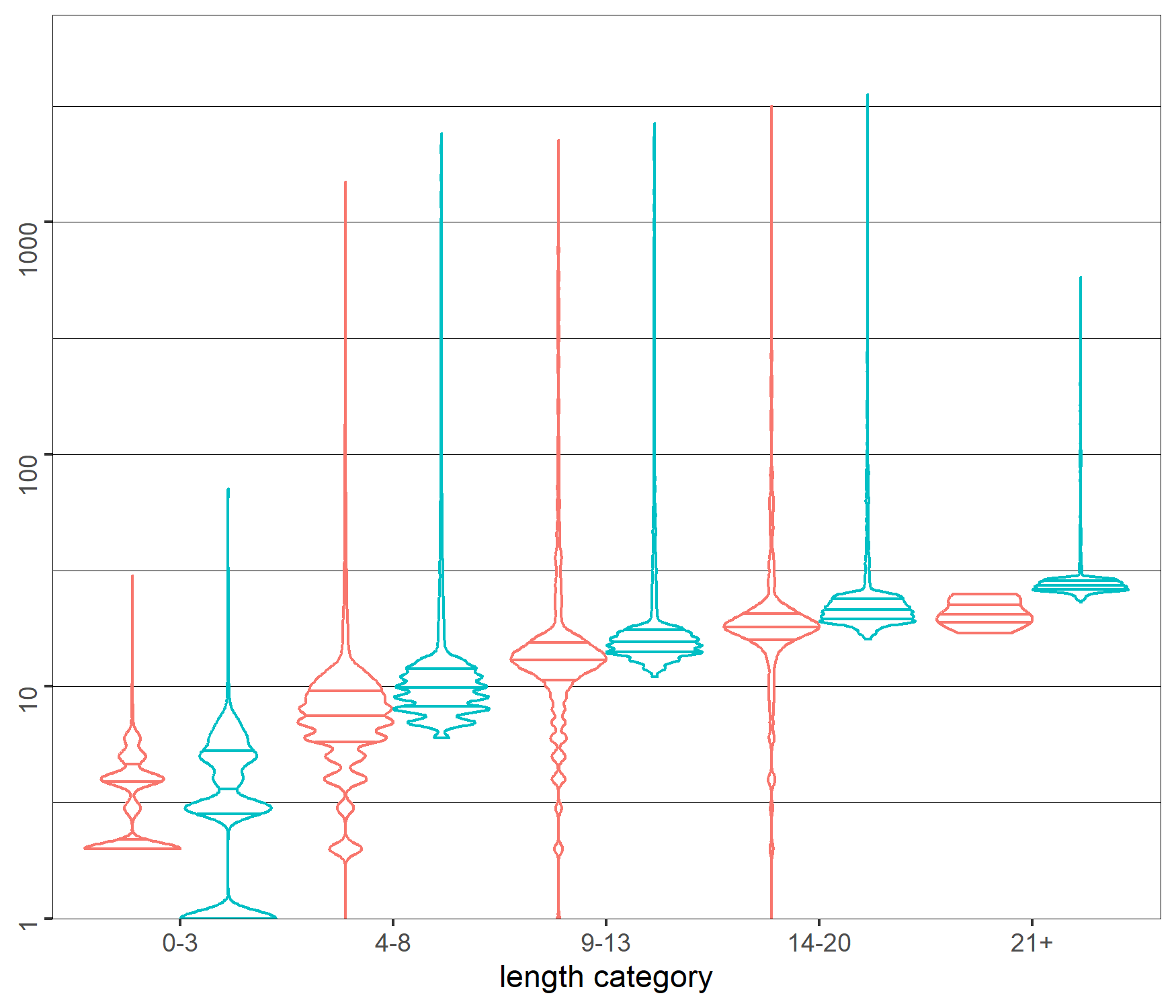}
\\
\hline 
\multicolumn{1}{|c}{
\raisebox{9\normalbaselineskip}[0pt][0pt]{\rotatebox[origin=c]{90}{{\footnotesize\textbf{LOC}}}}
}
&
\includegraphics[scale=1]{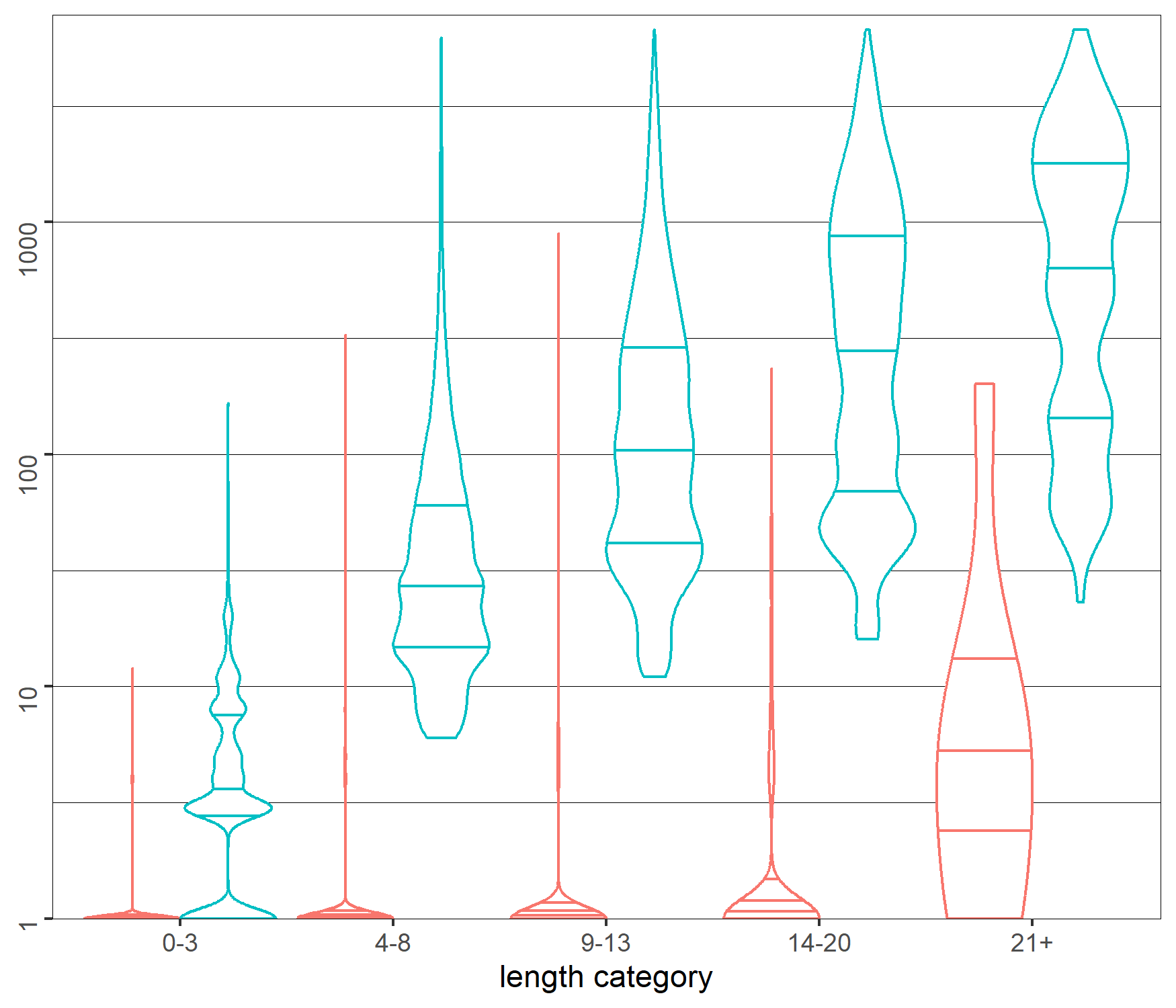}
&
\includegraphics[scale=1]{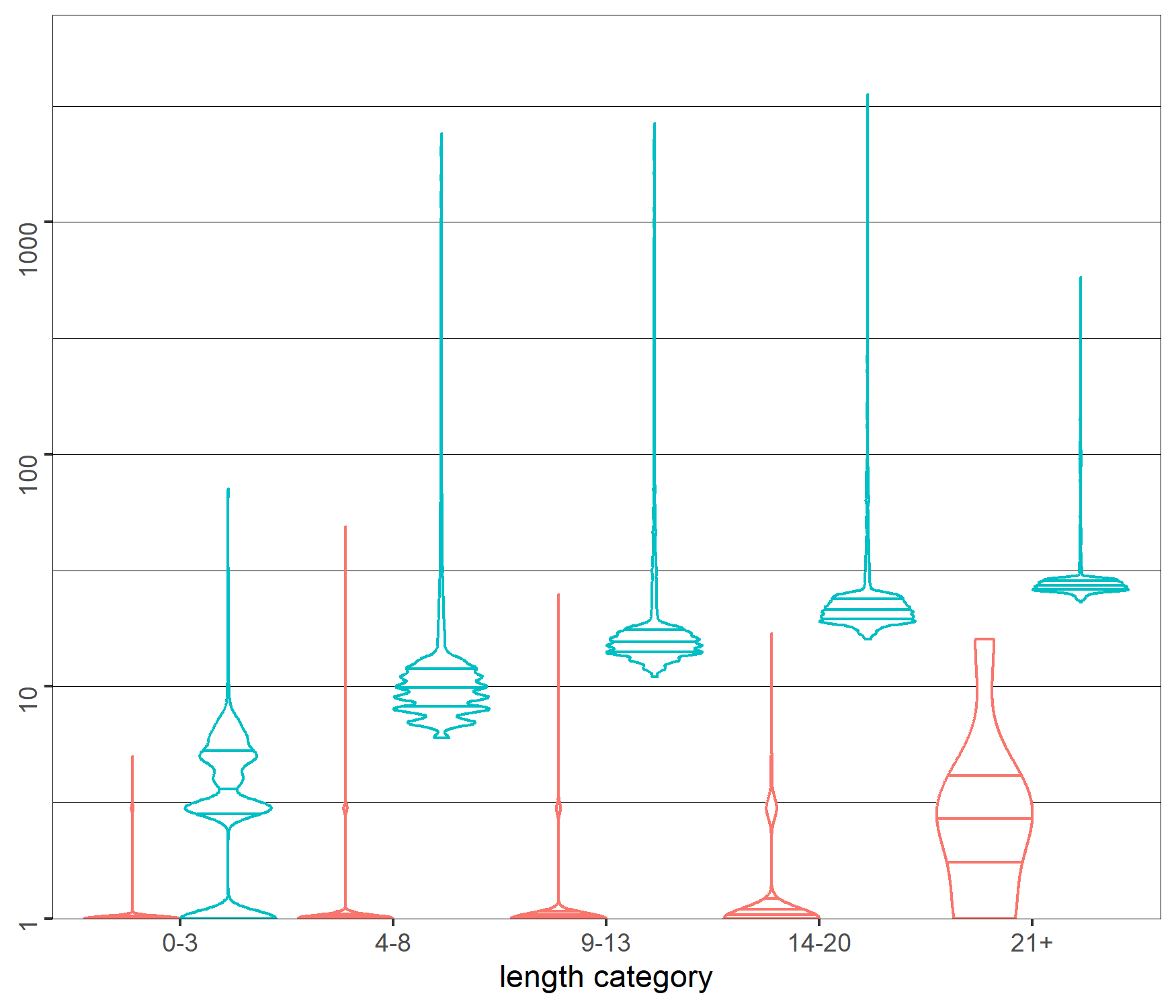}
\\
\hline 
\end{tabular}
}
    
    \caption{Effect of partial order reduction \& local analysis on graph size ($log_{10}$ scale)}
    \label{fig:graph_size_experiments}
\end{figure}

Fig.\ref{fig:graph_size_experiments} compares the distributions of graph sizes on our dataset depending on the method that is used (with or without POR and LOC) and the kind of multi-trace (by length and $\macroOKVerdict$ or $\macroKOVerdict$ verdict).
To have a fair comparison, we only consider $(i,\mu)$ that didn't return a timeout with any of the four methods.
This amounts to $112~755$ unique vertices $(i,\mu)$.

Fig.\ref{fig:graph_size_experiments} presents the results on a double-entry table.
The two columns $\neg \textbf{POR}$ and $\textbf{POR}$ correspond to the non-use and use of the Partial Order Reduction technique from Sec.\ref{ssec:partial_order_reduction} while the two rows $\neg \textbf{LOC}$ and $\textbf{LOC}$ refer to the Local Analysis technique from Sec.\ref{ssec:local_analyses}. Hence, each table cell corresponds to one of the previous $|\mathtt{graph}^X_Y|$ metrics.

We represent statistical distributions of graph sizes from our dataset on each cell via violin plots \cite{violin_plots}.
The $y$ axis corresponds to the graph size and is in logarithmic (base 10) scale.
These distributions are further divided according to whether or not the multi-trace is a correct prefix (via the blue and red colors) and according to the size of the multi-trace (we gather them into groups per length category, which correspond to the $x$ axis).
Derived from classical box-plots, these diagrams also represent the density around values in the fashion of kernel density plots \cite{violin_plots}.
The wider the violin is at a certain $y$ value, the more there are data points in its neighborhood\footnote{we use an additional scaling option so that all violins have the same maximal width, see \url{https://ggplot2.tidyverse.org/reference/geom_violin.html}}.
The 3 horizontal lines correspond, from bottom to top, to the 1st quartile, median and 3rd quartile.

We observe that the effect of POR is to reduce graph size globally. 
The maximal values are always reduced, but the statistical effect is most important on the median value of graph sizes, which seem to follow a linear curve (in $log_{10}$ scale) with POR and an exponential curve without.
This change in the shape of the distribution is most visible without local analyses and on $\macroOKVerdict$ traces with local analyses (they appear packed at the bottom, around the median value).

The effect of LOC is very impressive when considering $\macroKOVerdict$ verdict multi-traces. This is not surprising in the case of \textbf{NOIS} $\macroKOVerdict$ multi-traces as there can be an immediate failure of local analysis at the start, from the initial $(i,\mu)$ vertex, the graph then being of size $1$.
However, a distinct effect can be observed for longer multi-traces (length category 14-20 and 21$+$) as the graph size is consequently reduced.
Indeed, LOC enables avoiding parts of the graph following a wrong choice that has been made (e.g., the choice of a branch of an alternative, of instantiating a loop, etc.).

\subsection{Experimental evaluation of analysis time}


In all generality, the search graph $\mathbb{G}$ is not known in advance, and exploring it in its entirety is often unnecessary.
Indeed, it suffices to find a single path from $(i,\mu)$ to $\macroOKVerdict$ to prove that $\mu \in \overline{\sigma_L(i)}$. Upon reaching $\macroOKVerdict$, the algorithm can immediately stop, further graph exploration being unnecessary.
As a result, the order of graph traversal plays a crucial role, and the time performances of the algorithm depend on implementation-dependent search strategies and heuristics.
Of course, this only concerns the analysis of correct multi-prefixes, because, to prove that $\mu \not\in \overline{\sigma_L(i)}$, we still need to explore the whole graph.

\begin{figure}[ht]
    \centering
\scalebox{.875}{
\begin{tabular}{|l|r|r|r|r|r|r|r|r|}
\hline 
\multirow{2}{*}{}
&
\multirow{2}{*}{\textbf{ACPT}}
&
\multirow{2}{*}{\textbf{PREF}}
&
\multicolumn{2}{c|}{\textbf{NOIS}}
&
\multicolumn{2}{c|}{\textbf{SACT}}
&
\multicolumn{2}{c|}{\textbf{SCMP}}
\\
\cline{4-9} 
&
&
&
$\macroOKVerdict$
&
$\macroKOVerdict$
&
$\macroOKVerdict$
&
$\macroKOVerdict$
&
$\macroOKVerdict$
&
$\macroKOVerdict$
\\
\hline 
{\footnotesize \textbf{TOTAL}}
&
23114
&
23114
&
2543
&
20565
&
16738
&
6376
&
22498
&
616
\\
\hline 
{\footnotesize $\mathtt{timeout}$}
&
192
&
1
&
1
&
174
&
1
&
15
&
3
&
10
\\
\hline 
{\footnotesize $\mathtt{timeout}_{\mathbf{LOC}}$}
&
0
&
0
&
0
&
13
&
0
&
0
&
0
&
3
\\
\hline 
{\footnotesize $\mathtt{timeout}^{\mathbf{POR}}$}
&
37
&
0
&
0
&
52
&
0
&
6
&
0
&
0
\\
\hline 
{\footnotesize $\mathtt{timeout}_{\mathbf{LOC}}^{\mathbf{POR}}$}
&
0
&
0
&
0
&
0
&
0
&
0
&
0
&
0
\\
\hline 
\end{tabular}
}    
    \caption{Occurrences of timeouts for analyses}
    \label{fig:time_timeouts}
\end{figure}

In the following, we use the same benchmark as in Sec.\ref{sec:size-graph}. Still, instead of exploring the graph in its entirety, we traverse it using a simple Depth First Search with a stopping criterion upon finding $\macroOKVerdict$.
This allows us to focus on the time required for the analysis and contrast results from Sec.\ref{sec:size-graph}.
Details on this second set of experiments are available in \cite{hibou_lfrem_exp_ana_time}.

Using the same $3$ seconds timeout, we observe, in Fig.\ref{fig:time_timeouts}, different results than in Fig.\ref{fig:graph_size_timeouts} concerning the occurrences of timeouts. Indeed, the fact that we do not need to compute the whole graph greatly improves performances for $\macroOKVerdict$ multi-traces, which mechanically reduces the number of timeouts (especially for \textbf{ACPT}).

\begin{figure}[ht]

\resizebox{\textwidth}{!}{
\begin{tabular}{|c|c|c|}
\cline{2-3}
\multicolumn{1}{c|}{}
&
{\footnotesize$\neg$\textbf{POR}}
&
{\footnotesize\textbf{POR}}
\\
\cline{1-1}
\multicolumn{1}{|c}{
\raisebox{9\normalbaselineskip}[0pt][0pt]{\rotatebox[origin=c]{90}{{\footnotesize$\neg$\textbf{LOC}}}}
}
&
\includegraphics[scale=1]{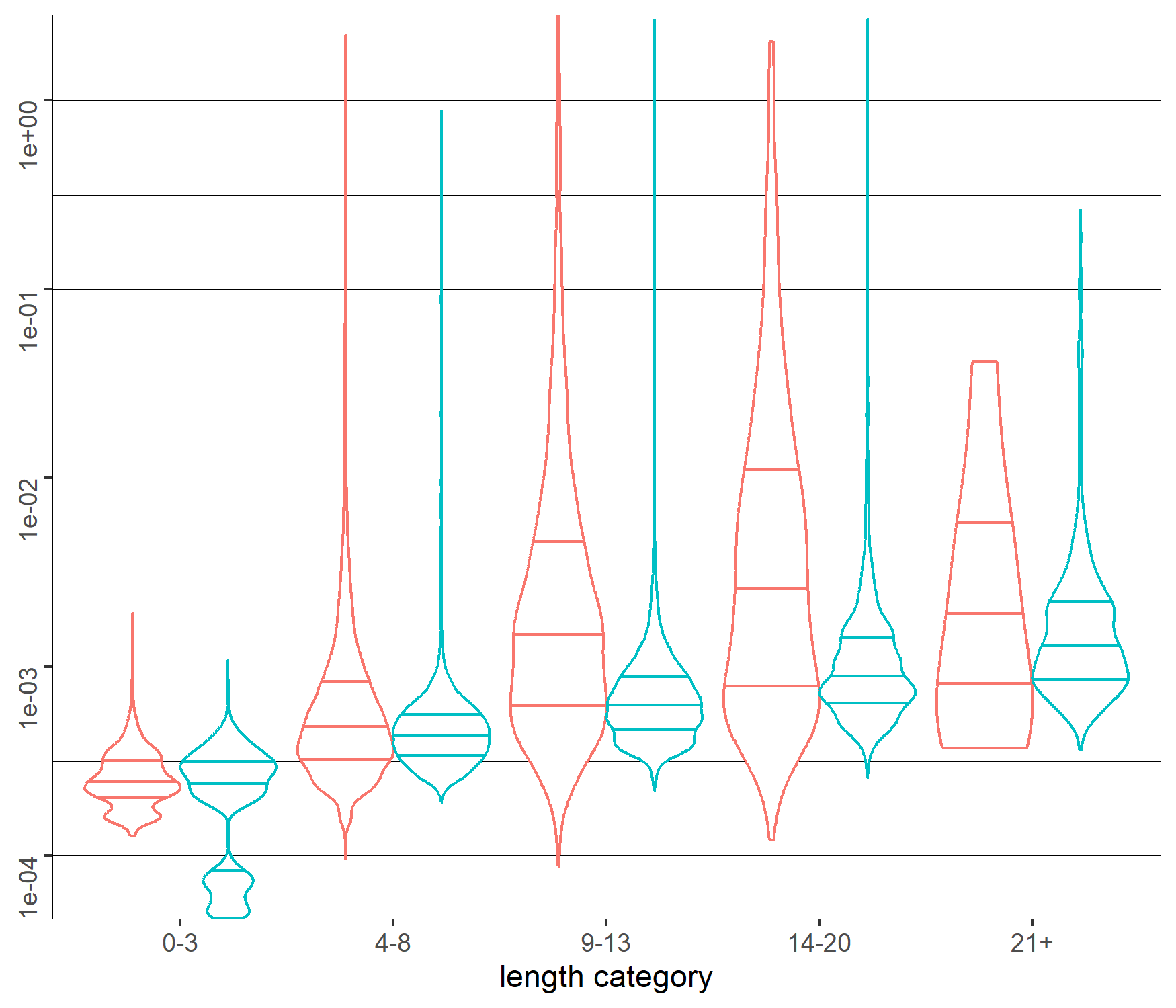}
&
\includegraphics[scale=1]{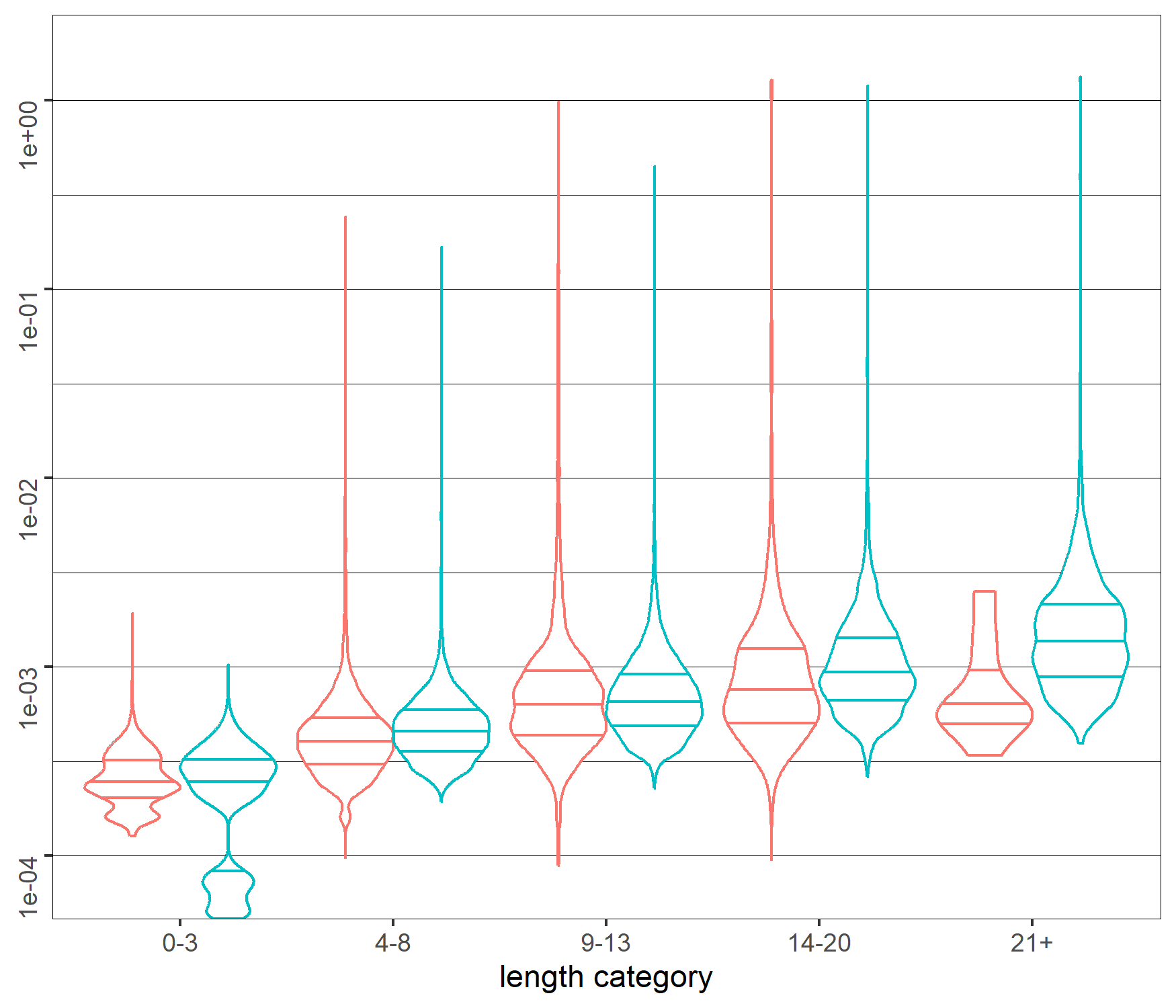}
\\
\hline 
\multicolumn{1}{|c}{
\raisebox{9\normalbaselineskip}[0pt][0pt]{\rotatebox[origin=c]{90}{{\footnotesize\textbf{LOC}}}}
}
&
\includegraphics[scale=1]{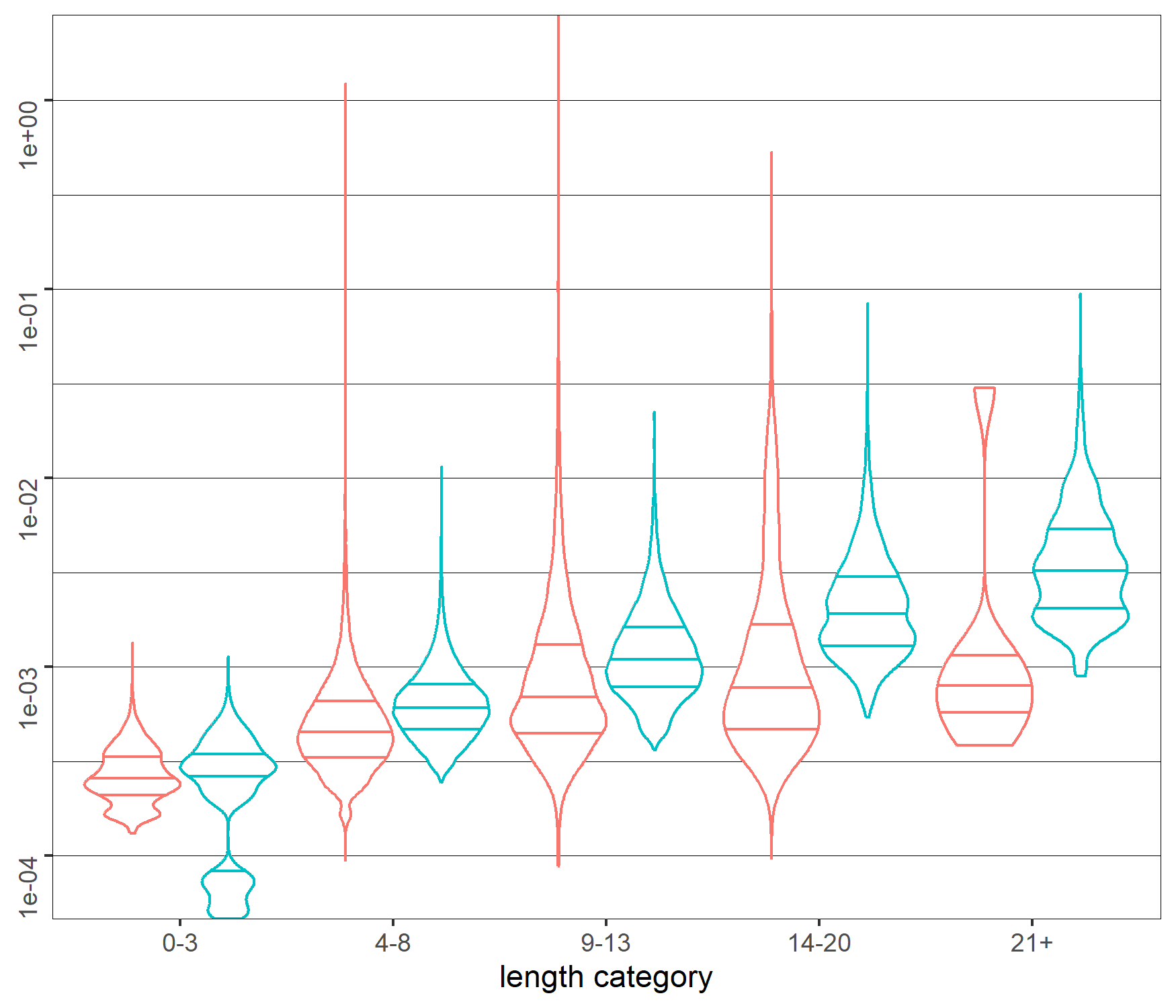}
&
\includegraphics[scale=1]{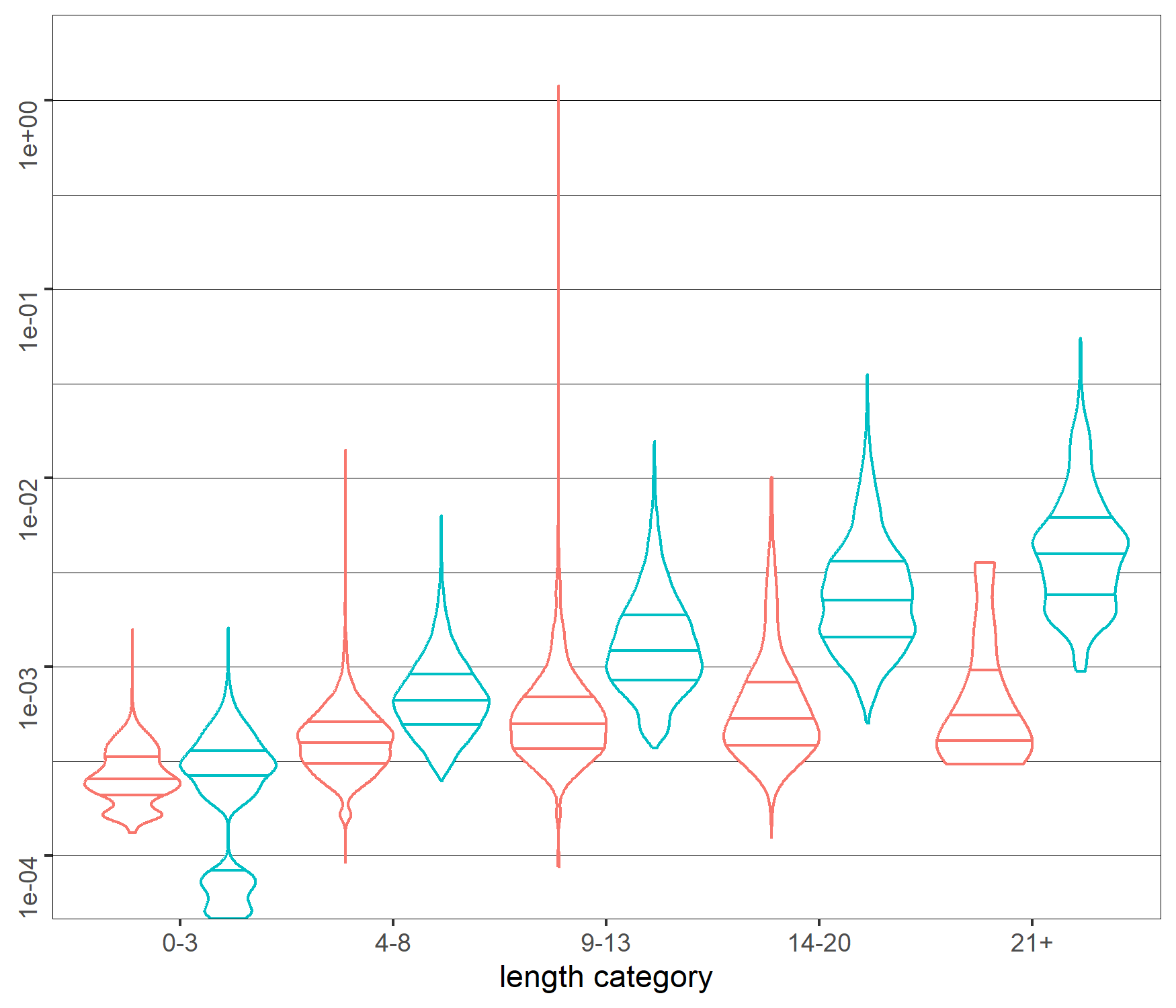}
\\
\hline 
\end{tabular}
}
    
    \caption{Effect of \textbf{POR} \& \textbf{LOC} on analysis time ($log_{10}$ scale)}
    \label{fig:time_experiments}
\end{figure}

When plotting the analysis time in the same manner as we plotted the graph size (on Fig.\ref{fig:graph_size_experiments}), we obtain the diagrams on Fig.\ref{fig:time_experiments}.

The use of DFS and a stopping criterion generally makes the analysis of $\macroOKVerdict$ multi-traces less time-consuming than their $\macroKOVerdict$ counterparts of similar length in the baseline case (i.e., we use neither \textbf{POR} nor \textbf{LOC}).
Indeed, while for the former, we do not need to compute the whole graph, it remains the case for the latter, and, without additional techniques, the graph size may remain similar. 
This can be observed on the top left of Fig.\ref{fig:time_experiments} (looking at the median and quartiles of each distribution).

The effect of \textbf{LOC} remains the same as in Fig.\ref{fig:graph_size_experiments} (especially without \textbf{POR}). However, it is not as spectacular because we do not measure graph size but analysis time and, while for the former, the cost of local analyses is hidden, it is not the case for the latter (the time for the global analysis includes the sum of the times for all local analyses). This also causes the times for $\macroOKVerdict$ multi-traces with \textbf{POR} and \textbf{LOC} to be slightly higher than that for $\macroOKVerdict$ multi-traces with \textbf{POR} but without \textbf{LOC}. Indeed, when analyzing correct behaviors, local analyses may only incur an overhead without necessarily reducing the number of vertices that need to be explored.
This overhead effect increases with the size of the multi-trace but can be mitigated via manipulating the look-ahead depth $\depthAnaLoc$ or using additional heuristics to condition the use of local analyses on specific lifelines.

From these experiments, a conservative conclusion would be that:
\begin{itemize}
    \item if we know that the multi-trace is correct and want to prove it, we should use \textbf{POR} without \textbf{LOC}
    \item if we know that the multi-trace exhibits a failure and want to prove it, we should use \textbf{POR} and \textbf{LOC}
    \item if we do not know whether the multi-trace is correct or not, we may run in two separate threads an analysis with \textbf{POR} and \textbf{LOC} and another one with only \textbf{POR} and get the verdict from the quickest method.
\end{itemize}

\section{Related works\label{sec:related}}

In the following, we discuss related works, primarily concerned with the application of Runtime Verification (RV) techniques to Distributed Systems (DS). 
We first delve into RV approaches that utilize reference specifications of interaction models. We examine both offline and online approaches, in particular scrutinizing their treatment of the challenge posed by partial observation. Subsequently, we shift our attention to RV studies for DS against references of temporal properties, thereby adapting the more prevalent online form of centralized RV to suit the distributed context. Finally, we explore verifying basic Message Sequence Charts (MSCs) against High-level MSCs as a design challenge, highlighting the complexity of membership assessment.

In \cite{Hierons14, passive_conformance_testing_of_service_choreographies, interaction_based_runtime_verification_for_systems_of_systems_integration_KrugerMM10, monitoring_networks_through_multiparty_session_types_BocchiCDHY17, InckiA18}, the focus is on RV verifying distributed executions against models of interaction. While \cite{interaction_based_runtime_verification_for_systems_of_systems_integration_KrugerMM10, Hierons14, InckiA18} concern MSC, \cite{passive_conformance_testing_of_service_choreographies} considers choreographic languages, \cite{monitoring_networks_through_multiparty_session_types_BocchiCDHY17} session types, and \cite{coping_with_bad_agent_interaction_protocols_when_monitoring_partially_observable_multiagent_systems_AnconaFFM18} trace expressions. We discuss these works according to their relevance to offline RV, which is the scope of our approach, and online RV.

The works~\cite{Hierons14,passive_conformance_testing_of_service_choreographies} propose offline RV that relies on synchronization hypotheses and on reconstructing a global trace by ordering events occurring at the distributed interfaces (by exploiting the observational power of testers~\cite{Hierons14} or timestamp information assuming clock synchronization~\cite{passive_conformance_testing_of_service_choreographies}). Our RV approach for multi-traces does not require synchronization prerequisites on logging. Thus, unlike previous works on offline RV, we can analyze executions of DS without the need for a synchronization hypothesis on the ending of local observations. 

For online RV, \cite{InckiA18} depends on a global component, specifically a network sniffer. Conversely, the works~\cite{interaction_based_runtime_verification_for_systems_of_systems_integration_KrugerMM10} and~\cite{monitoring_networks_through_multiparty_session_types_BocchiCDHY17} propose local RV against local projections of interactions. \cite{interaction_based_runtime_verification_for_systems_of_systems_integration_KrugerMM10} is focused on the instrumentation of online RV using aspect-oriented programming techniques to inject monitors into the code of subsystems without modifying it for a non-intrusive verification. 
\cite{monitoring_networks_through_multiparty_session_types_BocchiCDHY17} derives consistent monitors by considering interactions satisfying conditions that enforce intended global behaviors. Unlike these works, our approach involves processing collections of local logs against interactions while leveraging projection techniques to handle partial observation and to optimize performances.

\cite{coping_with_bad_agent_interaction_protocols_when_monitoring_partially_observable_multiagent_systems_AnconaFFM18} focuses on how distributed monitors can be adapted for partial observation. Yet, our notion of partial observation is distinct from that of \cite{coping_with_bad_agent_interaction_protocols_when_monitoring_partially_observable_multiagent_systems_AnconaFFM18} where messages are exchanged via channels which are associated to an observability likelihood. \cite{coping_with_bad_agent_interaction_protocols_when_monitoring_partially_observable_multiagent_systems_AnconaFFM18} uses trace expressions as specifications and proposes transformations that can adapt those expressions to partial observation by removing or making optional some identified unobservable events. Instead, we deal with partial observation from the perspective of analyzing truncated multi-traces due to synchronization issues.

Interaction models are well-suited for specifying DS because they specifically capture interaction patterns and causal communication relationships between subsystems.
Yet, as highlighted in the specification section of the taxonomy of RV tools given in \cite{a_taxonomy_for_classifying_runtime_verification_tools_FalconeKRT21}, interactions are scarcely utilized in RV. 
This is not the case for logical properties extensively employed as reference specifications. 
Linear Temporal Logic (LTL) is particularly prevalent in RV, with its semantics generally given as sets of traces. \cite{SenVAR04} expands a variant of LTL where formulas concern subsystems and their knowledge about other subsystems' local states. It considers a collection of decentralized observers that share information about the subsystem executions that affect the validity of the formula. In contrast, other works such as \cite{Falcone16,El-HokayemF17} express properties at the system level and transform them into decentralized observers using LTL formula rewriting. This approach eliminates the need for a centralized verifier to gather complete information on the system's execution. Considering interaction languages aligns with the direction of recent RV efforts, focusing on comprehensive and targeted verification of DS against global specifications.

Earlier works~\cite{realizability_and_verification_of_msc_graphs,pattern_matching_and_membership_for_hierarchical_message_sequence_charts} considered checking basic MSCs against HMSC (High-level MSC, which are graphs of MSCs) as an MSC membership problem. Roughly speaking, a basic MSC represents a multi-trace and may define a desired or undesired scenario thus resolving membership can aid in identifying design errors or avoiding redundancy. 
In those works, some MSCs are marked as accepting within an HMSC, and a basic MSC belongs to the semantics of the graph if and only if it fully covers accepting (finite) paths in the graph. Consequently, partially observed multi-traces cannot be assessed against HMSC, which does not answer the RV problem under observability limitations. These works demonstrated the NP-hardness of MSC membership by reducing it to the 1-in-3-SAT problem~\cite{1in3SAT78}. The RV problem we tackle in this paper aligns with this complexity, as we establish membership of prefixes of multi-traces through a reduction to the 3SAT problem. This computational cost prompts us to address optimization aspects in our work by proposing 1-unambiguous Partial Order Reduction (POR) and local analyses (LOC) techniques to reduce the size of the search space.

\section{Conclusion\label{sec:conclusion}}

This paper extends previous works~\cite{interaction_based_offline_runtime_verification_of_distributed_systems} in which we proposed an RV approach to monitor DS executions against interactions. 
This approach mitigates an intrinsic problem that occurs when observing DS executions: the inability to synchronize the ending of observations of local executions due to the absence of a global clock.
In our approach, DS executions are represented as tuples of local traces, called multi-traces, in which each local trace corresponds to the local observation of a subsystem execution.
As formal specifications, we consider interaction models in which each subsystem corresponds to a unique lifeline. Our RV process corresponds to re-enacting the observed behavior on the interaction model by executing execution steps derived from the operational semantics of interactions.
We address partial observation via a lifeline removal operator, which enables us to remove from the interaction the specification of all atomic events that occur on a subsystem that is no longer observed (when the corresponding local observation has been fully re-enacted).

We have proposed two techniques to improve the performances of our RV algorithm that both rely on local projections of the interaction model via lifeline removal.
The first one implements a form of partial order reduction by identifying one-unambiguous actions in the interaction model.
During the analysis, if one such action is encountered, we can select it as a unique successor of the current node, thus reducing the size of the search graph.
The second is to carry out local analyses at each exploration step, i.e. to check that the local traces are consistent with the interactions resulting from the projection onto the lifeline concerned. If not, then the current node can be discarded.
Extensive experimentation suggests that partial order reduction improves performances in all cases by reducing the size of the search space. As for local analyses, while their use can significantly improve performances for the analysis of erroneous multi-traces, their impact is more nuanced for correct multi-traces as they incur a potentially non-negligible time overhead without necessarily causing an important reduction in the size of the search space.

\bibliographystyle{elsarticle-harv} 
\bibliography{
biblio/biblio,
biblio/interaction_languages,
biblio/others,
biblio/own,
biblio/por,
biblio/standards
}

\appendix

\section{Proofs of Section \ref{sec:core}\label{app:section2}}

\begin{property*}[Prop.\ref{prop:projection_on_algebraic_operators}, Sec.\ref{sec:multitraces}]
For any $L \subseteq \mathcal{L}$ and any $t_1$ and $t_2$ in $\mathbb{A}(L)^*$:
\[
\begin{array}{rclcrcl}
\muProjection_L(t_1 \cup t_2) & = & \muProjection_L(t_1) \cup \muProjection_L(t_2)
&~~~~&
\muProjection_L(t_1 \opStrictSeq t_2) & = & \muProjection_L(t_1) \opStrictSeq \muProjection_L(t_2)\\
\muProjection_L(t_1 \opInterleaving t_2) & = & \muProjection_L(t_1) \opInterleaving \muProjection_L(t_2)
&~~~~&
\muProjection_L(t_1 \opWeakSeq t_2) & = & \muProjection_L(t_1) \opStrictSeq \muProjection_L(t_2)
\end{array}
\]
\end{property*}

\begin{proof}
The case of $\cup$ is trivial. Those of $\opStrictSeq$ and $\opInterleaving$ are immediate because $\muProjection_L(a.t) = a \multiAppend \muProjection_L(t)$.
For $\opWeakSeq$ let us reason by induction:
\[
\begin{array}{rcl}
\muProjection_L(t_1 \opWeakSeq \varepsilon)
&
=
&
\muProjection_L(t_1) = \muProjection_L(t_1) \opStrictSeq \varepsilon_L = \muProjection_L(t_1) \opStrictSeq \muProjection_L(\varepsilon)
\\
\muProjection_L(\varepsilon \opWeakSeq t_2)
&
=
&
\muProjection_L(t_2) = \varepsilon_L \opStrictSeq \muProjection_L(t_2) = \muProjection_L(\varepsilon) \opStrictSeq \muProjection_L(t_2)
\\
\muProjection_L(a_1.t_1 \opWeakSeq a_2.t_2)
&
=
&
\muProjection_L(a_1.(t_1 \opWeakSeq a_2.t_2))~~~~~~~~~~~~~~~~
\text{ if $(a_1.t_1) \conflictPredicate \theta(a_2)$}
\\
&
=
&
a_1 \multiAppend \muProjection_L(t_1 \opWeakSeq a_2.t_2)
\\
&
=
&
a_1 \multiAppend (\muProjection_L(t_1) \opStrictSeq \muProjection_L(a_2.t_2))
\\
&
=
&
(a_1 \multiAppend \muProjection_L(t_1)) \opStrictSeq \muProjection_L(a_2.t_2)
\\
&
=
&
\muProjection_L(a_1.t_1) \opStrictSeq \muProjection_L(a_2.t_2)
\end{array}
\]
For the case $\neg (a_1.t_1) \conflictPredicate \theta(a_2)$, it suffices to observe that this implies that\\\noindent$\muProjection_L(a_1.t_1)_{|\theta(a_2)} = \varepsilon$ and hence $a_2 \multiAppend \muProjection_L(a_1.t_1) = \muProjection_L(a_1.t_1) \multiAppend a_2$ and thus:
\[
\begin{array}{rcl}
\muProjection_L(a_2.(a_1.t_1 \opWeakSeq t_2))
&
=
&
a_2 \multiAppend \muProjection_L(a_1.t_1 \opWeakSeq \muProjection_L(t_2))
\\
&
=
&
a_2 \multiAppend (\muProjection_L(a_1.t_1) \opStrictSeq \muProjection_L(t_2))
\\
&
=
&
(a_2 \multiAppend \muProjection_L(a_1.t_1)) \opStrictSeq \muProjection_L(t_2)
\\
&
=
&
(\muProjection_L(a_1.t_1) \multiAppend a_2) \opStrictSeq \muProjection_L(t_2)
\\
&
=
&
\muProjection_L(a_1.t_1) \opStrictSeq ( a_2 \multiAppend \muProjection_L(t_2) )
\\
&
=
&
\muProjection_L(a_1.t_1) \opStrictSeq \muProjection_L(a_2.t_2)
\\
\end{array}
\]
Finally:
\[
\begin{array}{ll}
& 
\muProjection_L(a_1.t_1 \opWeakSeq a_2.t_2)
\\
=
&
\muProjection_L(a_1.(t_1 \opWeakSeq a_2.t_2)) \cup \muProjection_L(a_2.(a_1.t_1 \opWeakSeq t_2))
\\
=
&
(\muProjection_L(a_1.t_1) \opStrictSeq \muProjection_L(a_2.t_2)) \cup (\muProjection_L(a_1.t_1) \opStrictSeq \muProjection_L(a_2.t_2))
\\
=
&
\muProjection_L(a_1.t_1) \opStrictSeq \muProjection_L(a_2.t_2)
\end{array}
\]
\end{proof}

\begin{property*}[Prop.\ref{prop:multi_trace_elimination_preserves_sched}, Sec.\ref{sec:lf_tr_removal}]
For any $L \in \mathcal{L}$, any $H \subseteq L$, any $\mu_1$ and $\mu_2$ in $\mathbb{M}(L)$, for any $\diamond \in \{\cup,\opStrictSeq,\opInterleaving\}$, we have:
\[
\lifelineElim_H(\mu_1 \diamond \mu_2) = \lifelineElim_H(\mu_1) \diamond \lifelineElim_H(\mu_2)
\]
\end{property*}

\begin{proof}
\noindent For $\diamond = \cup$, $\lifelineElim_\ell(\mu_1 \cup \mu_2) = \lifelineElim_\ell(\{\mu_1,\mu_2\})$ by definition of the $\cup$ operator between multi-traces
and $\lifelineElim_\ell(\mu_1) \cup \lifelineElim_\ell(\mu_2) = \{ \lifelineElim_\ell(\mu_1) , \lifelineElim_\ell(\mu_2) \}$ by definition of the $\cup$ operator between multi-traces.

\noindent For $\diamond = \opStrictSeq$, let us reason by induction on $\mu_2$:
\begin{itemize}
    \item $\lifelineElim_\ell(\mu_1 \opStrictSeq \varepsilon_L) = \lifelineElim_\ell(\mu_1) = \lifelineElim_\ell(\mu_1) \opStrictSeq  \varepsilon_{L'} = \lifelineElim_\ell(\mu_1) \opStrictSeq \lifelineElim_\ell(\varepsilon_L)$
    \item if $\mu_2 = a \multiAppend \mu_2'$ then 
    \[
    \begin{array}{lclr}
    \lifelineElim_\ell(\mu_1 \opStrictSeq \mu_2) 
    &
    =
    &
    \lifelineElim_\ell((\mu_1 \multiAppend a) \opStrictSeq \mu_2')
    &
    ~~~~\text{by definition of } \opStrictSeq
    \\
    &
    =
    &
    \lifelineElim_\ell((\mu_1 \multiAppend a)) \opStrictSeq \lifelineElim_\ell(\mu_2')
    &
    ~~~~\text{by induction}
    \end{array}
    \]
    Then:
    \begin{itemize}
        \item if $\theta(a) = h$ we have:
        \[
        \begin{array}{lclr}
        \lifelineElim_\ell(\mu_1 \opStrictSeq \mu_2) 
        &
        =
        &
        \lifelineElim_\ell(\mu_1) \opStrictSeq \lifelineElim_\ell(\mu_2')
        &
        ~~~~\text{by Prop.\ref{lem:elimination_append}}
        \\
        &
        =
        &
        \lifelineElim_\ell(\mu_1) \opStrictSeq \lifelineElim_\ell(a \multiAppend\mu_2')
        &
        ~~~~\text{by Prop.\ref{lem:elimination_append}}
        \\
        &
        =
        &
        \lifelineElim_\ell(\mu_1) \opStrictSeq \lifelineElim_\ell(\mu_2)
        &
        \end{array}
        \]
        \item if $\theta(a) \neq h$ we have:
        \[
        \begin{array}{lclr}
        \lifelineElim_\ell(\mu_1 \opStrictSeq \mu_2) 
        &
        =
        &
        (\lifelineElim_\ell(\mu_1) \multiAppend a) \opStrictSeq \lifelineElim_\ell(\mu_2')
        &
        ~~~~\text{by Prop.\ref{lem:elimination_append}}
        \\
        &
        =
        &
        \lifelineElim_\ell(\mu_1) \opStrictSeq (a \multiAppend\lifelineElim_\ell(\mu_2'))
        &
        ~~~~\text{by definition of }\opStrictSeq
        \\
        &
        =
        &
        \lifelineElim_\ell(\mu_1) \opStrictSeq \lifelineElim_\ell(\mu_2)
        &
        ~~~~\text{by Prop.\ref{lem:elimination_append}}
        \end{array}
        \]
    \end{itemize}
\end{itemize}
\noindent For $\diamond = \opInterleaving$ we can reason similarly, using induction on both $\mu_1$ and $\mu_2$.
\end{proof}

\begin{theorem*}[Th.\ref{th:semantics_of_lifeline_removal_in_interactions}, Sec.\ref{sec:semantics}]
For any $L \subseteq \mathcal{L}$, any $H \subseteq L$ and any $i \in \mathbb{I}(L)$, we have:
\[
\sigma_{L\setminus H}(\lifelineElim_H(i)) = \lifelineElim_H(\sigma_L(i))
\]

\end{theorem*}

\begin{proof}
Let us reason by induction on the structure of interaction terms:\\
\noindent$\bullet$ 
$
\sigma_{L\setminus H}(\lifelineElim_H(\varnothing)) 
= 
\sigma_{L\setminus H}(\varnothing) 
= 
\{\varepsilon_{L\setminus H}\}
= 
\lifelineElim_H(\{\varepsilon_L\}) 
= 
\lifelineElim_H(\sigma_L(\varnothing))
$
\\
\noindent$\bullet$ for any $a \in \mathbb{A}(L)$ we have:
\[
\begin{array}{lrcl}
\text{if }\theta(a) \in H, 
&
\sigma_{L \setminus H}(\lifelineElim_H(a))
&
=
&
\sigma_{L \setminus H}(\varnothing)
=
\{ \varepsilon_{L \setminus H} \}
=
\lifelineElim_H( \{ \varepsilon_L \} )
\\
&
&
=
&
\lifelineElim_H( \{ a \multiAppend \varepsilon_L \} )
=
\lifelineElim_H(\sigma_L(a))
\\
\text{if }\theta(a) \not\in H, 
&
\sigma_{L \setminus H}(\lifelineElim_H(a))
&
=
&
\sigma_{L \setminus H}(a)
=
\{ a \multiAppend \varepsilon_{L \setminus H} \}
\\
&
&
=
&
\lifelineElim_H( \{ a \multiAppend \varepsilon_{L} \} )
=
\lifelineElim_H(\sigma_L(a))
\end{array}
\]
\noindent$\bullet$ with $(f,\diamond) \in \{(strict,\opStrictSeq),~(seq,\opStrictSeq),~(par,\opInterleaving),~(alt,\cup)\}$, for any $i_1$ and $i_2$:
\[
\begin{array}{lclr}
\sigma_{L \setminus H}(\lifelineElim_H(f(i_1,i_2)))
&
=
&
\sigma_{L \setminus H}(f(\lifelineElim_H(i_1),\lifelineElim_H(i_2)))
&
\text{ Def.\ref{def:lifeline_removal_operator_on_interactions}}
\\
&
=
&
\sigma_{L \setminus H}(\lifelineElim_H(i_1)) \diamond \sigma_{L \setminus H}(\lifelineElim_H(i_2))
&
\text{ Def.\ref{def:interaction_semantics}}
\\
&
=
&
\lifelineElim_H(\sigma_L(i_1)) \diamond \lifelineElim_H(\sigma_L(i_2))
&
\text{induction}
\\
&
=
&
\lifelineElim_H(~\sigma_L(i_1) \diamond \sigma_L(i_2)~)
&
\text{ Prop.\ref{prop:multi_trace_elimination_preserves_sched}}
\\
&
=
&
\lifelineElim_H(~\sigma_L(f(i_1,i_2))~)
&
\text{ Def.\ref{def:interaction_semantics}}
\end{array}
\]
\noindent$\bullet$ for any interaction $i$ and any $(k,\diamond) \in \{(S,\opStrictSeq),~(W,\opStrictSeq),~(P,\opInterleaving)\}$:
\[
\begin{array}{lclr}
\sigma_{L \setminus H}(\lifelineElim_H(loop_k(i)))
&
=
&
\sigma_{L \setminus H}(loop_k(\lifelineElim_\ell(i)))
&
\text{ Def.\ref{def:lifeline_removal_operator_on_interactions}}
\\
&
=
&
(~\sigma_{L \setminus H}(\lifelineElim_H(i))~)^{\diamond *}
&
\text{ Def.\ref{def:interaction_semantics}}
\\
&
=
&
(~\lifelineElim_H(\sigma_L(i))~)^{\diamond *}
&
~\text{ induction}
\\
&
=
&
\lifelineElim_H(~\sigma_L(i)^{\diamond *}~)
&
\text{ Prop.\ref{prop:multi_trace_elimination_preserves_sched}}
\\
&
=
&
\lifelineElim_H( \sigma_L(loop_k(i)) )
&
\text{ Def.\ref{def:interaction_semantics}}
\end{array}
\]
\end{proof}

\section{Operational semantics and lifeline removal\label{app:opsem_lfrem}}

Our structural operational semantics $\sigma : \mathbb{I}(L) \rightarrow \mathcal{P}(\mathbb{A}(L)^*)$ defined on global traces rely on the two predicates: 
\begin{itemize}
    \item the termination predicate $i \downarrow$ indicating $i$ accepts the empty trace 
    \item and the execution relation $i \xrightarrow{a@p} i'$ indicating that the action $a$ occurring at position $p$ in $i$ is executable from $i$ and that traces $a.t$ such that $t$ is accepted by $i'$ are accepted by $i$.
\end{itemize}

$\sigma : \mathbb{I}(L) \rightarrow \mathcal{P}(\mathbb{A}(L)^*)$ is defined by:

{
\centering
\begin{minipage}{4cm}
\begin{prooftree}
\AxiomC{$i \downarrow$}
\UnaryInfC{$\varepsilon \in \sigma(i)$}
\end{prooftree}
\end{minipage}
\begin{minipage}{4cm}
\begin{prooftree}
\AxiomC{$t \in \sigma(i')$}
\AxiomC{$i \xrightarrow{a@p} i'$}
\BinaryInfC{$a.t \in \sigma(i)$}
\end{prooftree}
\end{minipage}\\
}

\noindent
with $i$ and $i'$ interactions in $\mathbb{I}(L)$, $t$ a global trace, $a$ an action in $\mathbb{A}(L)$ and $p$ a position in $\{1,2\}^*$.

\textbf{Notation.}
In the following, we will generically consider an interaction $i$, a set $H$ of lifelines verifying $\emptyset \subsetneq H \subsetneq L$, an action $a$ occurring in $i$ at position $p$ and verifying $\theta(a) \not\in H$. Moreover, we will denote $i_h$ the interaction $\lifelineElim_{H}(i)$.

By Prop.\ref{prop:positions_and_lifeline_removal}, $i$ and $i_h$ have the same term structure except for some leaves corresponding to actions that are transformed to $\varnothing$. In particular, all operators of arity 1 or 2 ($alt$, $seq$, $strict$, $par$, $loop_k$) occur at the same positions in both terms.

\subsection{Termination}

\begin{definition}[Termination]\label{def:termination}
The predicate $\downarrow \subset \mathbb{I}(L)$ is such that for any $i_1$ and $i_2$ from $\mathbb{I}(L)$, any $f \in \{strict,seq,par\}$ and any $k \in \{S,W,P\}$ we have:

{
\centering
\begin{minipage}{1.60cm}
\begin{prooftree}
\AxiomC{\phantom{$\top$}}
\UnaryInfC{$\varnothing \downarrow$}
\end{prooftree}
\end{minipage}
\begin{minipage}{2.50cm}
\begin{prooftree}
\AxiomC{$i_1 \downarrow$}
\UnaryInfC{$alt(i_1,i_2) \downarrow$}
\end{prooftree}
\end{minipage}
\begin{minipage}{2.50cm}
\begin{prooftree}
\AxiomC{$i_2 \downarrow$}
\UnaryInfC{$alt(i_1,i_2) \downarrow$}
\end{prooftree}
\end{minipage}
\begin{minipage}{2.50cm}
\begin{prooftree}
\AxiomC{$i_1 \downarrow$}
\AxiomC{$i_2 \downarrow$}
\BinaryInfC{$f(i_1,i_2) \downarrow$}
\end{prooftree}
\end{minipage}
\begin{minipage}{2.50cm}
\begin{prooftree}
\AxiomC{\phantom{$\top$}}
\UnaryInfC{$loop_k(i_1) \downarrow$}
\end{prooftree}
\end{minipage}\\
}
\end{definition}

For any $i \in \mathbb{I}(L)$, $i \downarrow$ iff $ \varepsilon$ is an accepted global trace of $i$.

\begin{property}
\label{prop:termination_and_lifeline_removal}
$i \downarrow \Rightarrow i_h \downarrow$ 
\end{property}

\begin{proof}
By induction of the term structure of $i$.
If $i = \varnothing$ then $i_h = \varnothing$ and the property holds.
We cannot have $i = a \in \mathbb{A}(L)$ otherwise the premise $i \downarrow$ does not hold.
For any interactions $i_1$ and $i_2$, if $i$ is of the form $f(i_1,i_2)$ with $f \in \{strict,seq,par\}$ a scheduling operator then, by Def.\ref{def:termination}, $i \downarrow$ iff both $i_1 \downarrow$ and $i_2 \downarrow$. 
By induction, we have $i_1 \downarrow\;\Rightarrow i_{1h} \downarrow$ and $i_2 \downarrow\;\Rightarrow i_{2h} \downarrow$. Hence both $i_{1h} \downarrow$ and $i_{2h} \downarrow$.
On the other hand, by Def.\ref{def:lifeline_removal_operator_on_interactions}, $i_h = f(i_{1h},i_{2h})$. Hence, by Def.\ref{def:termination}, we have $i_h \downarrow$.
If $i = alt(i_1,i_2)$ we reason similarly depending on whether we have $i_1 \downarrow$ or $i_2 \downarrow$.
If $i = loop_k(i_1)$ with $k \in \{S,W,P\}$, we always have $i \downarrow$ as well as $i_h \downarrow$ because by Def.\ref{def:lifeline_removal_operator_on_interactions}, $i_h = loop_k(i_{1h})$.
\end{proof}

\subsection{Collision and Pruning}

The execution relation $\rightarrow$ relies on the definition of two other inductive predicates: collision and pruning.
The collision predicate $\collidesLf$ is such that $i \collidesLf l$ signifies that all traces of $i$ have at least an action occurring on the lifeline $l$. 

\begin{definition}[Collision]\label{def:collision}
The predicate $\collidesLf \subset \mathbb{I}(L) \times L$ is such that for $i_1$ and $i_2$ in $\mathbb{I}(L)$, $l\in L$, $a \in \mathbb{A}(L)$ and $f \in \{strict,seq,par\}$ we have:

{
\centering
\begin{minipage}{2cm}
\begin{prooftree}
\AxiomC{$\theta(a) = l$}
\UnaryInfC{$a \collidesLf l$}
\end{prooftree}
\end{minipage}
\vspace*{.1cm}
\begin{minipage}{3.5cm}
\begin{prooftree}
\AxiomC{$i_1 \collidesLf l$}
\AxiomC{$i_2 \collidesLf l$}
\BinaryInfC{$alt(i_1,i_2) \collidesLf l$}
\end{prooftree}
\end{minipage}
\begin{minipage}{2.75cm}
\begin{prooftree}
\AxiomC{$i_1 \collidesLf l$}
\UnaryInfC{$f(i_1,i_2) \collidesLf l$}
\end{prooftree}
\end{minipage}
\begin{minipage}{2.75cm}
\begin{prooftree}
\AxiomC{$i_2 \collidesLf l$}
\UnaryInfC{$f(i_1,i_2) \collidesLf l$}
\end{prooftree}
\end{minipage}
}
\end{definition}

\begin{property}
\label{prop:collision_and_lifeline_removal_1}
For all $l \in L \setminus H$ and for all $i \in \mathbb{I}(L)$ we have:
$(i \collidesLf l) \Rightarrow (i_h \collidesLf l)$
\end{property}

\begin{proof}
By induction of the term structure of $i$.
If $i = \varnothing$ then the premise does not hold.
If $i = a \in \mathbb{A}(L)$ then in order to have $i \collidesLf l$ we must have $\theta(a) = l$ which implies that $i_h = i$ because $l \not\in H$.

Given $i_1$ and $i_2$ two interactions, if $i=f(i_1,i_2)$ with $f \in \{strict,seq,par\}$ a scheduling operator then $i \collidesLf l$ iff either (or both of) $i_1 \collidesLf l$ or $i_2 \collidesLf l$. Let us suppose the former, all other things being equivalent.
By induction, the fact that $i_1 \collidesLf l$ implies that $i_{1h} \collidesLf$. On the other hand, $i_h = f(i_{1h},i_{2h})$ by Def.\ref{def:lifeline_removal_operator_on_interactions}.
Hence, by Def.\ref{def:collision}, we also have $i_h \collidesLf l$.
For $i=alt(i_1,i_2)$, we reason similarly (except that it is a logic AND instead of OR).
If $i = loop_k(i_1)$, with $k \in \{S,W,P\}$, the premise does not hold.
\end{proof}

\begin{property}
\label{prop:collision_and_lifeline_removal_2}
For all $l \in L \setminus H$ and for all $i \in \mathbb{I}(L)$ we have:
$(i_h \collidesLf l) \Rightarrow (i \collidesLf l)$
\end{property}

\begin{proof}
By induction of the term structure of $i$.
If $i = \varnothing$ then $i_h = \varnothing$ the premise does not hold.
If $i = a \in \mathbb{A}(L)$ then:
\begin{itemize}
    \item if $\theta(a) \in H$, we have $i_h = \varnothing$ and the premise does not hold
    \item otherwise, we have $i_h = i$ and the conclusion trivially holds
\end{itemize}

Given $i_1$ and $i_2$ two interactions, if $i=f(i_1,i_2)$ with $f \in \{strict,seq,par\}$ a scheduling operator then $i_h = f(i_{1h},i_{2h})$ and $i_h \collidesLf l$ iff either (or both of) $i_{1h} \collidesLf l$ or $i_{2h} \collidesLf l$. Let us suppose the former, all other things being equivalent.
By induction, the fact that $i_{1h} \collidesLf l$ implies that $i_1 \collidesLf$ and therefore that $i \collidesLf l$.
For $i=alt(i_1,i_2)$, we reason similarly (except that it is a logic AND instead of OR).
If $i = loop_k(i_1)$, with $k \in \{S,W,P\}$, we have $i_h = loop_k(i_{1h})$ and the premise does not hold.
\end{proof}

Pruning $i \isPruneOf{l} i'$ is defined so that $i'$ characterizes the maximum subset of accepted traces by $i$ that contains no trace with actions on $l$.

\begin{definition}[Pruning]\label{def:pruning_relation}
The pruning relation $\isPruneBase \; \subset \mathbb{I}(L) \times L \times \mathbb{I}(L)$ is s.t. for any $l \in L$, any $f \in \{strict,seq,par\}$ and any $k \in \{S,W,P\}$:

{
\centering
\begin{minipage}{2.5cm}
\begin{prooftree}
\AxiomC{\vphantom{$\isPruneOf{l}$}}
\UnaryInfC{$\varnothing \isPruneOf{l} \varnothing$}
\end{prooftree}
\end{minipage}
\begin{minipage}{2.5cm}
\begin{prooftree}
\AxiomC{\vphantom{$\theta(a) \neq l$} \vphantom{$\isPruneOf{l}$}}
\RightLabel{$\theta(a) \neq l$}
\UnaryInfC{$a \isPruneOf{l} a$}
\end{prooftree}
\end{minipage}
\begin{minipage}{4cm}
\begin{prooftree}
\AxiomC{$i_1 \isPruneOf{l} i_1'$}
\AxiomC{$i_2 \isPruneOf{l} i_2'$}
\BinaryInfC{$f(i_1,i_2) \isPruneOf{l} f(i_1',i_2')$}
\end{prooftree}
\end{minipage}

\vspace*{.1cm}

\begin{minipage}{3.85cm}
\begin{prooftree}
\AxiomC{$i_1 \isPruneOf{l} i_1'$}
\AxiomC{$i_2 \isPruneOf{l} i_2'$}
\BinaryInfC{$alt(i_1,i_2) \isPruneOf{l} alt(i_1',i_2')$}
\end{prooftree}
\end{minipage}
\begin{minipage}{3.85cm}
\begin{prooftree}
\AxiomC{$i_1 \isPruneOf{l} i_1'$}
\RightLabel{$i_2 \collidesLf l$}
\UnaryInfC{$alt(i_1,i_2) \isPruneOf{l} i_1'$}
\end{prooftree}
\end{minipage}
\begin{minipage}{3.85cm}
\begin{prooftree}
\AxiomC{$i_2 \isPruneOf{l} i_2'$}
\RightLabel{$i_1 \collidesLf l$}
\UnaryInfC{$alt(i_1,i_2) \isPruneOf{l} i_2'$}
\end{prooftree}
\end{minipage}

\vspace*{.1cm}

\begin{minipage}{5cm}
\begin{prooftree}
\AxiomC{$i_1 \isPruneOf{l} i_1'$}
\UnaryInfC{$loop_k(i_1) \isPruneOf{l} loop_k(i_1')$}
\end{prooftree}
\end{minipage}
\begin{minipage}{5cm}
\begin{prooftree}
\AxiomC{$\phantom{\isPruneOf{l}}$}
\RightLabel{$i_1 \collidesLf l$}
\UnaryInfC{$loop_k(i_1) \isPruneOf{l} \varnothing$}
\end{prooftree}
\end{minipage}
\\
}
\end{definition}

\begin{property}
\label{prop:pruning_and_lifeline_removal}
Let $l \not\in H$.
$i \isPruneOf{l}  i' \Rightarrow i_h \isPruneOf{l}  i'_h$. 
\end{property}

\begin{proof}
By induction of the term structure of $i$. 
If $i = \varnothing$ we trivially have $i_h = i = \varnothing$ and $\varnothing \isPruneOf{l} \varnothing$ always holds.

If $i = a \in \mathbb{I}(L)$, the premise can hold iff $\theta(a) \neq l$ and we then have $i \isPruneOf{l} i$. Then:
\begin{itemize}
    \item if $\theta(a) \in H$, then $i_h = \varnothing$, $i'_h = \varnothing$ and the property trivially holds
    \item if $\theta(a) \not\in H$, then $i_h = i$ and $i'_h = i'$ and the property trivially holds
\end{itemize}

Given two interactions $i_1$ and $i_2$, if $i$ is of the form $f(i_1,i_2)$ with $f \in \{strict,seq,par\}$ a scheduling operator, then the premise holds iff there exist a $i'_1$ and $i'_2$ s.t.~$i_1 \isPruneOf{l}  i'_1$, $i_2 \isPruneOf{l}  i'_2$ and $i' = f(i'_1,i'_2)$.
By induction, we have $i_{1h} \isPruneOf{l}  i'_{1h}$ and $i_{2h} \isPruneOf{l}  i'_{2h}$ which then implies as per Def.\ref{def:pruning_relation} and Def.\ref{def:lifeline_removal_operator_on_interactions}, that $i_h = f(i_{1h},i_{2h}) \isPruneOf{l} f(i'_{1h},i'_{2h}) = i'_h$.

A trivial property in \cite{denotational_and_operational_semantics_for_interaction_languages_application_to_trace_analysis} states that we have $i \isPruneOf{l} i'$ iff $\neg(i \collidesLf l)$.
Hence, for the case of $i=alt(i_1,i_2)$ we can distinguish between four cases:
\begin{itemize}
    \item either both $i_1 \collidesLf l$ and $i_2 \collidesLf l$ which contradicts the premise
    \item or both $\neg(i_1 \collidesLf l)$ and $\neg(i_2 \collidesLf l)$. 
    This then implies that we have $i'_1$ and $i'_2$ s.t.~$i_1 \isPruneOf{l} i'_1$, $i_2 \isPruneOf{l} i'_2$ 
    and $i \isPruneOf{l} alt(i'_1,i'_2)$. By induction, we have $i_{1h} \isPruneOf{l} i'_{1h}$ 
    and $i_{2h} \isPruneOf{l} i'_{2h}$, which implies via Def.\ref{def:pruning_relation} 
    and Def.\ref{def:lifeline_removal_operator_on_interactions}, that $i_h = alt(i_{1h},i_{2h}) \isPruneOf{l} alt(i'_{1h},i'_{2h}) = i'_h$
    \item or $i_1 \collidesLf l$ and $\neg(i_2 \collidesLf l)$. This implies that $i_2 \isPruneOf{l} i'_2$ and $i = alt(i_1,i_2) \isPruneOf{l} i'_2$. By induction, we have $i_{2h} \isPruneOf{l} i'_{2h}$.
    By Prop.\ref{prop:collision_and_lifeline_removal_1} and Prop.\ref{prop:collision_and_lifeline_removal_2}, we have $i_{1h} \collidesLf l$ and $\neg(i_{2h} \collidesLf l)$ hence, by Def.\ref{def:pruning_relation} and Def.\ref{def:lifeline_removal_operator_on_interactions}, $i_h = alt(i_{1h},i_{2h}) \isPruneOf{l} i'_{2h} = i'_h$
    \item the case $\neg(i_1 \collidesLf l)$ and $i_2 \collidesLf l$ is symmetric.
\end{itemize}

For the case $i = loop_k(i_1)$, we have two cases:
\begin{itemize}
    \item if $i_1 \collidesLf l$ then $loop_k(i_1) \isPruneOf{l} \varnothing$. By Prop.\ref{prop:collision_and_lifeline_removal_1}, we have $i_{1h} \collidesLf l$ which implies that $loop_k(i_{1h}) \isPruneOf{l} \varnothing$. Then, by Def.\ref{def:lifeline_removal_operator_on_interactions}, we have $i_h = loop_k(i_{1h})$ and $\varnothing_h = \varnothing$ hence the property holds
    \item or $\neg(i_1 \collidesLf l)$ which implies that there exists $i'_1$ s.t.~$i_1 \isPruneOf{l} i'_1$ and we have $i \isPruneOf{l} loop_k(i'_1)$. By induction, $i_1 \isPruneOf{l} i'_1$ implies $i_{1h} \isPruneOf{l} i'_{1h}$ which in turn implies $loop_k(i_{1h}) \isPruneOf{l} loop_k(i'_{1h})$. Then, by Def.\ref{def:lifeline_removal_operator_on_interactions}, $i_h = loop_k(i_{1h})$, $i'_h = loop_k(i'_{1h})$ and the property holds.
\end{itemize}
\end{proof}

\subsection{Execution}

The execution relation (Def.\ref{def:execution_relation} below) relates an initial term $i$ with a follow-up term $i'$ that results from the execution of an action $a$ at a specific position $p$ in $i$. 

\begin{definition}
\label{def:execution_relation}
Given $i_1$, $i_2$, $i_1'$ and $i_2'$ interactions in $\mathbb{I}(L)$, $a$ an action in $\mathbb{A}(L)$ and $p$ a position in $\{1,2\}^*$, the execution relation $\rightarrow \subset \mathbb{I}(L) \times \{1,2\}^* \times \mathbb{A}(L)  \times \mathbb{I}(L)$ is s.t.:

{
\centering
\begin{minipage}{3cm}
\begin{prooftree}
\AxiomC{\phantom{$\xrightarrow{a@p}$}}
\UnaryInfC{$a \xrightarrow{a@\varepsilon} \varnothing$}
\end{prooftree}
\end{minipage}
\begin{minipage}{3.5cm}
\begin{prooftree}
\AxiomC{$i_1 \xrightarrow{a@p} i'_1$}
\UnaryInfC{$alt(i_1,i_2) \xrightarrow{a@1p} i'_1$}
\end{prooftree}
\end{minipage}
\begin{minipage}{3.5cm}
\begin{prooftree}
\AxiomC{$i_2 \xrightarrow{a@p} i'_2$}
\UnaryInfC{$alt(i_1,i_2) \xrightarrow{a@2p} i'_2$}
\end{prooftree}
\end{minipage}

\vspace{0.1cm}

\begin{minipage}{5.5cm}
\begin{prooftree}
\AxiomC{$i_1 \xrightarrow{a@p} i'_1$}
\UnaryInfC{$par(i_1,i_2) \xrightarrow{a@1p} par(i'_1,i_2)$}
\end{prooftree}
\end{minipage}
\begin{minipage}{5.5cm}
\begin{prooftree}
\AxiomC{$i_2 \xrightarrow{a@p} i'_2$}
\UnaryInfC{$par(i_1,i_2) \xrightarrow{a@2p} par(i_1,i'_2)$}
\end{prooftree}
\end{minipage}

\vspace{0.1cm}

\begin{minipage}{5.5cm}
\begin{prooftree}
\AxiomC{$i_1 \xrightarrow{a@p} i'_1$}
\UnaryInfC{$strict(i_1,i_2) \xrightarrow{a@1p} strict(i'_1,i_2)$}
\end{prooftree}
\end{minipage}
\begin{minipage}{5.5cm}
\begin{prooftree}
\AxiomC{$i_2 \xrightarrow{a@p} i'_2$}
\RightLabel{$i_1 \downarrow$}
\UnaryInfC{$strict(i_1,i_2) \xrightarrow{a@2p} i'_2$}
\end{prooftree}
\end{minipage}

\vspace{0.1cm}

\begin{minipage}{5.5cm}
\begin{prooftree}
\AxiomC{$i_1 \xrightarrow{a@p} i'_1$}
\UnaryInfC{$seq(i_1,i_2) \xrightarrow{a@1p} seq(i'_1,i_2)$}
\end{prooftree}
\end{minipage}
\begin{minipage}{5.5cm}
\begin{prooftree}
\AxiomC{$i_1 \isPruneOf{\theta(a)} i_1'$}
\AxiomC{$i_2 \xrightarrow{a@p} i_2'$}
\BinaryInfC{$seq(i_1,i_2) \xrightarrow{a@2p} seq(i_1',i_2')$}
\end{prooftree}
\end{minipage}

\vspace{0.1cm}

\begin{minipage}{8cm}
\begin{prooftree}
\AxiomC{$i_1 \xrightarrow{a@p} i_1'$}
\UnaryInfC{$loop_S(i_1) \xrightarrow{a@1p} strict(i_1',loop_S(i_1))$}
\end{prooftree}
\end{minipage}
 
\vspace{0.1cm}
 
\begin{minipage}{8cm}
\begin{prooftree}
\AxiomC{$i_1 \xrightarrow{a@p} i_1'$}
\AxiomC{$loop_W(i_1) \isPruneOf{\theta(a)} i'$}
\BinaryInfC{$loop_W(i_1) \xrightarrow{a@1p} seq(i',seq(i_1',loop_W(i_1)))$}
\end{prooftree}
\end{minipage}

\vspace{0.1cm}

\begin{minipage}{8cm}
\begin{prooftree}
\AxiomC{$i_1 \xrightarrow{a@p} i_1'$}
\UnaryInfC{$loop_P(i_1) \xrightarrow{a@1p} par(i_1',loop_P(i_1))$}
\end{prooftree}
\end{minipage}\\
}

\end{definition}

We prove below Prop.\ref{prop:execution_and_lifeline_removal} from Sec.\ref{sec:semantics} which relates the execution relation to the lifeline removal operator.

\begin{property*}[Prop.\ref{prop:execution_and_lifeline_removal}, Sec.\ref{sec:semantics}]
For any $L \subset \mathcal{L}$, any interaction $i \in \mathbb{I}(L)$, any set of lifelines $\emptyset \subsetneq H \subsetneq L$, any action $a \in \mathbb{A}(L)$ s.t.~$\theta(a) \not\in H$ and any position $p \in pos(i)$:
\[
(\exists~i' \in \mathbb{I}(L),~i \xrightarrow{a@p} i') ~\Rightarrow~ (i_h \xrightarrow{a@p} i'_h)
\]
\end{property*}

\begin{proof}
Let us suppose $i \xrightarrow{a@p} i'$ and reason by induction on the term structure of $i$.
We cannot have $i = \varnothing$ as the premise would not hold.
If $i \in \mathbb{A}(L)$ then $a = i$, $p = \varepsilon$ and $i' = \varnothing$. Then, because $\theta(a) \not\in H$, we have $i_h = a$ and $i'_h = \varnothing$ which satisfies $i_h \xrightarrow{a@\varepsilon} i'_h$.

Given $i_1$ and $i_2$ two interactions, if $i = f(i_1,i_2)$ with $f \in \{strict,seq,par,alt\}$ any binary operator, the fact that $i \xrightarrow{a@p}$ implies that $p$ is
\begin{itemize}
    \item either of the form $1.p_1$ and we have a $i_1'$ s.t., $i_1 \xrightarrow{a@p_1} i_1'$, and then, by induction, we have $i_{1h} \xrightarrow{a@p_1} i'_{1h}$
    \item or of the form $2.p_2$ and we have a $i_2'$ s.t., $i_2 \xrightarrow{a@p_2} i_2'$, and then, by induction, we have $i_{2h} \xrightarrow{a@p_2} i'_{2h}$
\end{itemize}

 Then:
\begin{itemize}
    \item if $i = alt(i_1,i_2)$, supposing the action is on the side of $i_1$ (the other case being symmetric), we have $i \xrightarrow{a@1.p_1} i_1'$. Then, because by induction we have $i_{1h} \xrightarrow{a@p_1} i'_{1h}$, by Def.\ref{def:execution_relation}, this implies $alt(i_{1h},i_{2h}) \xrightarrow{a@1.p_1} i'_{1h}$ and thus $i_h \xrightarrow{a@p} i'_h$
    \item if $i = strict(i_1,i_2)$:
    \begin{itemize}
        \item if the action comes from $i_1$ then we have $i \xrightarrow{a@1.p_1} strict(i_1',i_2)$ and via induction and Def.\ref{def:execution_relation}, $strict(i_{1h},i_{2h}) \xrightarrow{a@1.p_1} strict(i'_{1h},i_{2h})$ hence, via Def.\ref{def:lifeline_removal_operator_on_interactions} the property holds
        \item if the action comes from $i_2$ then we have $i \xrightarrow{a@2.p_2} i'_2$ and $i_1 \downarrow$. As per Prop.\ref{prop:termination_and_lifeline_removal}, we also have $i_{1h} \downarrow$. Hence, via induction and Def.\ref{def:execution_relation}, $strict(i_{1h},i_{2h}) \xrightarrow{a@2.p_2} i'_{2h}$ hence, via Def.\ref{def:lifeline_removal_operator_on_interactions} the property holds
    \end{itemize}
    \item if $i = par(i_1,i_2)$ we reason as in the first case of $strict$
    \item if $i = seq(i_1,i_2)$, let us consider the case where the action comes from $i_2$ (the other being similar to the first case of $strict$). Then we must have $i_1 \isPruneOf{\theta(a)} i_1'$ and $i_2 \xrightarrow{a@p_2} i_2'$.
    Via Prop.\ref{prop:pruning_and_lifeline_removal}, and given $\theta(a) \not\in H$, the former implies $i_{1h} \isPruneOf{\theta(a)} i'_{1h}$, and, via induction, the latter implies $i_{2h} \xrightarrow{a@p_2} i'_{2h}$. Hence, as per Def.\ref{def:execution_relation}, we have $seq(i_{1h},i_{2h}) \xrightarrow{a@2.p_2} seq(i'_{1h},i'_{2h})$ and thus, via Def.\ref{def:lifeline_removal_operator_on_interactions} the property holds.
\end{itemize}

In the case of loops, for a derivation $loop_k(i_1) \xrightarrow{a@p} x$ to exist (given $k \in \{S,W,P\}$), we must have $p$ of the form $1.p_1$ and there must exist a $i'_1$ s.t.~$i_1 \xrightarrow{a@p_1} i_1'$, which, by induction, implies $i_{1h} \xrightarrow{a@p_1} i'_{1h}$.
Then:
\begin{itemize}
    \item in the case $(k,f) \in \{(S,strict),(P,par)\}$, we have $i = loop_k(i_1) \xrightarrow{a@1.p_1} f(i_1',i)$. Because by induction we have $i_{1h} \xrightarrow{a@p_1} i'_{1h}$, we also have $loop_k(i_{1h}) \xrightarrow{a@1.p_1} f(i'_{1h},loop_k(i_{1h}))$ and because by Def.\ref{def:lifeline_removal_operator_on_interactions} we have $i_h = loop_k(i_{1h})$, the property holds.
    \item in the case $i = loop_W(i_1)$, there exists a $i'$ s.t.~$i \isPruneOf{\theta(a)} i'$ and we have $i \xrightarrow{a@1.p_1} seq(i',seq(i_1',i))$. Via Prop.\ref{prop:pruning_and_lifeline_removal}, and because $\theta(a) \not\in H$, we have $i_h \isPruneOf{\theta(a)} i'_h$. Hence, because $i_h = loop_W(i_{1h})$ (Def.\ref{def:lifeline_removal_operator_on_interactions}) and because by induction we have $i_{1h} \xrightarrow{a@p_1} i'_{1h}$, we also (Def.\ref{def:execution_relation}) have $loop_W(i_{1h}) \xrightarrow{a@1.p_1} seq(i'_h,seq(i'_{1h},i_h))$. Thus the property holds.
\end{itemize}

\end{proof}

\end{document}